\documentclass[draftcls,12pt,onecolumn]{IEEEtran}		
\usepackage{asymptote}
 \usepackage{tikz}
\usepackage{tikz-3dplot}

\usepackage[utf8]{inputenc}
\usepackage[T1]{fontenc}
\usepackage{url}
\usepackage{ifthen}
\usepackage{cite}
\usepackage[cmex10]{amsmath} % Use the [cmex10] option to ensure complicance

\usepackage{mypackage}
\usepackage{mathrsfs}

\newtheorem{remark}{Remark}
\usepackage{mathtools}

\DeclarePairedDelimiter\floor{\lfloor}{\rfloor}
\usepackage{float}

\newtheorem{proposition}{Proposition}
\newtheorem{corollary}{Corollary}
\newtheorem{theorem}{Theorem}
\usepackage[cmex10]{amsmath} % Use the [cmex10] option to ensure complicance
\usepackage{stackengine}
\usepackage{gensymb}
\usepackage{graphics}

\usepackage{dcolumn}
\newcolumntype{d}[1]{D{.}{\cdot}{#1} }

\usepackage{physics}
\usepackage{siunitx}

\newcommand{\sir}{\mathrm{SIR}}
\newcommand{\Pb}{\mathbb{P}}
\newcommand{\Eb}{\mathbb{E}}

\newcommand{\Lc}{\mathcal{L}}

\newcommand{\black}{\textcolor{black}}

\DeclareMathOperator*{\R}{\mathbb{R}}

\DeclareMathOperator*{\U}{\mathcal{U}}

\usepackage{multicol}

\usepackage[nomain,acronym,shortcuts]{glossaries}
\makeglossaries
\newcommand*{\acro}[3][]{\newacronym[#1]{#2}{#2}{#3}}

\acro{D2D}{device-to-device}
\acro{SIR}{signal-to-interference-ratio}
\acro{SINR}{signal-to-interference-plus-noise-ratio}
\acro{PCP}{Poisson cluster process}
\acro{CoMP}{coordinated multi-point}
\acro{BS}{base station} 
\acro{MD-CoMP}{macrodiversity CoMP transmission}
\acro{MAC}{medium-access-control}
\acro{JT-CoMP}{joint transmission CoMP}
\acro{CoMP-JT}{coordinated multipoint joint transmission}
\acro{SBS}{small base station}
\acro{MDSD}{multiple devices to the single device}
\acro{MDS}{maximum distance seperable}
\acro{SCN}{small cell network}
\acro{PPP}{Poisson point process}
\acro{TCP}{Thomas cluster process}
\acro{CSI}{channel state information}
\acro{PDF}{probability distribution function}
\acro{PMF}{probability mass function}
\acro{RV}{random variable}
\acro{i.i.d.}{independently and identically distributed}
\acro{MBMS}{multimedia broadcasting multicasting service}
\acro{EE}{energy efficiency}
\acro{HCP}{hard-core placement}
\acro{CCDF}{complementary cumulative distribution function}
\acro{CDF}{cumulative distribution function}
\acro{PC}{probabilistic caching}
\acro{RC}{random uniform caching}
\acro{CPF}{caching popular files} 
\acro{PGFL}{probability generating functional}
\acro{KKT}{Karush-Kuhn-Tucker}
\acro{PGF}{point generating function}
\acro{BPP}{binomial point process}
\acro{3GPP}{3rd Generation Partnership Project}
\acro{UAV}{unmanned aerial vehicle}
\acro{FAA}{Federal Aviation Administration}
\acro{UAS}{unmanned aircraft systems}
\acro{AP}{access point}
\acro{GSM}{Global System for Mobile Communications}
\acro{LTE}{long term evolution}
\acro{IoT}{Internet-of-things}
\acro{UE}{user equipment}
\acro{CNPC}{control and non-payload communication}
\acro{FHD}{full high definition}
\acro{NLoS}{non-line-of-sight}
\acro{LoS}{line-of-sight}
\acro{5G}{fifth-generation}
\acro{MGF}{moment generating function}
\acro{3D}{three-dimensions}
\acro{2D}{two-dimensions}
\acro{1D}{one-dimension}
\acro{MIMO}{multiple-input multiple-output}
\acro{CB}{conjugate beamforming}
\acro{AU}{aerial user}
\acro{GU}{ground user}
\acro{CLT}{central limit theorem}
\acro{SE}{spectral efficiency}
\acro{MRT}{maximum ratio transmission}
\acro{ZFBF}{zero-forcing beamforming}
\acro{QoS}{quality-of-service}
\acro{GBS}{ground base station}
\acro{F-RAN}{fog-radio access network}
\acro{RAN}{radio access network}
\acro{RRH}{remote radio head}
\acro{mmWave}{millimeter wave}
\acro{C&C}{command and control}
\acro{ASE}{area spectral efficiency}
\acro{RWP}{random waypoint}
\acro{AGL}{above ground level}
\acro{w.r.t.}{with respect to}
\acro{GTA}{ground-to-air}
\acro{user equipment}{UE}
\acro{AUE}{aerial user equipment}
\acro{GUE}{ground user equipment}
\acro{UB}{upper bound}
\acro{LB}{lower bound}
\acro{ABS}{aerial base station}
\acro{UAV-UE}{UAV-\ac{UE}}
\usepackage{amssymb}
\begin{document} 
\title{Mobility in the Sky: Performance  and Mobility Analysis for Cellular-Connected UAVs}
%\title{Sky-Aware Cellular Networks: Performance Analysis of Communication-Assisted CoMP for Cellular-Connected UAVs}

\author{\IEEEauthorblockN{Ramy Amer\IEEEauthorrefmark{1},			
Walid Saad\IEEEauthorrefmark{2},
%Hesham ElSawy\IEEEauthorrefmark{3},
%M. Majid Butt\IEEEauthorrefmark{1}\IEEEauthorrefmark{4},
Nicola Marchetti\IEEEauthorrefmark{1}}\\
\IEEEauthorblockA{\IEEEauthorrefmark{1}CONNECT, Trinity College, University of Dublin, Ireland}\\
\IEEEauthorblockA{\IEEEauthorrefmark{2}Wireless@VT, Bradley Department of Electrical and Computer Engineering, Virginia Tech, Blacksburg, VA, USA}\\
%\IEEEauthorblockA{\IEEEauthorrefmark{3}King Fahd University of Petroleum and Minerals (KFUPM)}\\
%\IEEEauthorblockA{\IEEEauthorrefmark{4}Nokia Bell Labs, France}\\
\IEEEauthorblockA{email:\{ramyr, nicola.marchetti\}@tcd.ie, walids@vt.edu}}

 \author{
 Ramy Amer,~\IEEEmembership{Student~Member,~IEEE,} Walid~Saad,~\IEEEmembership{Fellow,~IEEE,} 
 %Hesham~Elsawy,~\IEEEmembership{Senior~Member,~IEEE,}  
%M.~Majid~Butt,~\IEEEmembership{Senior~Member,~IEEE,} 
~and~Nicola~Marchetti,~\IEEEmembership{Senior~Member,~IEEE}
 \thanks{The material in this paper is presented in part to IEEE WCNC 2019 \cite{amer2020caching}.}
\thanks{Ramy Amer and and Nicola~Marchetti are with CONNECT Centre for Future Networks, Trinity College Dublin, Ireland. Emails:\{ramyr, nicola.marchetti\}@tcd.ie.}
\thanks{Walid Saad is with Wireless@VT, Bradley Department of Electrical and Computer Engineering, Virginia Tech, Blacksburg, VA, USA. Email: walids@vt.edu.}
%\thanks{Hesham ElSawy is with King Fahd University of Petroleum and Minerals (KFUPM), Saudi Arabia. Email: hesham.elsawy@kfupm.edu.sa.}
%\thanks{M. Majid Butt is with Nokia Bell Labs, France.}
\thanks{This publication has emanated from research conducted with the financial support of Science Foundation Ireland (SFI) and is co-funded under the European Regional Development Fund under Grant Number 13/RC/2077, and the U.S. National Science Foundation under Grants CNS-1836802 and IIS-1633363.}
}
\maketitle
%\red{0- revise simulation files 1- LoS prob from Mahdi paper 3- Theorems 2/3 (3D dist and insights)}
% A novel three-dimensional (3D) mobility model that effectively captures the spatial as well as vertical movements of mobile UAV-UEs is proposed in order to characterize their performance.
\begin{abstract}
%\red{Fig pages 11 and 23}
Providing connectivity to unmanned aerial vehicle-user equipments (UAV-UEs), such as drones or flying taxis, is a major challenge for tomorrow's cellular systems. In this paper, the use of coordinated multi-point (CoMP) transmission for providing seamless connectivity to UAV-UEs is investigated. In particular, a network of clustered ground base stations (BSs) that cooperatively serve a number of UAV-UEs is considered. Two scenarios are studied: scenarios with static, hovering UAV-UEs and scenarios with mobile UAV-UEs. Under a maximum ratio transmission, a novel framework is developed and leveraged to derive upper and lower bounds on the UAV-UE coverage probability for both scenarios. Using the derived results, the effects of various system parameters such as collaboration distance, UAV-UE altitude, and UAV-UE velocity on the achievable performance are studied. Results reveal that, for both static and mobile UAV-UEs, when the BS antennas are tilted downwards, the coverage probability of a high-altitude UAV-UE is upper bounded by that of ground users regardless of the transmission scheme. Moreover, for low signal-to-interference-ratio thresholds, it is shown that CoMP transmission can improve the coverage probability of UAV-UEs, e.g., from $28\%$ under the nearest association scheme to $60\%$ for a collaboration distance of \SI{200}{m}. Meanwhile, key results on mobile UAV-UEs unveil that not only the spatial displacements of UAV-UEs but also their vertical motions affect their handover rate and coverage probability. In particular, UAV-UEs that have frequent vertical movements and high direction switch rates are expected to have low handover probability and handover rate. Finally, the effect of the UAV-UE vertical movements on its coverage probability is marginal if the UAV-UE retains the same mean altitude. 
% yields the same coverage probability
% \magenta{For both scenarios, under maximum ratio transmission (MRT), novel upper and lower bounds on the UAV-UE coverage probability are derived.} 
%
%especially, in the case of dense ground networks and UAV-UEs with high direction switch rates. 
% of movement direction switch rates. we envisage the use
% movement direction switch rates are higher
%\red{Energy is more efficient in the ground than in the sky.}% UAV-UEs

%\red{We show that according to the collaboration distance and adopted beamforming scheme, the proposed model provides substantially reliable communication for high-altitude \acp{UAV-UE}, even though the\acp{BS}' antennas are down-tilted. We show that assuming horizontal-mobility only for the UAVs remarkably  underestimates the HO rate and Sojourn time.}
%%%%%%%%%%%%%%%%%%%%%%%%%%%%%%%%%%%%%%%%%%%%%%%%%%%%%%7510820%%%%%%%%%%%%%%%
\end{abstract}
\vspace{-0.6 cm}
\vspace{-0.6 cm}
\section{Introduction}	%\cite{mozaffari2016unmanned}
The past few years have witnessed a tremendous increase in the use of \acp{UAV}, popularly called drones, in many applications, such as aerial surveillance, package delivery, and even flying taxis \cite{8473483,saad2019vision,kishk2019capacity,kishk20193,eldosouky2019drones}. Enabling such \ac{UAV}-centric applications requires ubiquitous wireless connectivity that can be potentially provided by the pervasive wireless cellular network \cite{8660516} and \cite{8533634}. However, in order to operate  cellular-connected \acp{UAV} using existing wireless systems, one must address a broad range of challenges that include interference mitigation, reliable communications, resource allocation, and mobility support \cite{8470897}. Next, we review some of the works relevant to the cellular-connected \ac{UAV}-enabled networks.  
\vspace{-0.4 cm}
\subsection{State of the Art and Prior Works}
\vspace{-0.2 cm}
Recently, cellular-connected \acp{UAV} have received significant attention, whereby \acp{UAV} as new \acp{UE} are integrated into the cellular network in order to ensure reliable and secure connectivity for the operations of UAV systems. However, it has been established that the dominance of \ac{LoS} links makes inter-cell interference a critical issue for cellular systems with hybrid terrestrial and aerial \acp{UE}. In this regard, extensive real-world simulations and fields trials in \cite{8470897,qualcomm2017unmanned,lin2018sky,van2016lte} have shown that a UAV-UE, in general, has poorer performance than a \ac{GUE}. Due to the down-tilted \ac{BS} antennas, it is found that \acp{UAV} at \SI{40}{m} and higher, will be eventually served by the side-lobes of the \ac{BS} antennas, which have reduced antenna gain compared to the corresponding main-lobes. However, UAV-UEs at \SI{40}{m} and above experience favorable free-space propagation conditions. Interestingly, the work in \cite{lin2018sky} showed that the   free-space propagation can make up for the \ac{BS} antenna side-lobe gain reduction. However, this merit of such a favorable \ac{LoS} channel that UAV-UEs enjoy vanishes at high altitudes and turns to be one of their key limiting factors. This is because the free-space propagation also leads to stronger \ac{LoS} interfering signals. Eventually, UAV-UEs at high altitudes potentially exhibit poorer communication and coverage compared to \acp{GUE} \cite{8470897,qualcomm2017unmanned,lin2018sky,van2016lte,azari2017coexistence}. 

While the works in \cite{8470897,qualcomm2017unmanned,lin2018sky,van2016lte,azari2017coexistence}  explored the feasibility of providing cellular connectivity for \acp{UAV}, they did not envision new techniques to improve their performance. In particular, \acp{UAV}, at high altitudes, have limited coverage and connectivity  due to the encountered \ac{LoS} interference and reduced antenna gains. Moreover, their cell association will be essentially driven by the side-lobes of \ac{BS} antennas, which will lead to more handovers and handover failures for mobile UAV-UEs \cite{lin2018sky}. This necessitates the need to have sky-aware cellular networks that can seamlessly cover high altitudes UAV-UEs  and support their inevitable mobility. Next, we review some recently-adopted techniques that aimed to provide reliable connectivity to the UAV-UEs.

%ensuring  reliable cellular connectivity to high altitude UAV-UEs requires overcoming the limiting factors such as \ac{LoS}-dominated interference, reduced antenna gains, and mobility support.

%In particular, \ac{UAV} especially at high altitudes have a limited coverage and  
%due to the encountered \ac{LoS} interference and reduced antenna gains. 
%
% that is limited by \ac{LoS} interference and reduced antenna gains. However, as the \ac{UAV} altitude further increases, new solutions are needed to enable cellular \acp{BS} to seamlessly cover the sky. In particular, ensuring  reliable cellular connectivity to UAV-UEs requires overcoming their limiting factors that include, \red{\ac{LoS}-dominated interference}, reduced antenna gains, and mobility support. Particularly, as the \ac{UAV}-\ac{UE} altitude increases, the cell association is essentially driven by the side-lobes of BS antennas and mobility it may result in more handovers and possibly more handover failures \cite{lin2018sky}. 
%  For instance, using system-level simulations, the work in \cite{8528463} showed that the problem of UAV cell selection is essentially driven by the side-lobes of \ac{BS} radiation pattern. 

% 8528463,

Recently, various approaches have been proposed in \cite{8756296,d2019cell,rahmati2019energy,cherif2019downlink,liu2018multi} in order to improve the cellular connectivity for UAVs using, e.g., massive \ac{MIMO}, \ac{mmWave}, and beamforming. For instance, in \cite{8756296}, we proposed a MIMO conjugate beamforming scheme that can improve the cellular connectivity for UAV-UEs and enhance the system spectral efficiency. Moreover, the authors in \cite{cherif2019downlink} incorporated directional beamforming at  aerial \acp{BS} to alleviate the strong \ac{LoS} interference seen by their served UAV-UEs. However, while interesting, the works in \cite{8756296,d2019cell,rahmati2019energy,cherif2019downlink,liu2018multi} only considered scenarios of static UAV-UEs. Moreover, they did not consider the use of \ac{CoMP} transmission for UAV-UEs, which is a prominent interference mitigation tool that can diminish the effect of \ac{LoS} interference. 

% the use of \ac{CoMP} transmission for UAV-UEs has not been investigated yet in the literature.
%However, none of these prior works studied the cooperative transmission through \ac{CoMP} and
%underlaid D2D communications.

Unlike the static \ac{UAV} assumptions in \cite{azari2017coexistence,8756296,d2019cell,rahmati2019energy,cherif2019downlink,liu2018multi}, the study of mobile \acp{UAV} has been conducted in  \cite{85317111,zhang2019trajectory,8654727,8421028,8671460,8681266,8692749,abs-1804-04523,HCC:3325421.3329770}.  Prior works in the literature followed two main directions pertaining to trajectory design for mobile \acp{UAV}. The first line of work focuses on deterministic trajectories, whereby a \ac{UAV} is assumed to travel between two deterministic, possibly known, locations \cite{85317111,zhang2019trajectory,8654727}. This type of trajectories can be used   for path planning and mission-related metrics' optimization, e.g., mission time and achievable rates. For instance, the authors in \cite{85317111} studied the problem of trajectory optimization for a cellular-connected \ac{UAV} flying from an initial location to a final destination. 
% Meanwhile, in \cite{zhang2019trajectory}, the authors studied the trajectory design for a flying cellular-connected UAV subject to a minimum \ac{SINR} requirement. 
Moreover, the work in \cite{8654727} proposed an interference-aware path planning scheme for a network of cellular-connected \acp{UAV} based on deep reinforcement learning.

The second line of work in \cite{8421028,8671460,8681266} considers stochastic trajectories in which the movements of \acp{UAV} are characterized by means of stochastic processes. This type of trajectories is usually adopted in the study of communication and mobility-related metrics such as coverage probability and handover rate.
For example, in \cite{8421028}, the authors proposed a mixed random mobility model that characterizes the movement process of \acp{UAV} in a finite \ac{3D} cylindrical region. The authors characterized the \ac{GUE} coverage probability in a network of one static serving aerial \ac{BS} and multiple mobile interfering aerial \acp{BS}. The authors extended their work in \cite{8671460} such that both serving and interfering aerial \acp{BS} are mobile. Meanwhile, the authors in \cite{8681266} showed that an acceptable \ac{GUE} coverage can be sustained if the aerial \acp{BS} move according to certain stochastic trajectory processes. However, while interesting, these mobility models can only describe the motions of aerial \acp{BS} deployed in a bounded cylindrical region in space. In contrast, cellular-connected UAV-UEs such as flying taxis and delivery drones would have very long trajectories that cross multiple areas served by different \acp{BS}. 
%
%To explore the role of mobility in cellular networks, particularly handover rate and sojourn time,  mobility modeling is obviously a necessary first step.
%
%In this paper, we are particularly interested in the second line of work which incorporates a realistic model for the underlying physical layer.
% Using terrestrial cellular networks to provide connectivity to the mobile UAV-UEs may face new challenges.

Ensuring reliable connectivity for such mobile UAV-UEs is of paramount importance for the control and operations of UAV systems. In this regard, the mobility performance of cellular-connected UAVs has been studied in recent works  \cite{8692749,abs-1804-04523,HCC:3325421.3329770}. In \cite{8692749}, the authors quantified the impact
of handover on the UAV-UE throughput, assuming that no payload data is received during the handover procedures. Meanwhile, in \cite{abs-1804-04523}, based on system-level simulations, it is revealed that high handover rate is encountered when the UAV-UE moves through the nulls between side-lobes of the \ac{BS} antennas. Moreover, based on experimental trials, in \cite{HCC:3325421.3329770}, the authors showed that under the strongest received power association, drones are subject to frequent handovers once the typical flying altitude is reached. However, the results presented in these works are based on simulations and measurements. 
% evaluate performance metrics such as coverage probability and handover rate is still lacking in the current literature. We note that, 

While there exist some approaches in the literature to improve the cellular connectivity for UAV-UEs \cite{8756296,d2019cell,rahmati2019energy,cherif2019downlink,liu2018multi}, none of these works studied the role of \ac{CoMP} transmission as an effective interference mitigation tool to support the UAV-UEs. Moreover, these works only considered scenarios of static UAV-UEs. Furthermore, while the authors in \cite{8692749,abs-1804-04523,HCC:3325421.3329770} studied the performance of mobile UAV-UEs, their results were based on system simulations and measurement trials. Particularly, a rigorous analysis for mobile UAV-UEs to quantify important performance metrics such as coverage probability and handover rate is still lacking in the current state-of-the-art. \emph{To our best knowledge, this paper provides the first rigorous analysis of CoMP transmission for both static and \ac{3D} mobile UAV-UEs, where a novel \ac{3D} mobility model is also provided.} 

\vspace{-0.5 cm}
\subsection{Contributions}			% envisage the use of 
\vspace{-0.2 cm}
The main contribution of this paper is a novel framework that leverages \ac{CoMP} transmissions for serving cellular-connected \acp{UAV}, and develops a novel mobility model that effectively captures the \ac{3D} movements of UAV-UEs. We propose a \ac{MRT} scheme aiming to maximize the desired signal at the intended UAV-UE, and, hence, the performance of cellular communication links for the UAV-UEs can be improved. In particular, we consider a network of disjoint clusters in which \acp{BS} in one cluster collaboratively serve one UAV-UE within the same cluster via coherent \ac{CoMP} transmission. For this network, we consider two key scenarios, namely, static and mobile UAV-UEs. Using Cauchy's inequality and Gamma approximations, we develop a novel framework that is then leveraged to derive tight \ac{UB} and \ac{LB} on the UAV-UE coverage probability for both scenarios. Moreover, for mobile UAV-UEs, we analytically characterize the handover rate, and handover probability based on a novel \ac{3D} mobility model. We further quantify the negative impact of the UAV-UEs' mobility on their achievable performance. %\ac{3D}  mobile UAV-UEs.
% Two transmission schemes are adopted, particularly, \ac{MRT} and transmit diversity.
%Overall, this paper provides a powerful toolbox for the evaluation and design of various multi-antenna wireless networks, which shall find ample applications.
%  we envisage a novel mobility model that is then used to characterize the performance of \ac{3D} mobile UAV-UEs and  both static and mobile scenarios,

Our results reveal that the achievable performance of UAV-UEs heavily depends on the UAV-UE altitude, UAV-UE velocity, and the collaboration distance, i.e., the distance within which the UAV-UE is cooperatively served from ground \acp{BS}. Moreover, while allowing \ac{CoMP} transmission substantially improves the UAV-UEs' performance, it is shown that their performance is still upper bounded by that of \acp{GUE} due to the down-tilt of the \ac{BS} antennas and the encountered \ac{LoS} interference. Additionally, results on mobile UAV-UEs unveil that the spatial displacements of UAV-UEs jointly with their vertical motions affect their handover rate and handover probability. Moreover, while the UAV-UE spatial movements considerably impact its coverage probability (due to handover), the effect of the UAV-UE vertical displacements is marginal if the UAV-UE fluctuates around the same mean altitude. 
\emph{Overall, cooperative transmission is shown to be particularly effective for high altitude UAV-UEs that are   susceptible to adverse interference conditions, which is the case in a variety of drone applications.}

%  velocity negatively impacts its coverage probability
%
%  However, while the UAV-UE horizontal velocity negatively impacts its coverage probability, the effect of the UAV-UE vertical displacements, while keeping the same mean altitude, is quite marginal.
%
% the imminent \ac{IoT} and massive machine type communications era.
% can we have an insightful summary of contributions at the end
% Overall, this paper provides a powerful toolbox for the evaluation and design of various multi-antenna wireless networks, which shall find ample applications.

The rest of this paper is organized as follows. Section II and Section III present, respectively, the system model and the coverage probability analysis for static UAV-UEs. Section IV develops a novel \ac{3D} mobility model and  Section V studies the performance of \ac{3D} mobile UAV-UEs. Numerical results are presented in Section VI and conclusions are drawn in Section VII.

\begin{figure}[!t]
\vspace{-0.9 cm}	
    \centering
    \subfigure[Illustration of cooperative transmission]
    {
        \includegraphics[width=2.3in]{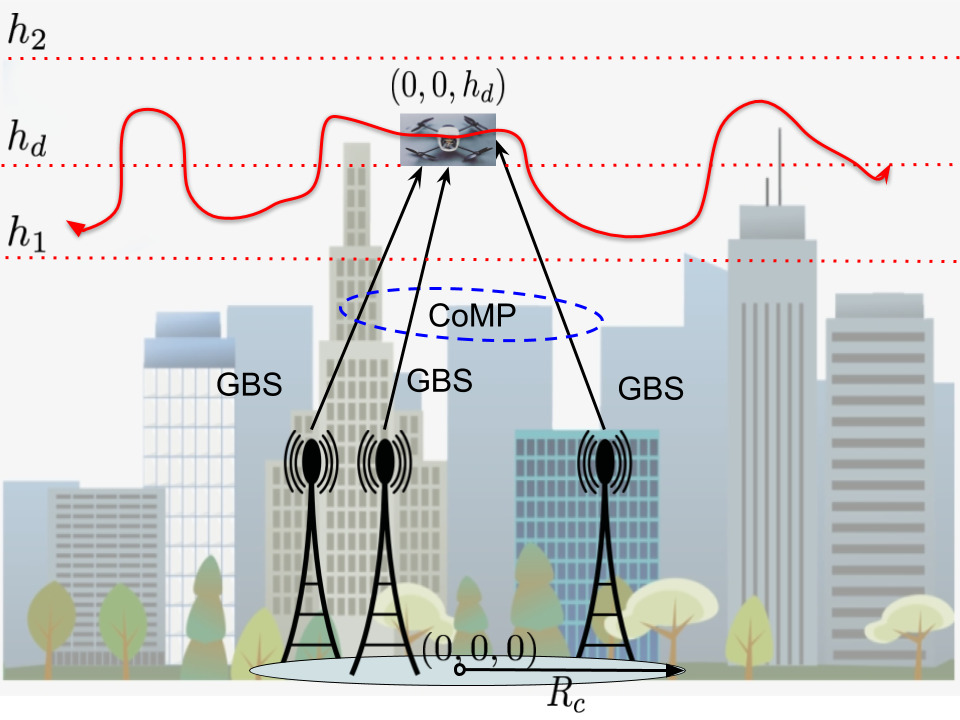}		%lower-bound.eps
        \label{cov_prob_vs_theta11}
    }
    \subfigure[Snapshot of a cluster-centric UAV-UE topology]	% \red{see slope} 	
     {\hspace*{0.1 in}
    {
        \includegraphics[width=2.5in]{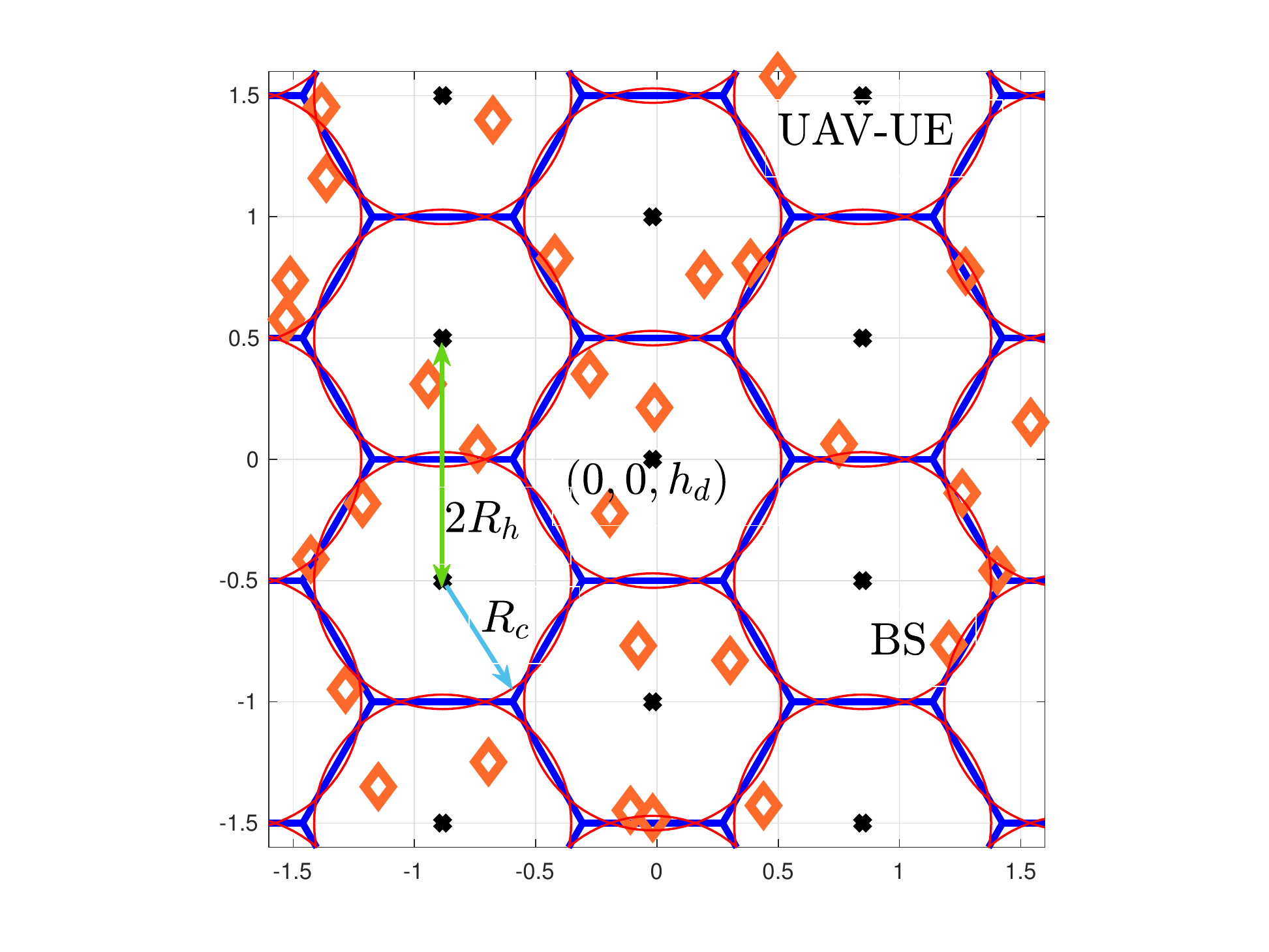}		%lower-bound.eps
        \label{cov_prob_vs_theta12}
    }}
    \caption{Illustration of the proposed system model where \acp{BS} cooperatively serve high-altitude UAV-UEs via CoMP transmission. UAV-UEs can be either hovering at a fixed altitude $h_d$ or flying within minimum and maximum altitudes $h_1$ and $h_2$, respectively. In (b), the clusters are defined by a hexagonal grid, wherein BSs (orange diamonds) are distributed according to a homogeneous PPP and the UAV-UEs (black stars) are hovering above the centers of disjoint clusters.}		
    \label{system-model-comp}
    \vspace{-0.7 cm}
\end{figure}

% Snapshot of the cluster-centric SCN topology. The clusters are defined
%by a hexagonal grid, wherein SBSs (red triangles) are distributed according
%to a homogeneous PPP. A cluster of interest is considered for performance
%analysis with cluster center at the origin.

% \begin{figure} [!t]	%[!t] %%    [htbp]
% \vspace{-0.6 cm}
%\centering
%\includegraphics[width=0.45\textwidth]{Figures/ch5/uav-comp9}		%letter			
%\caption {Illustration of the proposed system model where \acp{BS} cooperatively serve high-altitude UAV-UEs via CoMP transmission.}
%\label{system-model-comp}
%\vspace{-0.6 cm}
%\end{figure}
%\vspace{-0.6 cm}
% \red{Scheduling}
\vspace{-0.5 cm}
\section{System Model}	
\vspace{-0.2 cm}
\label{system-model}
We consider a terrestrial cellular network in which \acp{BS} are distributed according to a \ac{2D} homogeneous \ac{PPP}  $\Phi_b=\{ b_i \in \mathbb{R}^2, \forall i \in \mathbb{N}^+ \}$ with intensity $\lambda_b$. All \acp{BS} have the same transmit power $P_t$, and are deployed at the same height $h_{\rm BS}$. We consider a number of high altitude cellular-connected UAV-UEs that can be either static or mobile based on the application. Particularly, the UAV-UEs are either hovering or moving at altitudes higher than the BS heights. We consider a cluster-centric UAV-UE model in which \acp{BS} are grouped into disjoint clusters modeled using a hexagonal grid with an inter-cluster center distance equal to $2R_h$, see Fig.~\ref{system-model-comp}. The area of each cluster is hence given by $A=2\sqrt{3}R_h^2$. For analytical convenience, we approximate the cluster area to a circle with the same area, i.e., with collaboration distance $R_c$ where $\pi R_c^2=2\sqrt{3}R_h^2$, and $R_c = \sqrt{\frac{2\sqrt{3}}{\pi}}R_h$. \emph{\acp{BS} belonging to the same cluster can cooperate
to serve one UAV-UE within their cluster to mitigate the effect of \ac{LoS} interference and, hence, enhance the UAV-UE cellular connectivity.}

%we go back to the same issue, the previous paragraph says that you have a number of UAV-UEs, and now youi say "the UAV-UE" which means there is only one. DO you mean "a UAV-UE"? Do you want to say Hereianfter, we consider an arbitrary UAV-UE to analyze or something? You do that in Section III, so Section III is fine because you already say we refer to it as the typical UAV-UE, but in Section II it is still confusing

%\subsection{\red{\textbf{Angles and gain of main-lobes and side-lobes}}}
\vspace{-0.3 cm}
\subsection{Channel Model}
\vspace{-0.1 cm}
We consider a wireless channel that is characterized by both large-scale and small-scale fading. For the large-scale fading, the channel between \ac{BS} $i$ and an arbitrary UAV-UE is described by the \ac{LoS} and \ac{NLoS} components, which are considered separately along with their probabilities of occurrence \cite{ding2016performance}. This assumption is apropos for such \ac{GTA} channels that  often exhibit \ac{LoS} communication (e.g., see \cite{azari2017coexistence} and \cite{8713514}). 

For small-scale fading, we adopt a Nakagami-$m_v$ model as done in \cite{azari2017coexistence} for the channel gain, whose \ac{PDF} is given by:
%\begin{align}
$f(\omega) = \frac{2m_v^{m_v} \omega^{2m_v-1}}{\Gamma(m_v)} e^{-m_v \omega^2}$,  
%$\end{align}
where $m_v$, $v\in \{l,n \}$, is the fading parameter which is assumed to be an integer for analytical tractability, with $m_l>m_n$. In the special case when $m_n = 1$, Rayleigh fading is recovered with an exponentially distributed instantaneous power, which can be used for the performance evaluation of ground users. % Since the communication links between the UAV-UEs and \acp{BS}  are \ac{LoS}-dominated, e.g., suburban environments with $h_d>\SI{40}{m}$ \cite{lin2018sky}, it is assumed to have $m_l>1$.% 
Given that $\omega \sim $ Nakagami$(m_v)$, it directly follows that the channel gain power  $\chi=\omega^2 \sim \Gamma(m_v,1/m_v)$, where $\Gamma(K,\theta)$ is a Gamma \ac{RV} with $K$ and $\theta$  denoting the shape and scale parameters, respectively. Hence, the \ac{PDF} of channel power gain distribution will be:  
%\begin{align}
$f(\chi) = \frac{m_v^{m_v} \chi^{m-1}}{\Gamma(m_v)} {\rm exp}\big(-m_v\chi\big)$.  
%\end{align}

\ac{3D} blockage is characterized by the fraction $a$ of the total land area occupied by buildings, the mean number of buildings $\eta$ per \SI{}{km}$^2$, and the height of buildings modeled by a Rayleigh \ac{PDF} with a scale parameter $c$. Hence, the probability of having a \ac{LoS} communication from a \ac{BS} at horizontal-distance $r_i$ from an arbitrary UAV-UE  is given, similar to  \cite{8692749} and \cite{8713514}, as: %\red{Jacek delta - Tabassum Series expansions}
\begin{align}
\label{prob-los}
\Pb_{l}(r_i) = \prod_{n=0}^{m}\Bigg[1 - {\rm exp}\Big(- \frac{\big(h_{\rm BS} + \frac{h(n+0.5)}{m+1}\big)^2}{2c^2}\Big) \Bigg], 
\end{align}		% ${\rm max}(p-1,0)$
where $h$ represents the difference between the UAV-UE altitude and \ac{BS} height, which depends on whether the UAV-UE is static or mobile, and $m=\floor{\frac{r_i\sqrt{a\eta}}{1000}-1}$. Different terrain structures and environments can be considered by varying the tuple $(a,\eta,c)$. As previously discussed, the performance of relatively high-altitude UAV-UEs is limited by the LoS interference they encounter and reduced serving antenna gain (from the antennas' side-lobes). We hence propose a multi-\ac{BS}  cooperative transmission scheme that mitigates  inter-cell interference and, thus, improves the performance of high-altitude UAV-UEs. 
%\blue{\emph{It is worth pointing out that, under the high-altitude UAV-UE assumption, our model represents a scenario of \acp{BS} cooperating from their antennas' side-lobes}.} 
%entails 
%ensembles
%represents 
%resembles
%
Hence, the antenna gain plus path loss for each component, i.e., \ac{LoS} and \ac{NLoS}, will be		%\quad \quad \quad\red{$G(r_i)$}
\begin{align}
\label{zeta-fading} 
\zeta_v(r_i) &= A_v G_s d_i^{-\alpha_v}  = A_v G_s \big(r_i^2 + h^2\big)^{-\alpha_v/2},
%= A_v G_s \big(r_i^2 + h^2\big)^{-\alpha_v/2},
\end{align}
where $d_i$ is the communication link distance, $v\in\{l,n\}$, $\alpha_{l}$ and $\alpha_{n}$ are the path loss exponents for the  \ac{LoS} and NLoS links, respectively, with $\alpha_{l}<\alpha_{n}$, and $A_{l}$ and  $A_{n}$ are the path loss constants at the reference distance $d_i = \SI{1}{m}$ for the \ac{LoS} and \ac{NLoS}, respectively. $G_s$ is the antenna directivity gain of  side-lobes between \ac{BS} $i$ and an arbitrary  UAV-UE since, at such high altitudes,  UAV-UEs are served by the side-lobes of \ac{BS} antennas \cite{lin2018sky}. The \ac{BS} vertical antenna pattern is directional and typically down-tilted to account for \acp{GUE}. Given this setup, it is reasonable to assume that UAV-UEs are always served from the antennas' side-lobes while the \acp{GUE} are served from the antennas' main-lobes with antenna gains $G_s$ and $G_m$, respectively, where $G_s \ll G_m$. 
% \blue{Interested readers are referred to \cite{khoshkholgh2019non1,azari2018reshaping} for further details on the effect of antennas' down-tile angle and side-lobes on the performance of networks of UAVs.} 

%Next, we conduct the coverage probability analysis for static UAV-UEs. We develop a novel framework to characterize the performance of communication-assisted CoMP for cellular-connected UAVs. The performance of UAVs is then contrasted to their terrestrial  counterparts.
%Once those metrics are derived, considering the D2D users density, we obtain optimal values
%for the UAV altitude that maximize the coverage probability and average rate. 
Having defined our system model, next, we will consider two scenarios: Static UAV-UEs and mobile UAV-UEs. For each scenario, we will characterize the coverage probability of high altitude UAV-UEs that are collaboratively served from \acp{BS} within their cluster. The performance of collaboratively-served UAV-UEs is then compared to their terrestrial counterparts and to UAV-UEs under the nearest association scheme. Moreover, we will characterize the handover rate  for mobile UAV-UEs and quantify the negative impact of mobility on their achievable performance.

% and precision agriculture 
\vspace{-0.1 cm}
\section{Coverage Probability of Static UAV-UEs}		% Performance Analysis of communication-assisted CoMP
\label{cov-prob-static}
Hovering drones can provide appealing solutions for a wide range of applications such as traffic control and surveillance \cite{austin2011unmanned}. We here assume static UAV-UEs that hover at a fixed altitude $h_d$, where $h_d>h_{\rm BS}$. Hence, we set $h=h_d-h_{\rm BS}$ in (\ref{prob-los}) and (\ref{zeta-fading}). We also assume that  \acp{BS} within one cluster cooperatively serve one UAV-UE whose projection on the ground is at the cluster center. Note that assuming such a cluster-center UAV-UE is mainly done for tractability, but its performance can be seen as an \ac{UB} on the performance of a randomly located UAV-UE inside the cluster \cite{7880694}. Given that a \ac{PPP} is translation invariant with respect to the origin, for simplicity, we conduct the coverage analysis for a UAV-UE located at the origin in $\R^2$, referred to as the \emph{typical UAV-UE} \cite{haenggi2012stochastic}. Next, we first characterize the serving distance distribution, and then, we employ it to derive upper and lower bounds on the coverage probability of static UAV-UEs.

\vspace{-0.6 cm}
\subsection{Serving Distance Distributions}
Under the condition of having $\kappa$ serving \acp{BS} in the cluster of interest, the distribution of in-cluster  \acp{BS} will follow a \ac{BPP} \cite{haenggi2012stochastic}. This \ac{BPP} consists of $\kappa$ uniformly and independently distributed \acp{BS} in the cluster. The set of cooperative \acp{BS} is defined as $\Phi_{c} = \{b_i \in \Phi_{b} \cap \mathcal{B}(0, R_c)\}$, where $\mathcal{B}(0, R_c)$ denotes the ball centered at the origin $(0,0)\in\R^2$ with radius $R_c$. Recall that the typical UAV-UE is located at the origin in $\mathbb{R}^2$, i.e., $(0,0,h_d) \in \mathbb{R}^3$.  The \ac{2D} distances from the cooperative \acp{BS} to the typical UAV-UE are represented by $\boldsymbol{R}_{\kappa}= [R_1, \dots, R_{\kappa}]$. Then, conditioning on $\boldsymbol{R}_{\kappa} = \boldsymbol{r}_{\kappa}$, where $\boldsymbol{r}_{\kappa}= [r_1, \dots, r_{\kappa}]$, the conditional joint \ac{PDF} of the serving distances is $f_{\boldsymbol{R}_{\kappa}}(\boldsymbol{r}_{\kappa})$. 
The $\kappa$ cooperative \acp{BS} can be seen as the $\kappa$-closest \acp{BS} to the cluster center from the PPP $\Phi_{b}$. Since the $\kappa$ \acp{BS} are independently and uniformly distributed in the cluster approximated by $\mathcal{B}(0, R_c)$, the \ac{PDF} of the horizontal distance from the origin to \ac{BS} $i$ will be: $ f_{R_i}(r_i)=\frac{2r_i}{R_c^2}$, $0\leq r_i\leq R_c$, 
%\[
%    f_{R_i}(r_i)=\left\{
%                \begin{array}{ll}
%                  \frac{2r_i}{R_c^2}, \quad\quad\quad 0\leq r_i\leq R_c,\\ 
%                 0, \quad\quad\quad\quad {\rm otherwise},
%                \end{array}   
%              \right.
%  \]
for any $i \in \mathcal{K}_f = \{1, \dots, \kappa\}$, where $\mathcal{K}_f$ is the set of collaborative \acp{BS} within the ball $\mathcal{B}(0, R_c)$. From the \ac{i.i.d.} property of \ac{BPP}, the conditional joint \ac{PDF} of the serving distances $\boldsymbol{R}_{\kappa}= [R_1, \dots, R_{\kappa}]$ is expressed as
%\begin{align}
%\label{serv-dist}
$f_{\boldsymbol{R}_{\kappa}}(\boldsymbol{r}_{\kappa})=\prod_{i=0}^{\kappa} \frac{2r_i}{R_c^2}$. 
%\end{align}
 
%\section{Coverage Probability Analysis}		{prob-los} 		\Pb_{l}(r_i)
%\blue{Note that we omit the index $i$ as the horizontal distances from the typical UAV-UE to \acp{BS} have the same distribution regardless of their index values.}  		%\sqrt{P_t} G(r_i)^{0.5} \zeta_v(r_i)^{0.5}
\vspace{-0.5 cm}
\subsection{\black{Performance of UAV-UEs}}
\label{static-uav-comp} 
Under the condition of having $\kappa$ serving \acp{BS}, the received signal at the UAV-UE will be:	
\begin{align}
\label{rec-pwr}
P &= \underbrace{\sum_{i=1}^{\kappa}  P_v(r_i) \omega_i w_i Y_0}_{\text{desired signal}}+ \underbrace{\sum_{k\in\Phi_{b} \setminus \mathcal{B}(0, R_c)} P_v(u_k) \omega_k w_k Y_k}_{\text{interference}} + 
\quad  Z,
% \underbrace{\sum_{j \in \Phi_{b} \setminus\Phi_{c} }\sqrt{P_t} \sqrt{\zeta_v(r_j)}\omega_j w_j Y_j}_{\text{Interference}} ,	
\end{align}
where the first term represents the desired signal from $\kappa$ collaborative \acp{BS} with $P_v^2(r_i) =P_t  \zeta_v(r_i)$, $v \in \{l,n\}$, $\omega_i$ being the Nakagami-$m_v$ fading variable of the channel from \ac{BS} $i$ to the UAV-UE, $w_i$ is the precoder used by \ac{BS} $i$, and $Y_0$ is the channel input symbol that is sent by the cooperating \acp{BS}. The second term represents the inter-cluster interference, whose power is denoted as $I_{\rm out}$, where $Y_j$ is the transmitted symbol from interfering \ac{BS} $j$ and $u_j$ is the horizontal distance between interfering \ac{BS} $j$ and the UAV-UE; $Z$ is a circular-symmetric zero-mean complex Gaussian \ac{RV} that models the background thermal noise.  
%
%\red{In-cluster interference occurs only for the case in which not all of the collaborative \acp{BS}  (within distance $R_c$) have the cached content (i.e., $c_f <1$).} 
%
%
% Given the high-altitude assumption of the UAV-UE, we will consider that the interfering \acp{BS} have dominant \ac{LoS} communications to the UAV-UE, i.e,  $\Pb_{l}(\nu)=1$ and $\Pb_{n}(\nu)=0$.
% The accuracy of this will be verified in Section \ref{num-result}.
%%%%%%%%%%%%%%%%%%%%%%%%%%%%%%%%%%%%%%%%%%%%%%%%%%%%%%%%%%%%%%\sir_%%%%%%%

As discussed earlier, UAV-UEs exhibit a \ac{LoS} component which becomes dominant at relatively high altitudes. The \ac{LoS} probability in (\ref{prob-los}) represents a delta function that goes from one to zero as $r_i$ increases. This implies that the probability of \ac{LoS} communication from close \acp{BS} is higher than that of remote \acp{BS}. Hence, we consider that the desired signal is dominated by its \ac{LoS} component where $v=l$, $m_v=m_l$, and $P_v(r_i) =\sqrt{P_t}  \zeta_l(r_i)^{0.5}$. However, for the interfering signal, both \ac{LoS} and \ac{NLoS} components exist and, thus, we have: $P_v(u_j) =\sqrt{P_t}  \zeta_v(u_j)^{0.5}$, $v\in\{l,n\}$. This is due to the fact that, as the \ac{LoS} probability decreases with the interfering distance $u_j$, the \ac{LoS} assumption becomes less practical for far but interfering \acp{BS}. 

We assume that the \ac{CSI} is available at the serving \acp{BS}. Hence, \ac{MRT} can be adopted by \acp{BS} to maximize the received power at the typical UAV-UE. For the \ac{MRT}, we have the precoder $w_i$ set as $w_i=\frac{\omega_i^*}{|\omega_i|}$, where $\omega_i^*$ is the complex conjugate of $\omega_i$. Assuming that $Y_0$ and $Y_k$ in (\ref{rec-pwr}) are independent zero-mean \acp{RV} of unit variance, and neglecting the thermal noise, the conditional $\sir$ at the typical UAV-UE will then be:
\begin{align}
\label{nakagami}
\Upsilon_{|\boldsymbol{r}_{\kappa}} &= \frac{ P_t \Big|\sum_{i=1}^{\kappa} \zeta_{l}^{1/2}(r_i) w_i \omega_i \Big|^2}{\sum_{k\in\Phi_{b} \setminus \mathcal{B}(0, R_c)} \big|P_v(u_k) \omega_k w_k\big|^2},
%\nonumber \\
%&= {\kappa \choose o} \prod_{i=0}^{o}\Pb_{l}(r_i)\prod_{j=o+1}^{\kappa}\Pb_{n}(r_j) 
%\frac{ P_t \Bigg|\sum_{i=1}^{o} \zeta_{l}^{1/2}(r_i) \chi_i +  \sum_{j=o+1}^{\kappa} \zeta_{n}^{1/2}(r_j) \chi_i \Bigg|^2}{I_{\rm out}}
%\nonumber \\
%&<= {\kappa \choose o} \Bigg( \prod_{i=0}^{l}\Pb_{l}(r_i) \Bigg) \Bigg( \prod_{i=l+1}^{\kappa}\Pb_{n}(r_i) \Bigg)  
%\frac{\kappa P_t \sum_{i=1}^{\kappa} Y_i^2}{I}
\end{align}
where $\Upsilon_{|\boldsymbol{r}_{\kappa}}$ is conditioned on the number of collaborative \acp{BS} $\kappa$, and on $\boldsymbol{R}_{\kappa} = \boldsymbol{r}_{\kappa}$. 
In (\ref{nakagami}), we have $\Big|\sum_{i=1}^{\kappa} \zeta_{l}^{1/2}(r_i)w_i\omega_i \Big|^2$ representing the square of a weighted sum of $\kappa$ Nakagami-$m_l$ \acp{RV}. Since there is no known closed-form expression for a weighted sum of Nakagami-$m_l$ \acp{RV}, we use the Cauchy-Schwarz's inequality to get an \ac{UB} on a square of weighted sum as follows:

\begin{align}
 \label{cauchy}
\Bigg|\sum_{i=1}^{\kappa} \zeta_{l}^{1/2}(r_i) w_i \omega_i\Bigg|^2=
\Bigg|\sum_{i=1}^{\kappa} \zeta_{l}^{1/2}(r_i) \frac{\omega_i^* \omega_i}{|\omega_i|}\Bigg|^2 &=
 \Bigg(\sum_{i=1}^{\kappa}Q_i\Bigg)^2 
\leq \kappa \Bigg(\sum_{i=1}^{\kappa}Q_i^2\Bigg),
\end{align}
where $Q_i=\zeta_{l}^{1/2}(r_i)\frac{\omega_i^* \omega_i}{|\omega_i|}=\zeta_{l}^{1/2}(r_i)\omega_i$ is a scaled Nakagami-$m_l$ \ac{RV}, and $i \in \mathcal{K}_f$. Since $\omega_i \sim$ Nakagami$(m_l)$, from the scaling property of the Gamma \ac{PDF}, $Q_i^2 \sim \Gamma\big(K_i=m_l,\theta_i=\zeta_{l}(r_i)/m_l\big)$. To get a tractable statistical equivalence of a sum of $\kappa$ Gamma \acp{RV} with different scale parameters $\theta_i$, we adopt the method of sum of Gammas second-order moment match proposed in \cite[Proposition 8]{heath2011multiuser}. It is shown that the equivalent Gamma distribution, denoted as $J \sim \Gamma(K,\theta)$, with the same first and second-order moments has the following parameters: 
\begin{align}
\label{equiv-gamma}
K = \frac{\Big(\sum_i{K_i\theta_i}\Big)^2}{\sum_i{K_i\theta_i^2}}	= \frac{m_l\Big(\sum_i{\zeta_{l}(r_i)}\Big)^2}{\sum_i{\Big(\zeta_{l}(r_i)\Big)^2}} 
\quad \quad \text{and} \quad \quad
\theta = \frac{\sum_i{K_i\theta_i^2}}{\sum_i K_i\theta_i}=\frac{\sum_i{\zeta_{l}(r_i)^2}}{m_l \sum_i  \zeta_{l}(r_i)}.
\end{align}
The accuracy of the Gamma approximation can be easily verified via numerical simulations that are omitted due to space limitations. For tractability, we further upper bound the shape parameter $K$ using the Cauchy-Schwarz's inequality as:
%\begin{align}
%\label{approx-k} 
$K   \leq \frac{m_l\kappa\sum_{i} \big(\zeta_l(r_i)\big)^2}{\sum_{i} \big(\zeta_l(r_i)\big)^2} = m_l\kappa$, 
%\end{align}
where, by definition, $m_l\kappa$ is integer. We shall also see shortly the tightness of this \ac{UB}. 
%% We next use this approximation to derive an \ac{UB} on the coverage probability, and its tightness is shown in section IV. %which will be numerically shown to be considerably tight in Section V
%%%%%%%%%%%%%%%%%%%%%%%%%%%%%%%%%%%%%%%%%%%%%%%%%%%%%%%%%%%%%%%
%% if we need space we can perhaps remove this and make a short 1 sentence saying 
%%the accuracy can be easily verified via numerical simulations that are omitted due to space limitations.
%\begin{figure} [!t] 	%[!t] %%		[htbp]	
%\centering
%\includegraphics[width=0.45\textwidth]{Figures/ch5/gamma.eps}		%letter			
%\caption {Simulation of the PDF of the equivalent gain of desired channels between cooperating \acp{BS}  and the UAV-UE, including path-loss and fading. A PPP realization of density $\lambda_b=20$ \ac{BS}/\SI{}{km}$^2$ is run for a simulated area of $\SI{100}{km}^2$ and $R_c=\SI{200}{m}$.}		%$m_l=3$ and
%\label{gamma-effect}
%\vspace{-0.5 cm}
%\end{figure}

Next, we derive \ac{UB} and \ac{LB}  expressions on the UAV-UE coverage probability. Our developed approach is novel in the sense that it adopts the Cauchy-Schwarz's inequality and moment match of Gamma \acp{RV} to characterize an \ac{UB} on the coverage probability, which is difficult to obtain exactly. The UAV-UE coverage probability conditioned on $\boldsymbol{R}_{\kappa} = \boldsymbol{r}_{\kappa}$ is given by  
\begin{align}
\label{pc-rk}
\Pb_{{\rm c}|\boldsymbol{r}_{\kappa}}& \overset{(a)}{\leq}   
\Pb\Big(\frac{\kappa P_t \big(\sum_{i=1}^{\kappa}Q_i\big)^2 }{I_{\rm out}}>\vartheta\Big)
%\nonumber \\				% \Pb\big[\Upsilon_{|\boldsymbol{r}_{\kappa}}>\vartheta\big]
 \overset{(b)}{\approx} \Pb\Big(\frac{\kappa P_tJ}{I_{\rm out}}>\vartheta\Big),
\end{align}
where (a) follows from the Cauchy-Schwarz's inequality, (b) follows from the Gamma approximation and rounding the shape parameter $K=m_l\kappa$, and $\vartheta$ is the $\sir$ threshold. The coverage probability can be obtained as a function of the system parameters, particularly, the Nakagami fading parameter and collaboration distance, as stated formally in the following theorem.
\begin{theorem}
\label{ch5:cov-prob} An \ac{UB} on the coverage probability of UAV-UEs cooperatively served from \acp{BS} within a collaboration distance $R_c$ can be derived as follows:
\begin{align}
\label{cov-prob-theory}		
\Pb_{{\rm c}} = \sum_{\kappa=1}^{\infty} \Pb(n=\kappa) \int_{\boldsymbol{r_\kappa}=\boldsymbol{R_c}}^{\boldsymbol{\infty}} \Pb_{{\rm c}|\boldsymbol{r}_{\kappa}}^l \prod_{i=0}^{\kappa} \frac{2r_i}{R_c^2} \dd{ \boldsymbol{r}_{\kappa}},
\end{align}
where the conditional coverage probability $\Pb_{{\rm c}|\boldsymbol{r}}^l = \lVert e^{\boldsymbol{T}_{K}}\rVert_1$,     $\lVert.\rVert_1$ represents the induced $\ell_1$ norm, and $\boldsymbol{T}_{K}$ is the lower triangular Toeplitz matrix:

\[ \boldsymbol{T}_{K} =
\begin{bmatrix}
    t_{0}       \\
    t_{1}       & t_{0}  \\
    \vdots &  \vdots & \ddots  \\
    t_{K-1}       & \dots & t_{1}  & t_{0} 
\end{bmatrix};
\]
$K=m_l\kappa$, $t_{i} =\frac{(-\varpi)^{i}}{(i)!} \Omega^{(i)}(\varpi)$, $\Omega^{(i)}(\varpi)= \frac{d^{i}}{d \varpi^{i}} \Omega(\varpi)_{|\boldsymbol{r}_{\kappa}}$, $ \Omega(\varpi)_{|\boldsymbol{r}_{\kappa}} =-2\pi \lambda_b\int_{\nu=R_c}^{\infty}\Big(1 -\delta_l\Pb_{l}(\nu) - \delta_n\Pb_{n}(\nu) \Big)\nu\dd{\nu}$, $\delta_l=\Big(1 + \frac{\varpi P_l(\nu)^2 }{m_l} \Big)^{-m_l}$, $\delta_n=\Big(1 + \frac{\varpi P_n(\nu)^2 }{m_n} \Big)^{-m_n}$, and $\varpi = \vartheta/\kappa P_t\theta$. 
%
%\begin{align}
%&\Pb_{{\rm c}|\boldsymbol{r}_{\kappa}}=
%\sum_{k=0}^{K-1}  \frac{(-\varpi)^k}{k!}\frac{\partial^k}{\partial\varpi^k}  
%{\rm exp}\Bigg(-2\pi \lambda_b\int_{\nu=R_c}^{\infty}\Big(1 -
%        \delta_l\Pb_{l}(\nu) - \delta_n\Pb_{n}(\nu) 
%        \Big)\nu\dd{\nu}\Bigg), 
%        \label{cond-cov}
%\end{align}
\end{theorem}
\begin{proof}
Please see Appendix \ref{ch5:theorem1}.\footnote{Although there exists an infinite sum in (\ref{cov-prob-theory}), this sum vanishes for a small number of serving \acp{BS} that is determined by the collaboration distance $R_c$ and the BSs' density $\lambda_b$.}		%\red{Asymptotic with $R_c$}
\end{proof}
% %few number is wrong grammatically "either few BSs" or "a small number of"
%That said, you mean small number within a cluster?
%
The main steps towards tractable coverage are summarized as follows \cite{8490204}: We first derive the conditional log-Laplace transform $\Omega(\varpi)_{|\boldsymbol{r}_{\kappa}}$ of the aggregate interference. Then, we calculate the $i$-th derivative of $\Omega(\varpi)_{|\boldsymbol{r}_{\kappa}}$ to populate the entries $t_i$ of the lower triangular Toeplitz matrix $\boldsymbol{T}_{K}$. The conditional coverage probability can be then computed from $\Pb_{{\rm c}|\boldsymbol{r}}^l = \lVert e^{\boldsymbol{T}_{k}}\rVert_1$.		% [Methodology 1]
%

 %%%%%%%%%%%%%%%%%%%%%%%%%%%%%%%%%%%%%%%%%%%%%%%%%%%%%%%%%%%%%%%%%%%%%%%%%%%%%%%%
Important insights on the coverage probability can be obtained from (\ref{cov-prob-theory}). First, if the collaboration distance $R_c$ increases, both the probability $\Pb(n=\kappa)$ and the integrand value in (\ref{cov-prob-theory}) increase, and, thus,  the coverage probability grows accordingly. Furthermore, the effect of the \ac{BS} density $\lambda_b$ on the coverage probability is two-fold. On the  one hand, the average number of \acp{BS} increases with $\lambda_b$ as characterized by $\Pb(n=\kappa)$, which results in a higher desired signal power. On the other hand, this advantage is counter-balanced by the increase in (LoS) interference power when $\lambda_b$ increases, as captured in the decaying exponential functions in (\ref{LT_c1}). Additionally, this compact representation, i.e., $\Pb_{{\rm c}|\boldsymbol{r}}^l = \lVert e^{\boldsymbol{T}_{K}}\rVert_1$, leads to valuable system insights. For instance, the impact of the shape parameter $K=\kappa m_l$ on the intended channel gain $\Gamma(K,\theta)$ is rigorously captured by the finite sum representation in (\ref{prob-y2}) of Appendix \ref{ch5:theorem1}, which is typically related to the collaboration distance $R_c$ and the Nakagami fading parameter $m_l$.

%this yields the achievable rate \ac{LB}
%To circumvent this difficulty, we derive an approximation and a \ac{LB}
% hat the entries $t_k$ of the lower triangular Toeplitz matrix: $\boldsymbol{T}_{K}$ are obtained in closed-form expressions.
Next, we derive an \ac{LB} on the coverage probability, which will lead to closed-form expressions for $t_{k}$, i.e., the entries populating $\boldsymbol{T}_{K}$ in (\ref{cov-prob-theory}). Given the high-altitude assumption of UAV-UEs, we will consider a special case when interfering \acp{BS} have dominant \ac{LoS} communications to the typical UAV-UE, i.e,  $\Pb_{l}(\nu)=1$ and $\Pb_{n}(\nu)=0$ in (\ref{cov-prob-theory}). Since this case results in higher interference power, this yields the derived coverage probability \ac{LB}. 
\begin{corollary}
\label{ch5:cov-prob-lb}
An \ac{LB} on the coverage probability of the UAV-UEs can be computed from (\ref{cov-prob-theory}), where 
$\Pb_{{\rm c}|\boldsymbol{r}}^l = \lVert e^{\boldsymbol{T}_{K}}\rVert_1$, and the entries of $\boldsymbol{T}_{K}$ are given in closed-form expressions as
\begin{align}
\label{closed-form1}
%t_k&= a_k \Bigg(  \textbf{1}_{(k=0)} - c_k
%  {}_2 F_1(K,k-\delta_l;k+1-\delta_l;-\varpi L \eta R_{ch}^{-\alpha_l/2})  \Bigg), \\
 t_k& \overset{}{=} \pi  \lambda_b R_{ch}^2 \Big(  \textbf{1}\{k=0\} - c_k 
  {}_2 F_1(k+m_l,k-\delta_l;k+1-\delta_l;-\varpi L R_{ch}^{-\alpha_l/2} m_l)  \Big),
 \end{align}
 where $c_k =\frac{\delta_l a_k \Gamma(k+m_l) m_l^{-k} }{(\delta_l -k)\Gamma(k+1) \Gamma(m_l)}$, $a_k = (\varpi L R_{ch}^{-\alpha_l/2})^{k}$, $\delta_l = \frac{2}{\alpha_l}$, $R_{ch}^2 = R_c^2+h^2$, $\textbf{1}\{.\}$ is the indicator function, and ${}_2 F_1(\cdot,\cdot;\cdot;\cdot)$ is the ordinary hypergeometric function. 
\end{corollary}
\begin{proof}
Please see Appendix \ref{ch5:theorem2}.
\end{proof}

For comparison purposes, next, we derive the UAV-UE coverage probability under the nearest association scheme. 		% \red{Cell-free MIMO.}
\begin{corollary}
\label{corr-near}
The coverage probability of the UAV-UEs under the nearest association scheme is: 	% given by
\begin{align}
\label{cov-prob-theory0}		
\Pb_{{\rm c}} = \int_{0}^{\infty} \Pb_{{\rm c}|r_0}^l f_{R_0}(r_0) \dd{r_0},
\end{align}									% Theorem \ref{ch5:cov-prob}
where $\Pb_{{\rm c}|r_0}^l = \lVert e^{\boldsymbol{T}_{m_l}}\rVert_1$, $\boldsymbol{T}_{m_l}$ is defined as $\boldsymbol{T}_{K}$ in (\ref{cov-prob-theory}), with $\Omega(\varpi)=-2\pi \lambda_b\int_{\nu=r_0}^{\infty}\Big(1 -\delta_{l}\Pb_{l}(\nu) - \delta_{n}\Pb_{n}(\nu) \Big)\nu\dd{\nu}$, $\varpi = \frac{\vartheta m_l}{P_t \zeta_l(r_0)}$, and $f_{R_0}(r_0)=2\pi\lambda_b r_0 e^{-\pi\lambda_br_0^2}$ is the \ac{2D} serving distance \ac{PDF}. 
\end{corollary}
\begin{proof}
The proof follows directly from \cite{azari2017coexistence} and Theorem \ref{ch5:cov-prob}, and hence is omitted for brevity.
\end{proof}
\begin{figure}[!t]
\vspace{-0.9 cm}	
    \centering
    \subfigure[Inter-cluster half distance $R_h = \SI{190}{m}$]
    {
        \includegraphics[width=3.1in]{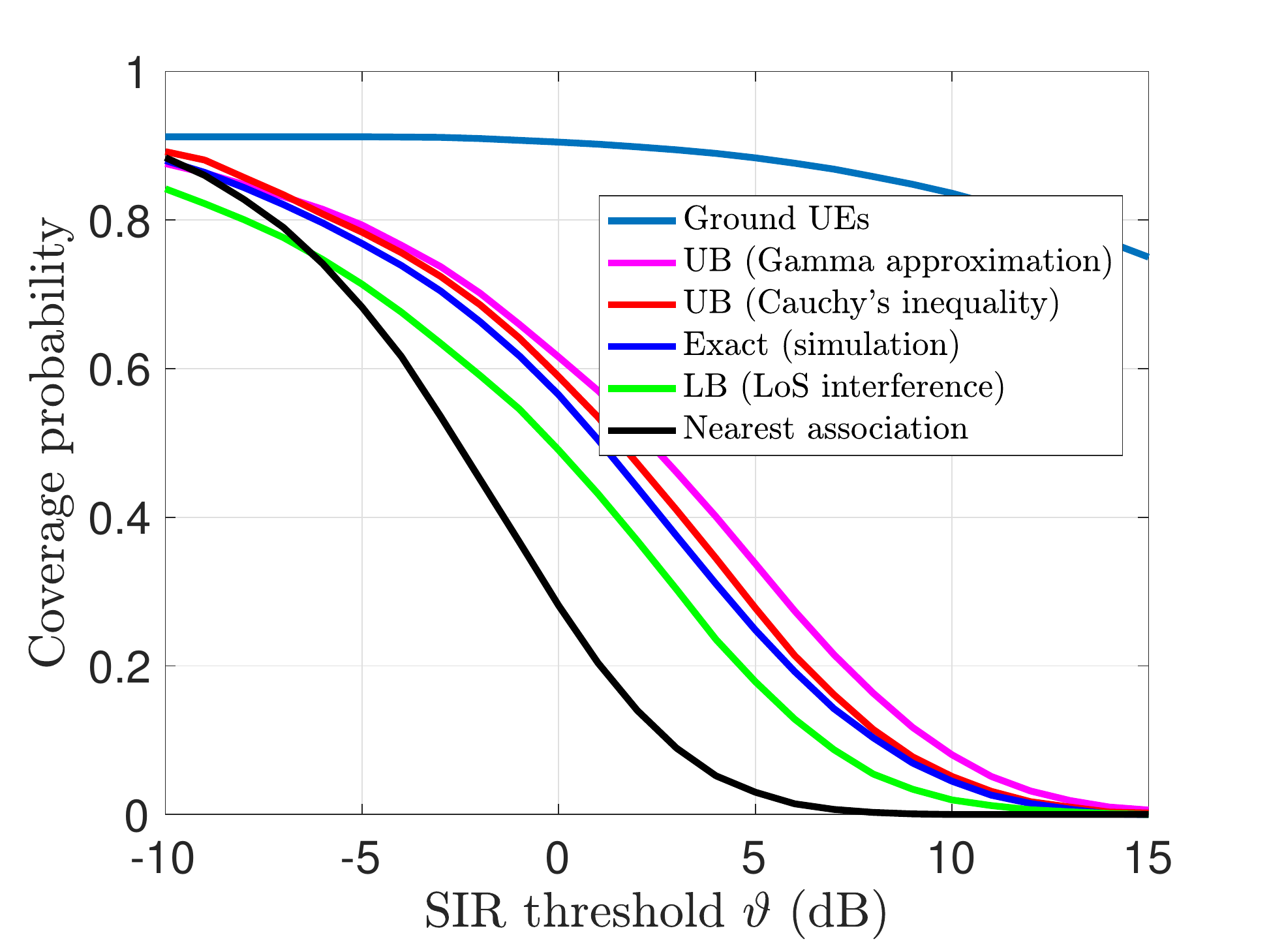}	%was 11.eps lower-bound.eps
        \label{cov_prob_vs_theta11}
    }
    \subfigure[SIR threshold $\vartheta=\SI{0}{dB}$]	% \red{see slope}
    {
        \includegraphics[width=3.1in]{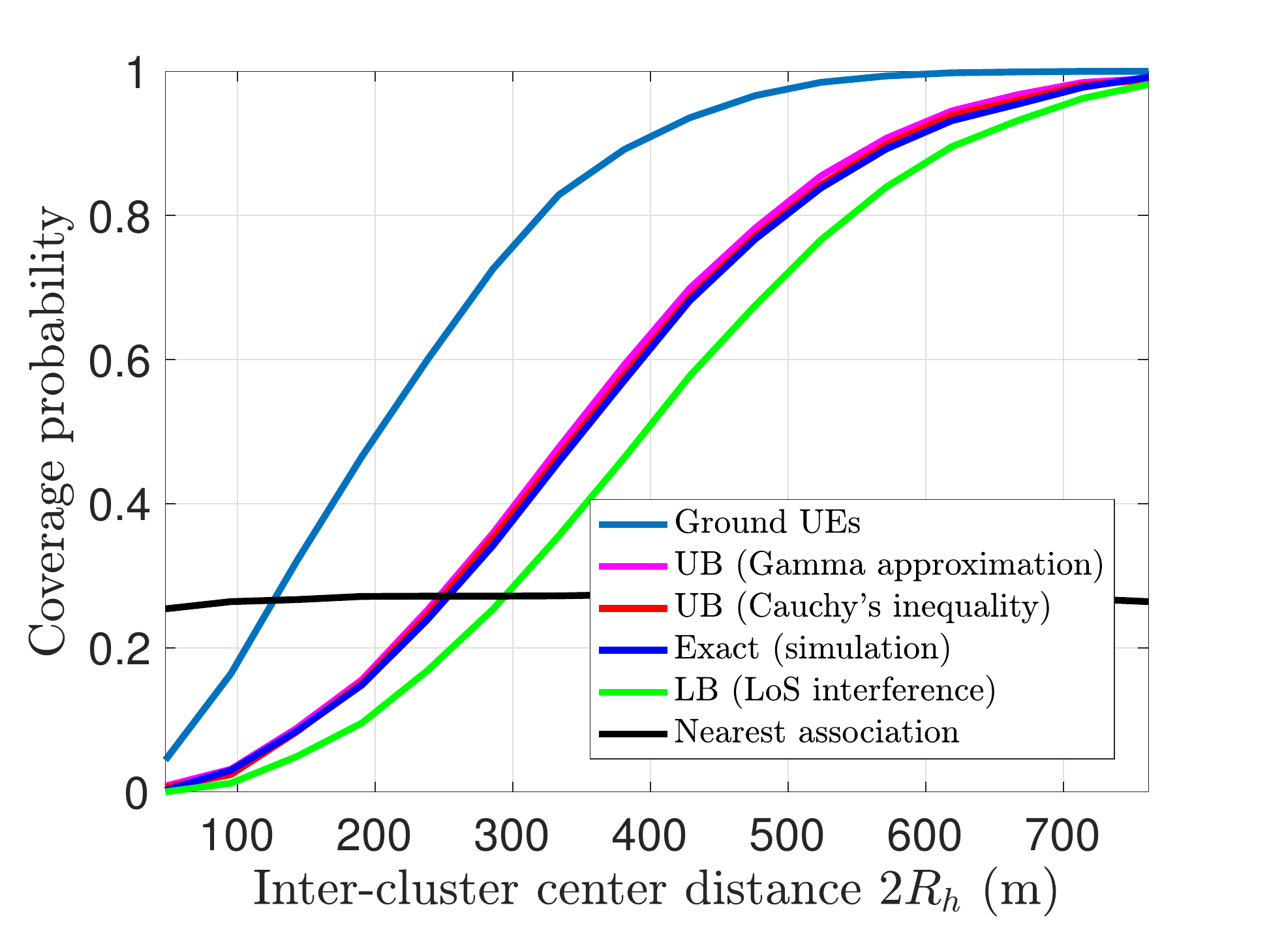}	 %22	%lower-bound.eps
        \label{cov_prob_vs_theta12}
    }
%    \subfigure[$a = 0.2$, $b = \SI{200}{km^{-2}}$, and $c = \SI{20}{m}$.]
%    {
%        \includegraphics[width=2.0in]{Figures/ch5/cov_prob_LUB3.eps}
%        \label{cov-prob-bounds}
%    }
    \caption{The derived upper and lower bounds on the coverage probability of UAV-UEs are plotted versus the \ac{SIR} threshold $\vartheta$ and collaboration distance $R_c$: $\lambda_b = \SI{20}{km^{-2}}$, $R_{\rm sim}= \SI{20}{km^2}$, $\alpha_l = 2.09$, $\alpha_n = 3.75$, $h_{\rm BS} = \SI{30}{m}$, $m_l=3$, $A_L = 0.0088$, $A_N = 0.0226$, $h_d = \SI{120}{m}$, $a = 0.3$, $b = \SI{300}{km^{-2}}$, and $c = \SI{20}{m}$.}		
    % $a = 0.3$, $b = \SI{300}{km^{-2}}$, and $c = \SI{20}{m}$
    \label{cov_prob_vs_theta_hd110}
    \vspace{-0.7 cm}
\end{figure}
To verify the accuracy of our proposed approach, in Fig.~\ref{cov_prob_vs_theta_hd110}, we show the theoretical UB and LB on the coverage probability of the UAV-UEs, and simulation of the \ac{UB} based on (\ref{cauchy}). Fig.~\ref{cov_prob_vs_theta11} shows that the  Cauchy's inequality-based \ac{UB} is remarkably tight. Moreover, although the obtained \ac{UB} expression in (\ref{cov-prob-theory}) is less tight, it still represents a reasonably tractable bound on the exact coverage probability. Hence, (\ref{cov-prob-theory}) can be treated as a proxy of the exact result. Recall that (\ref{cauchy}) is based on an \ac{UB} on a square of a sum of Nakagami-$m_l$ \acp{RV} while the expression in (\ref{cov-prob-theory}) goes further by two more steps. First, we approximate the sum of Gamma \acp{RV} to an equivalent Gamma \ac{RV}. Then, we round the shape parameter of the yielded Gamma \ac{RV} to an integer $m_l\kappa$. Finally, the \ac{LB} based on (\ref{closed-form1}) can be also seen as a relatively looser bound than the \acp{UB}. As evident from Fig.~\ref{cov_prob_vs_theta_hd110}, allowing CoMP transmission significantly enhances the coverage probability, \black{e.g., from $28\%$ for the baseline scenario with nearest serving \acp{BS} to $60\%$ at $\vartheta=$ \SI{-5}{dB} (for an average of 2.5 cooperating \acp{BS}).} In Fig.~\ref{cov_prob_vs_theta_hd110}, the performance  of UAV-UEs is also compared to that of their ground counterparts experiencing Rayleigh fading and \ac{NLoS} communications. We assume that the \acp{BS}' antennas are ideally down-tilted accounting for the \acp{GUE}, i.e., the antenna gains for desired and interfering signals are $G_m$ and $G_s$, respectively. Under such a setup, we observe that the coverage probability of \acp{GUE} substantially outperforms that of  UAV-UEs, especially at high \ac{SIR} thresholds. Fig.~\ref{cov_prob_vs_theta12} shows the prominent effect of the collaboration distance $R_c$ on the coverage probability of ground and aerial \acp{UE}. We can see that for both kind of \acp{UE}, the coverage probability monotonically increases with $R_c$ since more \acp{BS} cooperate to serve the UEs when $R_c$ increases. Moreover, due to the down-tilt of the \acp{BS}' antennas and LoS-dominated interference for UAV-UEs, the coverage probability of \acp{GUE} outperforms that of the UAV-UEs. However, we can see that the rate of coverage probability improvement with $R_c$, i.e., the slope, is higher for the UAV-UE. This can be interpreted by the fact that increasing $R_c$ yields more \ac{LoS} signals within the desired signal side and subtracts them from the interference. Conversely, for \acp{GUE}, the transmission is dominated by NLoS signals and Rayleigh fading. 
%  always

Having characterized the performance of static UAV-UEs, next, we turn our attention to applications in which the UAV-UEs can be mobile. It is anticipated that mobile UAV-UEs will span a wide variety of applications, e.g.,  flying taxis and delivery drones. Hence, it is quite  important to ensure reliable connections in the presence of UAV-UE mobility by potentially mitigating the \ac{LoS} interference through \ac{CoMP} transmissions. Moreover, unlike the \acp{GUE} that can only move horizontally, UAV-UEs can fly in \ac{3D} space. Hence, a \ac{3D} mobility model is essential to convey a realistic description of the performance of mobile UAV-UEs. As a first step in this direction, we develop a novel \ac{3D} \ac{RWP} mobility model that effectively captures the vertical displacement of UAV-UEs, along with their typical \ac{2D} spatial mobility. The use of RWP mobility is motivated by its simplicity and tractability that is widely adopted in the mobility analysis in cellular networks  \cite{1233531,bettstetter2004stochastic,1624340,6477064}. Moreover, as we will discuss shortly, it has tunable parameters that can be set to appropriately describe the mobility of different mobile nodes, ranging from walking or driving users to \ac{3D} UAVs, \cite{8671460} and \cite{6477064}.

\vspace{-0.2 cm}
\section{3D Mobility and Handover Analysis}	% \blue{Handover-Trajectory/Energy Dilemma}
\label{3D-RWP}
Next, a novel \ac{3D} \ac{RWP} model is presented to describe the motions of UAV-UEs. We first illustrate the various elements of our proposed model. Then, we characterize the handover rate and handover probability for mobile UAV-UEs. Since we introduce the first study on 3D mobile UAV-UEs, for completeness, we consider two cases: UAV-UEs under CoMP transmissions, and UAV-UEs served by the nearest GBS.

In the classical \ac{2D} mobility model, the spatial motion is considered only through a displacement and an angle. However, for the UAV-UE, due to the mission requirements, and environmental and atmospherical conditions, the UAV-UEs must change their altitude and make vertical motions. For instance, due to variations in the altitudes of buildings, UAV-UEs might have frequent up and down displacements along their trajectories. Indeed, the vertical motion is always associated with the take-off and landing of UAV-UEs. This inherently triggers the concept of \ac{3D} mobility in \ac{3D} space.\footnote{We assume that the UAV-UEs are sparsely deployed such that there are no imposed constraints on the trajectories of different UAV-UEs. The analysis of multiple trajectories with such constraints is interesting but beyond the scope of this paper.}

%On the basis of the proposed framework, incorporating user mobility is our work in progress
%In this paper, we are particularly interested in a special case
%Nevertheless, these results are equally applicable to every tier, and our focus in this paper is relative impact of using multiple tiers, rather than the optimal performance of each tier itself.

First, recall that in a classical \ac{RWP} mobility model \cite{1233531,bettstetter2004stochastic,1624340,6477064}, the movement trace of a node (e.g., the UAV-UE) can be formally described by an infinite sequence of tuples: $\{(\boldsymbol{X}_{n-1},\boldsymbol{X}_{n}, V_n)\}$, $\forall \boldsymbol{X}_{n} \in \R^3$, and $n \in \mathbb{N}$, where $n$ is the movement epoch and $\boldsymbol{X}_{n}=(\varrho_n,\phi_n,z_n)$ is the \ac{3D} cylindrical displacement of the UAV-UE at epoch $n$, see Fig.~\ref{shaded-area}. During the $n$-th movement epoch, $\boldsymbol{X}_{n-1}$ denotes the starting waypoint, $\boldsymbol{X}_{n}$ denotes the target waypoint, and $V_n$ is the velocity. For simplicity, we assume that the UAV-UE moves with a constant velocity $\bar{\nu}$. However, further extensions to generalized \acp{PDF} of the velocity directly follow by the same methodology of analysis. Given the current waypoint $\boldsymbol{X}_{n-1}$, the next waypoint $\boldsymbol{X}_n$ is chosen such that the included angle $\phi_n$ between the projection of the vector $\boldsymbol{X}_{n-1}-\boldsymbol{X}_{n}$ on the $x$-$y$ plane and the abscissa is uniformly distributed on $[0, 2\pi]$. We define the transition length as the Euclidean distance between two \ac{3D} successive waypoints, i.e., $u_n=\lVert \boldsymbol{X}_{n}-\boldsymbol{X}_{n-1}\rVert=\sqrt{\varrho_n^2+(z_n-z_{n-1})^2}$. Furthermore, the vertical displacement between two consecutive points, i.e., the change in $z$-axis, is also distributed according to a \ac{RV}. We also let $\varphi_n$ be  the acute angle of  $U_n = \lVert \boldsymbol{X}_{n}-\boldsymbol{X}_{n-1}\rVert$ relative to the horizontal line $\rho_n$.
%, see Fig.~\ref{shaded-area}.
 
% $\boldsymbol{X}_{n-1}$ and $\boldsymbol{X}_{n}$
% \magenta{Furthermore, if the velocity of UAV-UEs has small variations around its mean,  the velocity can be approximated by its mean.} 
%For instance, in \cite{8671460}, proposed a uniform distribution.
% this model is suitable for a scenario where

%By considering the  vertical displacement along $z$-direction, along with the typical \ac{2D} spatial mobility, we fully describe the \ac{3D} mobility. 
\begin{figure}[!t]	
\vspace{-0.8 cm}
    \centering
    \subfigure[A sample trace of the proposed \ac{3D}  \ac{RWP} mobility model]
    {
        \includegraphics[width=3.5in]{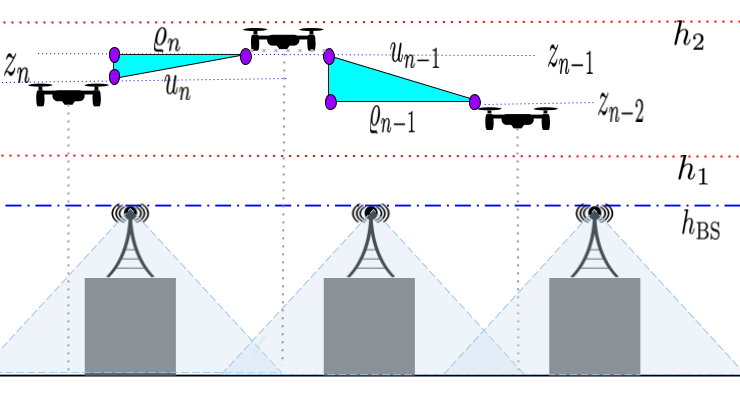}		%lower-bound.eps
        \label{shaded-area-1}
    }
    \subfigure[Illustration of \ac{1D} \ac{RWP} vertical mobility \cite{8671460}]
    {
        \includegraphics[width=2.5in]{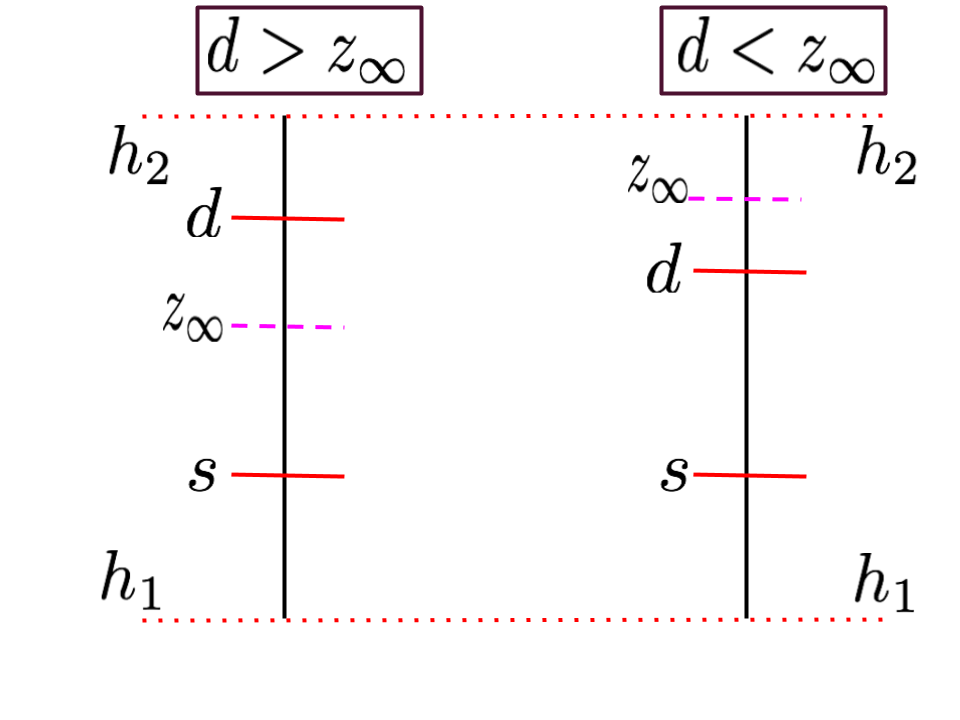}
        \label{shaded-area-2}
    }
    \caption{The proposed \ac{3D} mobility model for UAV-UEs which incorporates the typical \ac{2D} spatial \ac{RWP} and \ac{1D} \ac{RWP} for the vertical displacements.}	 % \red{squished}	
    \label{shaded-area}
    \vspace{-0.6 cm}
\end{figure}

 For simplicity, the selection of waypoints is assumed to be independent and identical for each movement epoch \cite{bettstetter2004stochastic}. Particularly, similar to \cite{6477064}, the horizontal transition lengths $\{\rho_1, \rho_2, \dots\}$ are chosen to be \ac{i.i.d.} with the \ac{CDF} 
%\begin{align}
%\label{cdf-hor-dist}
$F_{\rho_n}(\varrho_n) = 1 - {\rm exp}(-\pi \mu \varrho_n^2)$, 
%\end{align}
i.e., the spatial transition lengths are Rayleigh distributed in $\R^2$ with mobility parameter $\mu$. 
%In other words, given the waypoint $\boldsymbol{X}_{n-1} \in \mathbb{R}^3$, the horizontal projection of $\boldsymbol{X}_{n}$ on the horizontal plane $Z=z_{n-1}$ can be interpreted as the nearest point from \ac{2D} homogeneous PPP $\Phi_{\mu}(n) $ of intensity $\mu$ \cite{6477064}.  
%	i.e., $\boldsymbol{X}_n = \argmin_{x \in \Phi_{\mu}(n) } \lVert \boldsymbol{x} - \boldsymbol{X}_{n-1}\rVert $
The corresponding displacement \ac{PDF} is hence $f_{\rho_n}(\varrho_n) = \frac{\partial F_{\rho_n}(\varrho_n)}{\partial  \varrho_n} =2 \pi  \mu  \varrho_n e^{-\pi  \mu \varrho_n^2}$. As also done in \cite{8671460} and \cite{8421028}, we adopt a uniform distribution for the vertical displacement, however, the analysis for generalized \acp{PDF} can readily follow. In particular, we assume that $Z_n$ is uniformly distributed on $[h_1,h_2]$, i.e., $Z_n \sim \U(h_1,h_2)$ and $f_{Z_n}(z_n) = \frac{1}{h_2-h_1}, \forall h_1 \leq z_n \leq h_2$. 
We henceforth refer to $\hbar=h_2-h_1$ as altitude difference. Since the major restrictions of all drones' operations are their flying altitudes, it is reasonable to assume that $Z_n$ is bounded by $h_1$ and $h_2$. For instance, \acp{UAV} cannot fly higher than certain altitudes (\ac{AGL}) that are typically chosen below the cruising altitude of manned aircrafts. The \acp{UAV} also have an inherent minimum altitude of zero \ac{AGL}. However, due to mission requirements as well as environmental and atmospherical conditions, it is reasonable to assume $h_1>0$. We further assume that $h_1 > h_{\rm BS}$ for a high altitude UAV-UEs scenario. Finally, for the \ac{3D} displacement, we have  $u_n= \sqrt{\varrho_n^2 + (z_n-z_{n-1})^2}$. Since $U_n$, $\rho_n$, or $Z_n$ are \ac{i.i.d.} \acp{RV} among different time epochs, we henceforth omit the epoch index. % $n$  $n\in\mathbb{N}$

Given the independence assumption between $Z$ and $\rho$, we obtain their joint \ac{PDF} from
%\begin{align}
$f_{\rho,Z}(\varrho,z)  = f_{\rho}(\varrho) f_{Z}(z) = 
 \frac{ 2 \pi  \mu \varrho}{\hbar}     e^{-\pi  \mu \varrho ^2}$,  
 $\forall h_1 \leq z \leq h_2, 0 \leq \varrho \leq \infty$. 
% \quad \quad  \forall h_1 \leq z \leq h_2, 0 \leq \varrho \leq \infty.
%\Eb(U_n) = \frac{\Gamma \left(\frac{1}{3}\right)}{6^{2/3} \sqrt[3]{\pi } \sqrt[3]{\mu }}
%\end{align}
%
% $\lVert z_n-z_{n-1}\rVert$
Under the proposed mobility model, different mobility patterns can be captured by choosing different mobility parameters $\mu$. For example, larger values of $\mu$ statistically implies that $\rho$ and, consequently, the \ac{3D} transition lengths $U$ are shorter. This means that the movement direction switch rates are higher. These values of the mobility parameter appropriately describe the motion of UAV-UEs frequently travelling between nearby hovering locations such as for the use case of aerial surveillance cameras. In contrast, smaller $\mu$ statistically implies that $\rho$ and, consequently, $U$ are longer and the corresponding movement direction switch rates are lower. These values of $\mu$ would be suitable to describe the motion of UAV-UEs traveling large distances such as  for the use case of flying taxis and delivery drones. Given the \acp{PDF} of spatial and vertical motions, the \ac{PDF} of the \ac{3D} displacement $f_U(u)$ is readily obtained in the next lemma.
\begin{lemma}
\label{ch5:lemma1}
The PDF of the \ac{3D} transition lengths $U$ is given by 
%\begin{align}
%\label{3D-dist}
$f_U(u)  = 2 \pi  \mu u  e^{-\pi  \mu u^2}  \Omega(\mu,\hbar)$,
%\end{align}
where $\Omega(\mu,\hbar) = \frac{\pi \hbar \sqrt{\mu }  {\rm erfi}\left(\sqrt{\pi \mu } \hbar\right)-e^{\pi  \mu  \hbar^2}+1}{\pi  \mu \hbar^2}$, and ${\rm erfi}(.)= \frac{-2i}{\sqrt{\pi}}\int_{0}^{x}e^{-t^2}dt$. %is the imaginary error function. % , and ${\rm erfi}(x) = -i \frac{2}{\sqrt{\pi}}\int_{0}^{x}e^{-t^2}dt$
\end{lemma}
\begin{proof}
We can reach this result by transforming the \acp{RV} $Z_{n}$, $Z_{n-1}$, and $\rho_{n}$ to $U_n$, where $u_n=\sqrt{\varrho_{n}^2+(z_n-z_{n-1})^2}$, with the details omitted due to space limitations.
% Please see Appendix \ref{ch5:poof-lemma1}.
\end{proof}

\begin{remark}
{\rm If $h_1=h_2$, it can be easily verified that ${\rm lim }_{\hbar \to 0}\Omega(\mu,\hbar) \to1$, and $f_U(u) =2 \pi  \mu u  e^{-\pi  \mu u^2}$. This shows that if the UAV-UE moves only along a horizontal plane, the \ac{PDF} of the \ac{3D} displacement distance is reduced to its \ac{2D} counterpart, which verifies the correctness of the obtained \ac{3D} displacement distribution $f_U(u)$.} 
\end{remark} 		% using L'Hospital's rule
Having described the various elements of our proposed \ac{3D} \ac{RWP} model, our immediate objective is to  characterize the handover rate and handover probability for mobile UAV-UEs under \ac{CoMP} transmissions and nearest association. 
% . Since we introduce the first study on \ac{3D} mobile UAV-UEs for completeness, we consider two cases when UAV-UEs are served by nearest \ac{BS} and \ac{CoMP} transmissions. 
%To do that, similar to \cite{6477064,bettstetter2004stochastic,1233531,1624340}, we study the spatial node distribution, however, in the 3D. Nearest Serving \ac{BS}
% 
\vspace{-0.4 cm}			% be the first two successive waypoints
\subsection{Handover Rate and Handover Probability for Nearest Association}	% the origin in $\R^2$
\label{ho-rate-near}
Assume that a mobile UAV-UE is located at $\boldsymbol{X}_{n-1}$ and let $\boldsymbol{X}_{n-1}$ and $\boldsymbol{X}_{n}$ be two arbitrary successive waypoints. The handover rate is defined as the expected number of handovers per unit time. Hence, inspired from \cite{6477064}, we can compute the handover rate as follows. We first condition on an arbitrary position of the mobile UAV-UE $\boldsymbol{X}_{n}=\boldsymbol{x}_{n}$, and a given realization of the Poisson-Voronoi tessellation $\Phi_b$. Subsequently, the number of handovers will be equal to the number of intersections of the UAV-UE trajectory and the boundary of the Poisson-Voronoi tessellation. Then, by averaging over the spatial distribution of $\boldsymbol{X}_{n}$ and the distribution of Poisson-Voronoi tessellation, we derive the expected number of handovers. Alternatively, we notice that the number of handovers is equivalent to the number of intersections of the Poisson-Voronoi tessellation and the \emph{horizontal projection of the segment $\big[\boldsymbol{X}_{n-1},\boldsymbol{X}_{n}\big]$ on the $x$-$y$ plane}. Therefore, following \cite{6477064,8048668,tabassum2018mobility}, the expected number of handovers during one movement epoch will be: $\Eb\big[N\big]=\frac{2}{\pi}\sqrt{\frac{\lambda_b}{\mu}}$. The handover rate is then the ratio of the expected number of handovers during one movement $\Eb\big[N\big]$ to the mean time of one transition movement $\Eb\big[T\big]$. Since we have
$\Eb[T] = \Eb[\frac{U}{V}] = \frac{\Eb[U]}{\bar{\nu}} 
= \frac{\Omega(\mu,\hbar)}{2\sqrt{\mu }\bar{\nu} }$, where $\Eb[U]=\frac{\Omega(\mu,\hbar)}{2\sqrt{\mu }}$, then, the handover rate will be: 
\begin{align}
\label{ho-rate}
H = \frac{\Eb[N]}{\Eb[T]} = \frac{2}{\pi}\sqrt{\frac{\lambda_b}{\mu}} \Bigg/ \frac{\Omega(\mu,\hbar)}{2\sqrt{\mu }\bar{\nu} } = \frac{4 \bar{\nu} \sqrt{\lambda_b}}{\pi \Omega(\mu,\hbar)}.
\end{align}
\begin{remark}
{\rm Unlike the handover rate for \ac{2D} \ac{RWP} \cite{6477064,8048668,tabassum2018mobility}, $H$ in (\ref{ho-rate}) is a function of the mobility parameter $\mu$ through $\Omega(\mu,\hbar)$. This captures the fact that, in the case of an UAV-UE, since each stochastically generated horizontal displacement is accompanied with a vertical one, the handover rate depends on $\mu$ that affects the vertical displacement switch rates.}
\end{remark}

% \cite{tabassum2018mobility} and the references therein
Next, to characterize the coverage probability of mobile UAV-UEs, we use the concept of \emph{handover probability}. Similar to \cite{tabassum2018mobility} and \cite{7006787}, given the current location of a mobile UAV-UE, the handover probability is defined as the probability that there exists a \ac{BS} closer than the serving \ac{BS} after a unit time. From Fig.~\ref{prob-ho-near}, for two arbitrary consecutive waypoints $\boldsymbol{X}_{n-1}=(\varrho_{n-1},\phi_{n-1},z_{n-1})$ and  $\boldsymbol{X}_{n}=(\varrho_{n},\phi_{n},z_{n})$, the horizontal speed of the UAV-UE from waypoint $\boldsymbol{X}_{n-1}$ to waypoint $\boldsymbol{X}_{n}$ is $\nu_h=\bar{\nu}{\rm cos}(\varphi_n)$, where $\varphi_n= {\rm arccos}\Big(\frac{\varrho_n}{\sqrt{\varrho_n^2 + (z_n-z_{n-1})^2}}\Big)$. It is also assumed that the angle $\phi_{n}$ is taken with respect to the direction of connection as shown in Fig.~\ref{prob-ho-near}. Define $\boldsymbol{q}_{n-1}$ and $\boldsymbol{q}_n$ in $\R^2$ as the horizontal projections of $\boldsymbol{X}_{n-1}$ and the location reached by the UAV-UE after a unit time, respectively. Fig.~\ref{prob-ho-near} illustrates that the UAV-UE is first associated with its nearest \ac{BS} located at $\boldsymbol{q}_0$, i.e., there are no \acp{BS} in the ball of radius $r_0=\lVert \boldsymbol{q}_{n-1}-\boldsymbol{q}_0\rVert$ centered at $\boldsymbol{q}_{n-1}$. Using the law of cosines, $\boldsymbol{q}_{n}$ is at distance $R=\sqrt{r_0^2+(\bar{\nu}{\rm cos}(\varphi_n))^2+2r_0(\bar{\nu}{\rm cos}(\varphi_n)) cos(\phi_n)}$ from the \ac{BS} located at $\boldsymbol{q}_0$.\footnote{Since the UAV-UE starts from waypoint $\boldsymbol{X}_{n-1}$, we assume that it does not change its direction in a time shorter than the unit time. Hence, $q_n$ is assumed to be within the segment $[\boldsymbol{X}_{n-1},\boldsymbol{X}_{n}]$ in Fig.~\ref{prob-ho-near-comp}.}
% \red{Although, it has non-zero probability of occurrence that the movement direction changes in a time shorter than unit time, bases on the assumed speeds, $\mu$, and $\hbar$, this probability is always $\leq 0.01$, hence ignored for tractability.}

\begin{figure}[!t]	
\vspace{-0.8 cm}
    \centering
    \subfigure[Nearest association handover scenario]
    {
        \includegraphics[width=2.5in]{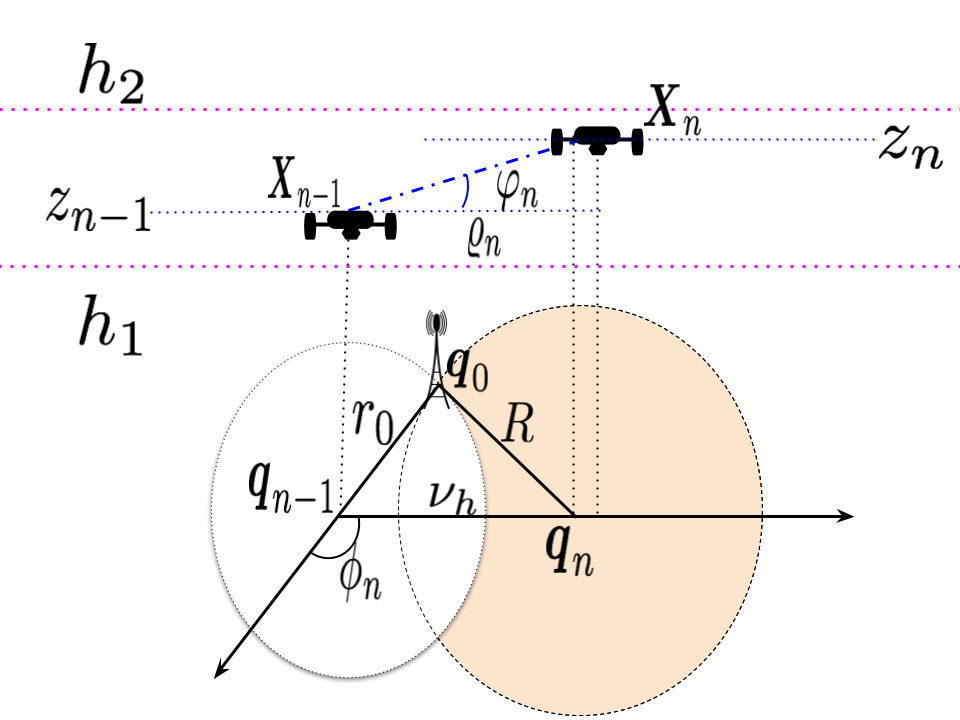}	%41	%letter	prob_HO	
        \label{prob-ho-near}
    }
    \subfigure[Inter-CoMP handover scenario]
    % The probability of handover is based on the relation between $r_0$, $\bar{\nu}$, and $R$.
    {
        \includegraphics[width=2.5in]{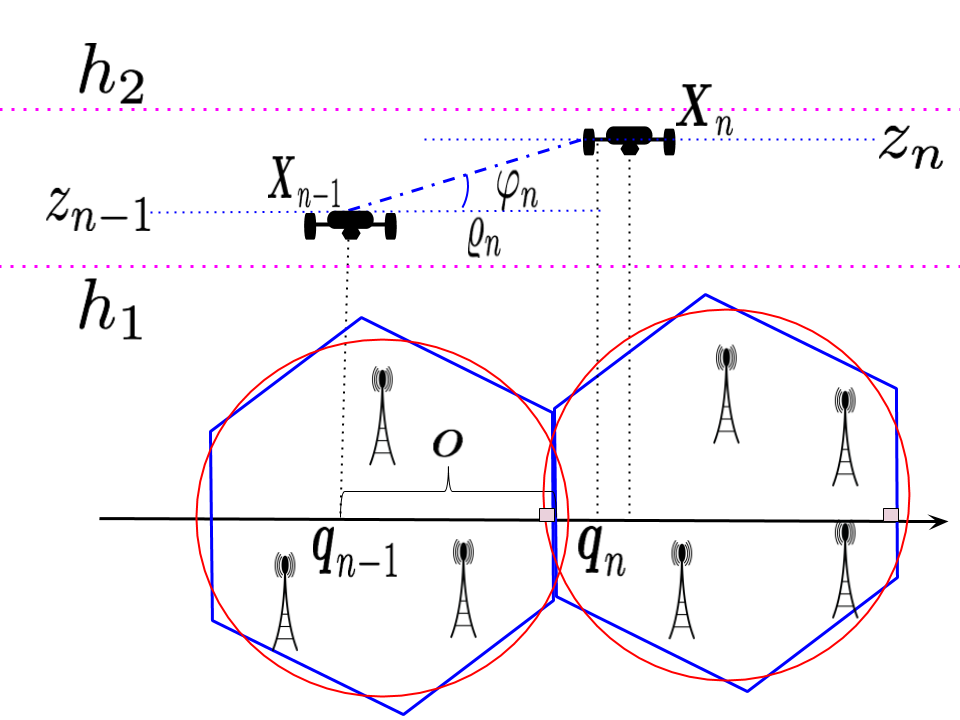}     %prob_HO10
        \label{prob-ho-comp}
    }
    \caption{The probability of handover is computed based on the network geometry.}
    % \red{what happens if it cuts $\varrho_n$ in less than 1 sec! - take care that also $\phi_n=0$}		
    \label{prob-ho-near-comp}
    \vspace{-0.6 cm}
\end{figure}

%.   \cite{tabassum2018mobility} and \cite{7006787}
The handover occurs only if another \ac{BS} becomes closer to $\boldsymbol{q}_{n}$ than the serving \ac{BS} located at $\boldsymbol{q}_0$, i.e., when there is at least one \ac{BS} in the shaded area in Fig.~\ref{prob-ho-near}. Therefore, given $\{r_0,z_{n-1}, z_{n}, \varrho_n,\phi_n\}$, the conditional probability of handover is $\Pb(H|r_0,z_{n-1},z_{n},\varrho_n,\phi_n)$
\begin{align}
&= \Pb\Big(\mathcal{B}(\boldsymbol{q}_{n},R)\setminus \mathcal{B}(\boldsymbol{q}_{n-1},r_0)>0|r_0,z_{n-1},z_{n},\varrho_n,\phi_n\Big)
\overset{(a)}{=} 1 - e^{-\lambda_b\big|\mathcal{B}(\boldsymbol{q}_{n},R)\setminus \mathcal{B}(\boldsymbol{q}_{n-1},r_0)\big|}
\nonumber \\
&
 = 1 - e^{-\pi\lambda_b(R^2- r_0^2)}
\overset{}{=} 1 - e^{-\pi\lambda_b\big(r_0^2+(\bar{\nu}{\rm cos}(\varphi_n))^2+2r_0\bar{\nu}{\rm cos}(\varphi_n) cos(\phi_n)- r_0^2\big)},
% \nonumber 
\end{align}
where (a) follows from the void probability of PPP. $\mathcal{B}(\boldsymbol{q}_{n},R)$ represents the ball with radius $R$ centered at $\boldsymbol{q}_{n}$ and $\mathcal{B}(\boldsymbol{q}_{n-1}, r_0)$ is excluded from $\mathcal{B}(\boldsymbol{q}_{n},R)$ since the \ac{BS} located at $\boldsymbol{q}_0$ is the nearest \ac{BS} to $\boldsymbol{q}_{n-1}$. Finally, averaging over $Z_{n-1},Z_{n},\rho_n$, and $\phi_n$, where $\phi_n\sim\U(0,2\pi)$, we get 
% similar to \cite{tabassum2018mobility}
\begin{align}								% 2\pi \lambda_b\int_{r_0=0}^{\infty}    \dd{r_0}
\Pb(H|r_0) &=  1 - \Eb_{\rho_n,Z_n,Z_{n-1},\phi_n}		
 % r_0 e^{-\pi\lambda_b r_0^2}
 \Big[  
 e^{-\pi\lambda_b\big((\bar{\nu}{\rm cos}(\varphi_n))^2+2r_0(\bar{\nu}{\rm cos}(\varphi_n)) cos(\phi_n)\big)} \Big].
%int_{0}^{\infty}\Pb(H|r_0,z_{n-1},z_{n},\varrho_n,\phi_n) 2\pi \lambda_b r_0 e^{-\pi\lambda_b r_0^2},
\end{align}

%For this reason it is not unreasonable to assume exact centering.
%It should be noted that this scheme is also a reasonable app
%It is reasonable to neglect the shadowing effect

For the special case in which the UAV-UE moves radially away from the serving \ac{BS}, i.e., $\phi_n=0$, next, we obtain a tractable yet accurate \ac{UB}  on the handover probability. This assumption is reasonable, particularly, if the UAV-UE follows a horizontally-direct path subject only to vertical fluctuations due to mission, environmental, and atmospheric conditions. 
 \begin{lemma} 
 \label{near-HO}
 An UB on the conditional probability of handover is given by
 \begin{align}
 \label{ub-on-ho-prob}
 \Pb(H|r_0)  &\overset{}{=} 1 - e^{-\frac{2\lambda_b r_0 \bar{\nu}}{\sqrt{\pi } \mu  \hbar^2} \psi(\mu,\hbar)  }									
     e^{-\pi\lambda_b \zeta(\mu,\hbar)  },			% \pi \mu \bar{\nu}^2
     \end{align} 
where $\psi(\mu,\hbar)= \pi  \hbar^2 \mu  G_{2,3}^{2,2}\Big(\hbar^2 \pi  \mu \Big|
\begin{array}{c} \frac{1}{2},\frac{1}{2} \\ 0,1,-\frac{1}{2} \\ \end{array}
\Big)-G_{2,3}^{2,2}\Big(\hbar^2 \pi  \mu \Big| \begin{array}{c}
 1,\frac{3}{2} \\ 1,2,0 \\ \end{array}\Big)$, $G_{p,q}^{m,n}$ denotes the Meijer $G$ function, defined as
\begin{align}
\label{G-fn}
 G_{p,q}^{m,n}  = \Big(x \Big| \begin{array}{c} a_1,\dots,a_p \\
 b_1,\dots,b_q \\ 
\end{array}
\Big)
= \frac{1}{2\pi i} \int \frac{\prod_{j=1}^{m}\Gamma(b_j+s) \prod_{j=1}^{n}\Gamma(1 - a_j+s)}{\prod_{j=n+1}^{p}\Gamma(a_j+s) \prod_{j=m+1}^{q}\Gamma(1 - b_j +s)} x^s \dd{s},
\end{align}
and 
%\begin{align}
%\label{zeta-fn}
$\zeta(\mu,\hbar) =  \bar{\nu}^2\Big(1   - \frac{2\pi \mu }{\hbar^2}  \int_{0}^{\hbar}(\hbar-p)  p^2 e^{\pi  \mu  p^2}  \Gamma \left(0,\pi  p^2 \mu \right)    \dd{p}$ \Big). 
%\end{align} 
  \end{lemma}
\begin{proof}
Please see Appendix \ref{proof:near-HO}.
\end{proof}
From (\ref{ub-on-ho-prob}), it is intuitive  to see that $\Pb(H|r_0) $ increases with $\bar{\nu}$ and $\lambda_b$ because there will be a higher probability of handover when the UAV-UE velocity is higher, and the network is denser. Moreover, $\Pb(H|r_0)$ decreases as the term $\mu\hbar^2$ increases. This reveals important insights on the effect of the altitude difference $\hbar$ and the density $\mu$. Particularly, the handover probability decreases when  the UAV-UE jointly has higher direction switch rates (higher $\mu$) and larger altitude difference $\hbar$. Next, we  obtain the handover rate for UAV-UEs under CoMP transmissions.

% We need insights from Lemma and why we did it.
%To obtain insight, consider the constant velocity case which yields the following result. 
%It is quite insightful to see that
%to provide further insight into the impact of access patterns 
% are important as they lead to simple and insightful expressions for the ASEP
% reveal important design insights.
%they are insightful because they reflect the worst-case network performance.
%Under certain plausible scenarios the derived expression is in closed form and provides insight into system design.
%One important insight from the analysis is that

 % For analytical convenience, we first approximate the cluster area to a hexagonal of the same area, denoted as $A$.
 % For the CoMP scenario, w 
 \vspace{-0.5 cm}
\subsection{Inter-CoMP Handover Rate and Handover Probability}
We define the number of handovers $\Eb[N]$ as the number of intersections of the horizontal projection of the segment $\big[\boldsymbol{X}_{n-1},\boldsymbol{X}_{n}\big]$ and the boundaries of disjoint clusters whose inter-cluster center distance is $2R_h$, as discussed in Section \ref{system-model}. The hexagonal cell has six sides of length $l=\frac{2R_h}{\sqrt{3}}$. Following the Buffon's needle approach for hexagonal cells \cite{tabassum2018mobility}, we next obtain the inter-CoMP handover rate.
\begin{proposition}
\label{inter-CoMP-HO}
The inter-CoMP handover rate for a network of disjoint clusters whose inter-cluster center distance is $2R_h$ is given by 
\begin{align}
H= \frac{\Eb[N]}{\Eb[T]} 
= 2
\frac{\pi  \hbar^2 \mu  G_{2,3}^{2,2}\Big(\hbar^2 \pi  \mu \big|
\begin{array}{c}
 \frac{1}{2},\frac{1}{2} \\
 0,1,-\frac{1}{2} \\
\end{array}
\Big)-G_{2,3}^{2,2}\Big(\hbar^2 \pi  \mu \big|
\begin{array}{c}
 1,\frac{3}{2} \\
 1,2,0 \\ 
\end{array} \Big)}{\pi \sqrt{\pi} R_h\hbar^2 \mu } \bar{\nu}.
\end{align}
\end{proposition}
\begin{proof}
Please see Appendix \ref{proof:inter-CoMP-HO}.
\end{proof}

%%%%%%%%%%%%%%%%%%%%%%%%%%%%%%%%%%%%%%%%%%%%%%%%%%%%%%%%%%%%%%%%%%%%%%%%%%%%%%%%%%%%
\begin{figure}[!t]	
\vspace{-1.0 cm}
    \centering
    \subfigure[Nearest association scenario]
    {
        \includegraphics[width=2.7in]{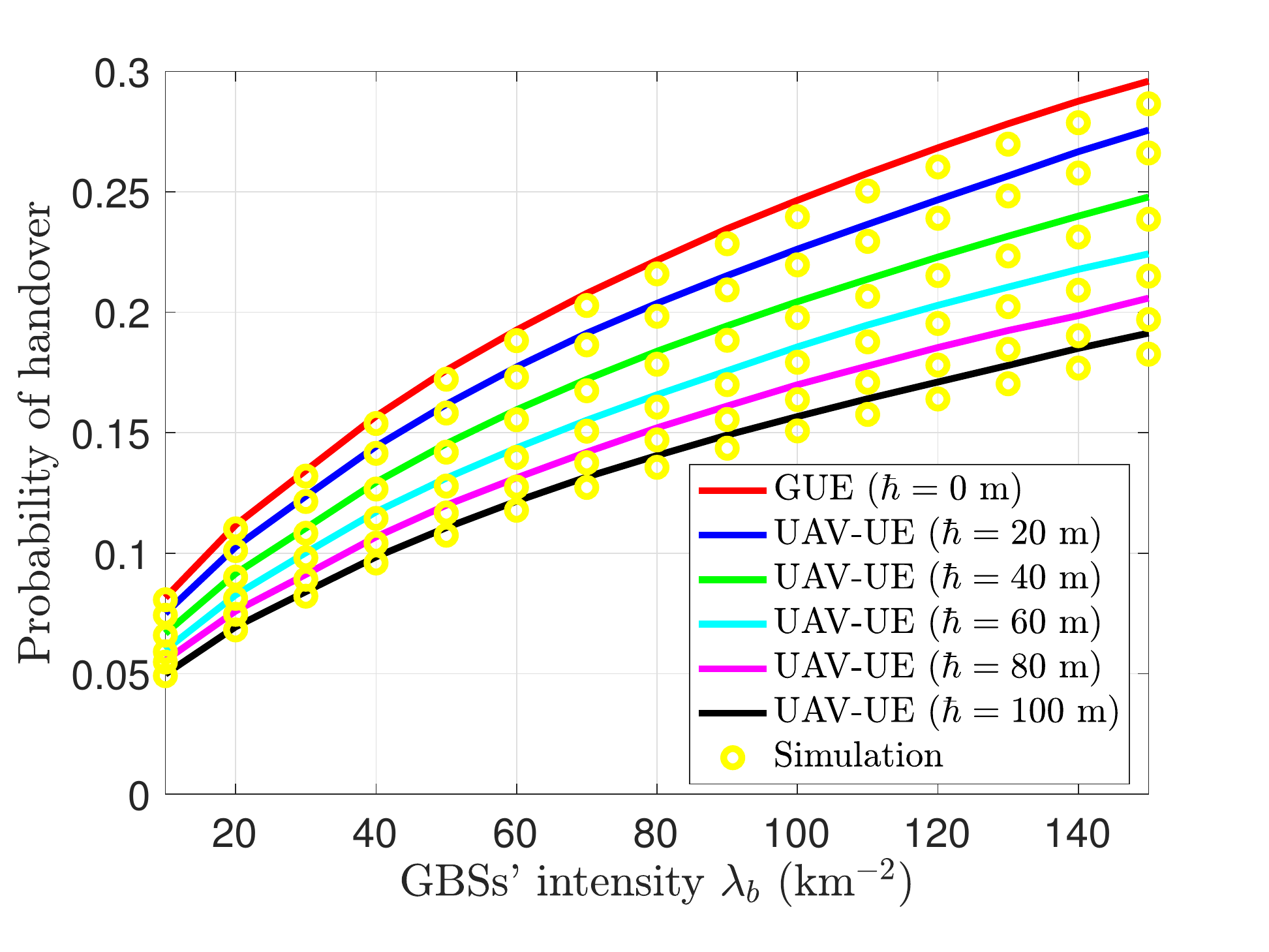}	%41	%letter	prob_HO	
        \label{prob-of-ho-near}
    }
    \subfigure[CoMP transmission scenario]
    % The probability of handover is based on the relation between $r_0$, $\bar{\nu}$, and $R$.
    {
        \includegraphics[width=2.7in]{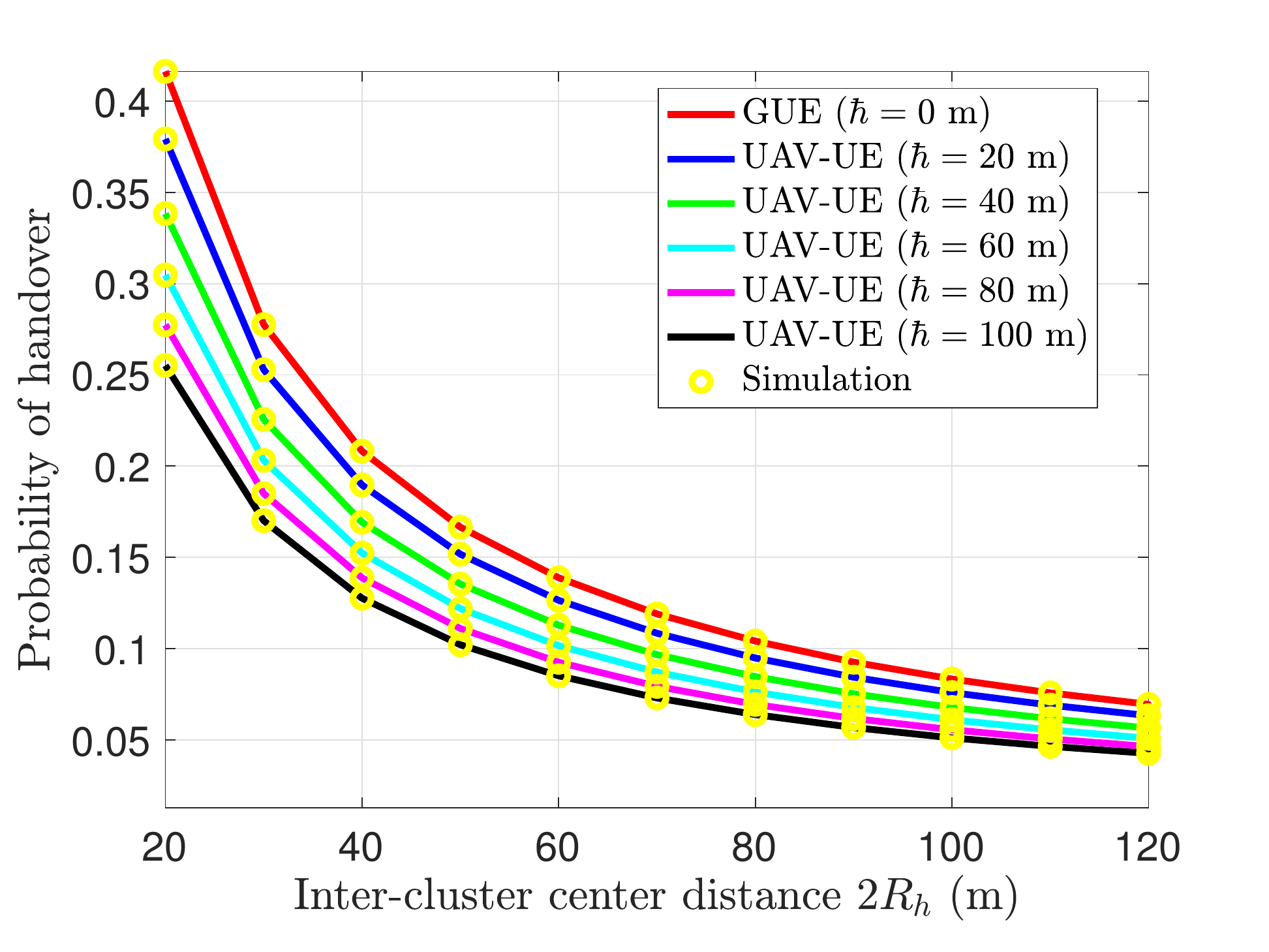}     %prob_HO10
        \label{prob-of-ho-comp}
    }
    \caption{The probability of handover is plotted versus network parameters for  nearest association and \ac{CoMP} transmission schemes ($\bar{\nu}=\SI{30}{kmh}$, $\mu=\SI{300}{km^{-2}}$, $h_1=\SI{100}{m}$).}
    % \red{what happens if it cuts $\varrho_n$ in less than 1 sec! - take care that also $\phi_n=0$}		
    \label{prob-of-ho-near-comp}
    \vspace{-0.8 cm}
\end{figure}
%%%%%%%%%%%%%%%%%%%%%%%%%%%%%%%%%%%%%%%%%%%%%%%%%%%%%%%%%%%%%%%%%%%%%%%%%%%%%%%%%%%%
 %\magenta{This is also practical because it represents scenarios where...what?}. %Such simplicity reveals several insights into the performance and effectively captures the aspects of our proposed mobility model.  
% From Fig.~\ref{prob-ho-comp}, given the radial direction of movement, the inter-CoMP handover probability represents the probability that the horizontal projection of the segment $\big[\boldsymbol{X}_{n-1},\boldsymbol{X}_{n}\big]$ is larger than the distance $o$ to the cluster side, formally stated as		, i.e., $\phi_n=0$.
% and obtain insightful expressions 
% In fact, the coverage probability of a mobile UAV-UE depends on its handover probability.
%  this assumption is practical. 

We now characterize the UAV-UE inter-CoMP handover probability. To keep the analysis simple, we consider a special case in which the UAV-UE moves perpendicularly to the inter-cluster boundaries, as shown in Fig.~\ref{prob-ho-comp}. As discussed in Section \ref{ho-rate-near}, this is a practical assumption for a UAV-UE that follows a horizontally straight path subject only to vertical fluctuations. Moreover, since these boundaries represent virtual borders between disjoint clusters, the assumption that such boundaries are in a direction perpendicular to the UAV-UE trajectory is quite reasonable. Hence, the UAV-UE moves by a horizontal distance $\bar{\nu}{\rm cos}(\varphi_n)$ in a unit of time in a direction perpendicular to the inter-cluster boundaries. A handover occurs if this travelled horizontal distance is larger than the distance $o$ to the cluster side, formally stated as
 \begin{align}
\Pb(H|o)&=  \Pb\big(\frac{\bar{\nu}\varrho_n}{\sqrt{\varrho_n^2 + (z_n-z_{n-1})^2}} >o\big) 
=  \Pb\big( \varrho_n >\frac{o (z_n-z_{n-1})}{\sqrt{\bar{\nu}^2 - o^2}}\big) 
 \nonumber \\
 \label{sol-intt}
 & = \Eb_{Z_n,Z_{n-1}} \int_{\varrho_n=\frac{o (z_n-z_{n-1})}{\sqrt{\bar{\nu}^2 - o^2}}}^{\infty} 2\pi\mu\varrho_n e^{-\pi\mu\varrho_n^2}  \dd{\varrho_n}
%  \\
\end{align}
\begin{align}
 &\overset{(a)}{=} \Eb_{Z_n,Z_{n-1}}  e^{-\pi \mu \big(\frac{o^2 (z_n-z_{n-1})^2}{\bar{\nu}^2 - o^2}\big) }        
\overset{(b)}{=} \Eb_{p}  e^{-\pi\mu p^2 \frac{o^2}{\bar{\nu}^2 - o^2} }
   \nonumber \\
 &\overset{(c)}{=} \frac{1}{\hbar^2} 
\frac{ \hbar \sqrt{\frac{\mu o^2}{\bar{\nu}^2 - o^2}} {\rm erf}\left(\sqrt{\pi } \hbar \sqrt{\frac{o^2}{\bar{\nu}^2 - o^2}} \sqrt{\mu }\right)+\frac{e^{-\pi\mu\hbar^2 \frac{o^2}{\bar{\nu}^2 - o^2} } -1}{\pi }}{\frac{\mu o^2}{\bar{\nu}^2 - o^2}  }
 \\
 &=  \frac{1}{\hbar}
\frac{ {\rm erf}\left(\sqrt{\pi } \hbar \sqrt{\frac{\mu o^2}{\bar{\nu}^2 - o^2}}\right)}{ \sqrt{\frac{\mu o^2}{\bar{\nu}^2 - o^2}} } + 
 \frac{1}{\pi \hbar^2}
\frac{ e^{-\pi\mu\hbar^2 \frac{o^2}{\bar{\nu}^2 - o^2} }-1}{\frac{\mu o^2}{\bar{\nu}^2 - o^2} },
\end{align}
where (a) follows from solving the integral of (\ref{sol-intt}), (b) follows from change of variables $p=z_n-z_{n-1}$, with $f_P(p)=\frac{\hbar-|p|}{\hbar^2}, -\hbar\leq p\leq\hbar$, and (c) follows from taking the expectation \ac{w.r.t.} $p$. Recall that $O$ is a \ac{RV} that models the distance between the UAV-UE and the inter-cluster boundaries. Averaging over $O$ given that $f_O(o)=\frac{1}{2R_h}, 0<o<2R_h$, we get $\Pb(H)$. 
\black{However, we observe that if $o>\bar{\nu}$, the handover probability will be zero since the UAV-UE can not travel the distance $o$ in a unit of time, hence, we have}
\begin{align}
\label{ho-prob-comp}
\Pb(H)&\overset{}{=} 
\frac{1}{2R_h\hbar} \int_{o=0}^{\bar{\nu}}
\frac{ {\rm erf}\left(\sqrt{\pi } \hbar \sqrt{\frac{\mu o^2}{\bar{\nu}^2 - o^2}}\right)}{ \sqrt{\frac{\mu o^2}{\bar{\nu}^2 - o^2}} } + 
 \frac{1}{2\pi R_h \hbar^2} \int_{o=0}^{\bar{\nu}}
\frac{ e^{-\pi\mu\hbar^2 \frac{o^2}{\bar{\nu}^2 - o^2} } -1}{\frac{\mu o^2}{\bar{\nu}^2 - o^2} }.
\end{align}

Fig.~\ref{prob-of-ho-near-comp} verifies the accuracy  of the obtained handover probabilities. Fig.~\ref{prob-of-ho-near} presents the handover probability versus BSs' intensity $\lambda_b$ under the nearest association scheme. The figure shows that the obtained \ac{UB} in (\ref{ub-on-ho-prob}) is quite tight. It is also noted  that as long as the UAV-UE has frequent vertical movements, i.e., larger $\hbar$, the handover probability is lower since the effective horizontal travelled distance becomes shorter. The handover probability also monotonically increases with $\lambda_b$ since a higher rate of handover occurs for denser networks. Fig.~\ref{prob-of-ho-comp} shows the inter-CoMP handover probability versus the inter-cluster center distance $2R_h$. The handover probability monotonically decreases with $R_h$ since a lower rate of handover is anticipated when the cluster size increases. 
% Having characterized the \ac{3D} RWP mobility model
Next, we will use our proposed RWP model to obtain the coverage probability of \ac{3D} mobile UAV-UEs under the nearest association and CoMP transmission schemes.

\vspace{-0.2 cm}
\section{Coverage Probability of Mobile UAV-UEs} 		% - \red{radial}
\label{3D-RWP-Cov}
\vspace{-0.1 cm}
Next, we will use the obtained handover probabilities in (\ref{ub-on-ho-prob}) and (\ref{ho-prob-comp}) to quantify the coverage probability of mobile UAV-UEs under the nearest association and CoMP transmissions, respectively. 
It is worth highlighting that (\ref{cov-prob-theory}) represents the probability that a static UAV-UE is in coverage with neither mobility nor handover considered. However, mobile UAV-UEs are susceptible to frequent handovers that would negatively impact their performance. For instance, handover typically results in dropped connections    and causes longer service delays. In fact, higher handover rates lead to a higher risk of \ac{QoS} degradation. % dropped connections connection failure

To account for the user mobility, similar to \cite{tabassum2018mobility,7006787,7827020}, we consider a linear function that reflects the cost of handovers. Under this model, the UAV-UE coverage probability can be  defined as:
\begin{align}
\label{mob-cov0}
P_c(\bar{\nu},\mu,\beta)  = \Pb(\Upsilon \geq \vartheta, \bar{H})+(1-\beta) \Pb(\Upsilon \geq \vartheta,H), 
\end{align}
where the first term represents the probability that the UAV-UE is in coverage and no handover occurring. Besides, the second term is the probability that the UAV-UE is in coverage and handover occurs penalized by a handover cost, where $\beta\in[0, 1]$ represents the probability of connection failure due to handover. The coefficient $\beta$, in effect, measures the system sensitivity to handovers, which highly depends on the hysteresis margin and ping-pong rate \cite{8048668,tabassum2018mobility,7006787,7827020}. Our goal is to obtain the coverage probability of a mobile UAV-UE for a given handover penalty $\beta$ \cite{tabassum2018mobility}. After some manipulations, we can rewrite (\ref{mob-cov0}) as		
% \red{$\phi$}	% \red{$r$,$\phi$}
% We next incorporate the UAV-UE mobility in the coverage analysis as follows.
\begin{align}
\label{mob-cov1}
P_c(\bar{\nu},\mu,\beta) = (1-\beta)\Pb(\Upsilon \geq \vartheta|r_0) +\beta \Pb(\Upsilon \geq \vartheta,\bar{H}|r_0).
\end{align}
%\begin{theorem}
%The probability of coverage $P_c(\bar{\nu},\lambda,\vartheta,\alpha)$ for a typical UAV-UEs moving a distance $\bar{\nu}$ in a unit of time in a network with \acp{BS} distributed according to a homogeneous \ac{PPP} with density $\mu$ is given by:
%%
%From (11), the probability of coverage conditioned on $r$ and $\varphi$ is given by:
%\begin{align}
%P_c  = \Pb(\Upsilon \geq \vartheta, \bar{H}|r,\varphi)+(1-\beta) \Pb(\Upsilon \geq \vartheta,H|r, \varphi),
%\end{align}
%where 
%\begin{align}
%\Pb(\Upsilon \geq \vartheta,H|r, \varphi) &= \Pb(\Upsilon \geq \vartheta|r)  Pb(H|r, \varphi)  
%\nonumber \\
%&= \Pb(\Upsilon \geq \vartheta|r) \big(1 -  Pb(\bar{H}|r, \varphi)  \big)
%\nonumber \\
%&= \Pb(\Upsilon \geq \vartheta|r)  -  \Pb(\Upsilon \geq \vartheta|r)  Pb(\bar{H}|r, \varphi)  \big)
%\nonumber \\
%&= \Pb(\Upsilon \geq \vartheta|r)  -  \Pb(\Upsilon \geq \vartheta,\bar{H}|r, \varphi)
%\end{align}
%Hence, we have \red{Check Jensen and Young inequalities please !}
%\begin{align}
%P_c  &= \Pb(\Upsilon \geq \vartheta, \bar{H}|r,\varphi)+(1-\beta) \Big(\Pb(\Upsilon \geq \vartheta|r)  -  \Pb(\Upsilon \geq \vartheta,\bar{H}|r, \varphi) \Big),
%\nonumber \\
%& = (1-\beta)\Pb(\Upsilon \geq \vartheta|r) +\beta \Pb(\Upsilon \geq \vartheta,\bar{H}|r, \varphi)
%\end{align}
%\end{theorem}
To obtain $P_c(\bar{\nu},\mu,\beta)$, we first need to calculate the statistical distribution of the UAV-UE altitude for our proposed \ac{3D} mobility model. As shown in Fig.~\ref{shaded-area-2}, the \ac{3D}   mobility model defines the vertical movement of the UAV-UE in a finite region $[h_1,h_2]$, \emph{referred to as vertical \ac{1D} RWP} as in \cite{8671460}. Initially, at time instant $t_0$, the UAV-UE is at an arbitrary altitude $h_0$ selected uniformly from the interval  $[h_1,h_2]$. Then, at next time epoch $t_1$, this UAV-UE at $h_0$ selects a new random waypoint $h_1$ uniformly in $[h_1,h_2]$, and moves towards it (along with the spatial movement characterized by $f_{\rho_n}(\varrho_n)$). %\footnote{It is worth highlighting that by the time UAV-UE reaches the new altitude, it should have travelled a horizontal distance $\varrho$ generated from the PDF $f_{\rho}(\varrho)$.}   
 Once the UAV-UE reaches $h_1$, it repeats the same procedure to find the next destination altitude and so on. After a long running time, the steady-state  altitude distribution converges to a \emph{nonuniform distribution} $F_{Z_{\infty}}(z_{\infty})$ \cite{1233531}, where $Z_\infty$ is a \ac{RV} representing the steady state vertical location of the UAV-UE. Note that random waypoints refer to the altitude of a UAV-UE at each time epoch, which is uniformly-distributed in $[h_1,h_2]$, while vertical transitions are the differences in the UAV-UE altitude throughout its trajectory. While the random waypoints are independent and uniformly-distributed by definition, the  random lengths of vertical transitions are not statistically independent. This is because the endpoint of one movement epoch is the starting point of the next epoch. 
 %
% A node moves according to the \ac{RWP} model on a line segment  $[h_1,h_2]$.      where $Z_\infty\in[h_1,h_2]$
% We now derive the asymptotically stationary distribution $F_{Z_\infty}(z_\infty)$ generated by the vertical \ac{1D} \ac{RWP}.  
In \cite{1233531}, it is shown that $F_{Z_{\infty}}(z_{\infty}) = \frac{\Eb[L_{z_\infty}]}{\Eb[L]}$, where $L_{z_\infty}$ and $L$ denote the length $\lVert z_\infty -h_1 \rVert$, and the entire movement length at any given epoch, respectively. From \cite{1233531}, we have $\Eb[L] = \frac{\hbar}{3}$ and $\Eb[L_{z_\infty}]$ can be similarly derived from:  
\begin{align}
\Eb[L_{z_\infty}] = \int_{s=h_1}^{h_2} \int_{d=h_1}^{h_2} l_z(s,d) f_S(s) f_D(d) \dd{d} \dd{s},
\end{align}
where $s$ and $d$ refer to the source and destination of a movement, respectively; $f_S(s) = f_D(d)=\frac{1}{\hbar}$, $h_1\leq s,d\leq h_2$, see Fig.~\ref{shaded-area-2}. Because of the symmetry of $s$ and $d$, it is sufficient to restrict the calculation to epochs with $s<d$, and then multiply the result by a factor of 2. A necessary condition for $l_{z_\infty}(s,d)\neq 0$ is that $s\leq z_\infty$. From Fig.~\ref{shaded-area-2}, if $d\leq z_\infty$, we have $l_{z_\infty}(s,d)=d-s$, however, if $d > z_\infty$, we get $l_{z_\infty}(s,d)=z_\infty-s$, which yields 
\begin{align}
\Eb[L_{z_\infty}] &=  \frac{2}{\hbar^2} \int_{s=h_1}^{z_\infty}  \int_{d=s}^{z_\infty} (d-s)  \dd{d} \dd{s} + 
 \frac{2}{\hbar^2} \int_{s=h_1}^{z_\infty}  \int_{d=z_\infty}^{h_2} (z_\infty-s)  \dd{d} \dd{s} 
%\nonumber \\
%&=  2 \frac{1}{\hbar^2} \int_{s=h_1}^{z}  (d-s) \frac{1}{2}(z^2 - s^2 -s) \dd{s} + 
%2 \frac{h_2-z}{\hbar^2} \int_{s=h_1}^{z} (z-s)  \dd{s} 
%&=  2 \frac{1}{\hbar^2} \int_{s=h_1}^{z}  (d-s) \frac{1}{2}(z^2 - s^2 -s) \dd{s} + 
%2 \frac{h_2-z}{\hbar^2} (z-\frac{1}{2}s)  
\nonumber \\
& = \frac{2}{\hbar^2} \big(-\frac{h_1^3}{6} + \frac{h_1^2 h_2}{2} - h_1 h_2 z_\infty + \frac{h_1 z_\infty^2}{2} + \frac{h_2 z_\infty^2}{2} - \frac{z_\infty^3}{3} \big).
\end{align}
Therefore, the \ac{PDF} of $Z_\infty$ is given by
\begin{align}
%\label{vert-cdf}
%F_{Z_\infty}(z_\infty) &= \frac{\Eb[L_{z_\infty}]}{\Eb[L]}  = 
%\frac{6\big(-\frac{h_1^3}{6} + \frac{h_1^2 h_2}{2} - h_1 h_2 z_\infty + \frac{h_1 z_\infty^2}{2} + \frac{h_2 z_\infty^2}{2} - \frac{z_\infty^3}{3} \big)}{\hbar^3},
% \\
\label{f-of-z-new}
f_{Z_\infty}(z_\infty) &= \frac{\partial F_{Z_\infty}(z_\infty)}{\partial z_\infty} = \frac{\partial}{\partial z_{\infty}}\frac{\Eb[L_{z_\infty}]}{\Eb[L]} =  \frac{h_1 z_\infty+h_2 z_\infty -h_1 h_2 - z_\infty^2}{\hbar^3/6}
\quad 
\forall h_1<z_\infty<h_2,
\end{align}
and the corresponding mean is given by $L_{z_{\infty}}=\Eb[Z_{\infty}]= \frac{1}{2\hbar^3}\big(h_2^4 -h_1^4+ 2h_1^3 h_2-2h_1 h_2^3\big)$. Next, we will use the derived \ac{PDF} $f_{Z_\infty}(z_\infty)$, along with the probability of handover from the previous section, to fully characterize $P_c(\bar{\nu},\mu,\beta)$ under the nearest association and \ac{CoMP} transmissions.
%%\[
%%    f_{Z_\infty}(z_\infty) = \frac{\partial F_{Z_\infty}(z_\infty)}{\partial z_\infty} 
%%    =\left\{
%%                \begin{array}{ll}
%%\frac{h_1 z_\infty+h_2 z_\infty -h_1 h_2 - z_\infty^2}{\hbar^3/6}, \quad h_1<z_\infty<h_2,\\ 
%%\\
%%                 0, \quad\quad\quad\quad\quad\quad\quad\quad\quad {\rm otherwise},
%%                \end{array}   
%%              \right.
%%  \]
%\emph{Furthermore, the steady-state distribution is independent of the initial node distribution and the speed $\nu$ \cite{8671460}.} 
%\subsubsection{Aerial User Distance Distribution}
%\black{check its mean and variance}
%\begin{align}
%\Eb[Z_{\infty}] &= -\frac{h_1^4}{2 \hbar^3}+\frac{2h_1^3 h_2}{2\hbar^3}+\frac{h_2^4}{2 \hbar^3}
%-\frac{2h_1 h_2^3}{2\hbar^3}
%\nonumber \\
% &= \frac{1}{2\hbar^3}\big(h_2^4 -h_1^4+ 2h_1^3 h_2 - 2h_1 h_2^3\big),
%\end{align} 
%To keep the analysis simple we assume (for now) that
%we derive a remarkably simple formula for coverage probability
%to simple and insightful expressions 
%Equation (11) gives a closed-form expression for the outage probability in its simplest form 
% As mentioned earlier, these distance distributions will be used to derive simple expressions 
% To keep the analysis simple and derive a remarkably simple formula for coverage probability, we assume that the UAV-UE is moving radially away from the serving \acp{BS}, i.e., $\phi=0$
\subsection{\black{Coverage Probability for Nearest Association}}
\label{cov-near-mob}
Next, we derive the coverage probability of a mobile UAV-UE under the nearest association scheme. Observing (\ref{mob-cov1}), for a given $\beta$, we must compute $\Pb(\Upsilon \geq \vartheta,\bar{H}|r_0)$ to obtain $P_c(\bar{\nu},\mu,\beta)$. The former probability is basically the joint event of being in coverage and no handover occurs. We adopt the tight \ac{UB} on the handover probability obtained in (\ref{ub-on-ho-prob}), where $\Pb(\bar{H},r_0) = 1 - \Pb(H,r_0)$ is the conditional probability of no handover. Unlike static UAV-UEs, under the \ac{3D} \ac{RWP} model, both the altitude of the UAV-UE and the horizontal distance $R_0$ to the nearest BS are \acp{RV}. Since $R_0$ and $Z_\infty$ are two independent \acp{RV}, we have $f_{R_0,Z_\infty}(r_0,z_\infty) = f_{R_0}(r_0)f_{Z_\infty}(z_\infty)$. % 7571101

We assume that the UAV-UE has an arbitrary long trajectory that passes through nearly all \ac{SIR} states. Therefore, the average \ac{SIR} through a randomly selected UAV-UE trajectory is inferred from a stationary PPP analysis. This assumption, which is adopted in \cite{tabassum2018mobility,7827020,7006787} for GUEs, is practically reasonable for mobile UAV-UEs such as flying taxis and delivery drones that typically have sufficiently long trajectories. Given the handover probability in (\ref{ub-on-ho-prob}) and the linear function in (\ref{mob-cov1}), the UAV-UE coverage probability under the nearest association scheme is given below.

%obtained in the next theorem
% \footnote{\blue{Hesham(and other papers - Revise carefully Hesham's paper 2 CoMP part)(Do you still need it with HO probability paper)} - (\red{PDF of $r_0$ since ergodic and stationary PPP.})}
\begin{theorem}
\label{the-2-mob-near}
The coverage probability of a \ac{3D} mobile UAV-UE associated with its nearest \ac{BS} is
\begin{align}
\label{cov-prob-theory01}		
P_c(\bar{\nu},\mu,\beta) &= 2(1-\beta)\pi\lambda_b \times \int_{h_1}^{h_2}  \int_{0}^{\infty} r_0e^{-\pi\lambda_br_0^2} \Pb_{{\rm c}|r_0,z_{\infty}}^l  f_{Z_\infty}(z_\infty)				
\dd{r_0} \dd{z_\infty} +
 \nonumber \\
 %\label{cov-near-mob}
 &  2\beta \pi\lambda_b  e^{-\pi\lambda_b \zeta(\mu,\hbar)  }  \times \int_{h_1}^{h_2}  \int_{0}^{\infty} r_0e^{-\pi\lambda_br_0^2} 
 %e^{-\lambda_b( \pi \bar{V}_{\rho}^2 + 2\pi r_0 \bar{V}_{\rho})} 
 e^{-\frac{2\lambda_b r_0 \bar{\nu}}{\sqrt{\pi } \mu  \hbar^2} \psi(\mu,\hbar)  }     %\Pb(\bar{H},r_0)
 \Pb_{{\rm c}|r_0,z_{\infty}}^l f_{Z_\infty}(z_\infty) \dd{r_0} \dd{z_\infty}, 
\end{align}			% =\Lc_{I|r_0,z_{\infty}}(\varpi) =
where $\Pb_{{\rm c}|r_0,z_{\infty}}^l = \lVert e^{\boldsymbol{T}_{m_l}}\rVert_1$,  $\boldsymbol{T}_{m_l}$ is defined as $\boldsymbol{T}_{K}$ in (\ref{cov-prob-theory}), with $\Omega(\varpi)_{|r_0,z_{\infty}} 
= -2\pi \lambda_b\int_{\nu=r_0}^{\infty}\big(1 - \delta_l\Pb_{l}(\nu) - \delta_n\Pb_{n}(\nu)  \big)\nu\dd{\nu}$, 
$\delta_l= \Big(1 + \frac{\varpi P_t A_l G_s (\nu^2 + z_\infty^2)^{-\alpha_l/2} }{m_l} \Big)^{-m_l}$, 
$\delta_n= \Big(1 + \frac{\varpi P_t A_n G_s (\nu^2 + z_\infty^2)^{-\alpha_n/2} }{m_n} \Big)^{-m_n}$, and 
 $\varpi =\frac{\vartheta m_l}{P_t A_l G_s (r_0^2 + z_\infty^2)^{-\alpha_l/2}}$; $\psi(\mu,\hbar)$ and $\zeta(\mu,\hbar)$ are given in Lemma \ref{near-HO}.
 %Lemma \ref{near-HO}.
\end{theorem}
\begin{proof}
%\begin{align}
%\label{cov-prob-theory01}		
%P_c(\bar{\nu},\mu,\beta) &= 2(1-\beta)\pi\lambda_b \times \int_{h_1}^{h_2}  \int_{0}^{\infty} \Pb_{{\rm c}|r_0,z_{\infty}}^l r_0e^{-\pi\lambda_br_0^2} f_{Z_\infty}(z_\infty)				
%\dd{r_0} \dd{z_\infty} +
% \nonumber \\
% &  2\beta \pi\lambda_b  \times
% \Eb_{\rho_n,Z_n,Z_{n-1}}
%  \int_{h_1}^{h_2}  \int_{0}^{\infty} r_0e^{-\pi\lambda_br_0^2} 
% e^{
% -\lambda_b \big(\frac{\pi\bar{\nu}^2 \varrho_n^2}{\varrho_n^2 + (z_n-z_{n-1})^2} +
%  \frac{2\pi r_0 \bar{\nu} \varrho_n}{\sqrt{\varrho_n^2 + (z_n-z_{n-1})^2}}\big)} \Pb_{{\rm c}|r_0,z_{\infty}}^l f_{Z_\infty}(z_\infty) \dd{r_0} \dd{z_\infty}, 
%\end{align}	
%The proof follows on the same lines as Theorem 1 and is hence skipped.
%the proof proceeds in a similar manner to Theorem 4
%The proof follows directly from (\ref{mob-cov1}) and Corollary \ref{corr-near}, with the second term in 
% (\ref{cov-prob-theory01}) representing $\Pb(\Upsilon \geq \vartheta,\bar{H}|r, \phi)$.
%
The first term in (\ref{cov-prob-theory01}) is obtained directly from (\ref{mob-cov1}) and Corollary \ref{corr-near},  where the UAV-UE altitude $h_d$ is replaced with the \ac{RV} $z_{\infty}$ whose \ac{PDF} is given in (\ref{f-of-z-new}). Additionally, the second term in (\ref{cov-prob-theory01}) represents the joint event of no handover and being in coverage, which is computed based on $\Pb(H|r_0, \phi)$ in (\ref{ub-on-ho-prob}). 
\end{proof}

It is clear from (\ref{cov-prob-theory01}) that, if $\beta=1$, the first term vanishes and the UAV-UE will be in coverage only if there is no handover associated with its mobility. This is because the handover will always cause connection failure. Moreover, since it is hard to directly obtain insights from (\ref{cov-prob-theory01}) on the effect of the altitude $z_{\infty}$ and the altitude difference $\hbar$, several numerical results based on (\ref{cov-prob-theory01}) will be shown in Section \ref{num-result} to provide key practical insights. Next, we similarly derive the coverage probability of a mobile UAV-UE under \ac{CoMP} transmission.

\vspace{-0.3 cm}
\subsection{Coverage Probability for CoMP Transmission}
Similar to Section \ref{cov-near-mob}, we employ the handover probability in (\ref{ho-prob-comp}) and the linear function in (\ref{mob-cov1}) to obtain the coverage probability under CoMP transmission. The probability of inter-cluster handover $\Pb(H)$ is derived in (\ref{ho-prob-comp}) assuming that the UAV-UE moves perpendicular to the cluster boundaries. To compute $\Pb(\Upsilon \geq \vartheta,\bar{H})$ in (\ref{mob-cov1}), the joint \ac{PDF} of the serving distances needs to be characterized given the random location of the UAV-UE along its trajectory. However, for tractability, we consider the joint serving distances when the UAV-UE horizontal projection is at the cluster center. Therefore, the obtained performance can be seen as an \ac{UB} on the performance of a randomly located UAV-UE. This assumption is in line with prior work \cite{7880694} and the analysis for static UAV-UEs, where we sought an \ac{UB} on the coverage probability. 
% due to its tractability			and also adopted in the literature \cite{7880694}
% in Section \ref{cov-prob-static}

Since $\boldsymbol{R}_{\kappa}= [R_1, \dots, R_{\kappa}]$ and $Z_\infty$ are independent \acp{RV}, their joint \ac{PDF} is $f_{\boldsymbol{R}_{\kappa},Z_\infty}(\boldsymbol{r}_{\kappa},r_0,z_\infty) =f_{\boldsymbol{R}_{\kappa}}(\boldsymbol{r}_{\kappa})f_{Z_\infty}(z_\infty)$. Given (\ref{ho-prob-comp}) and (\ref{mob-cov1}), an \ac{UB} on the coverage probability of a mobile UAV-UE under \ac{CoMP} transmissions is obtained in the next theorem.
\begin{theorem}
\label{ub-cov-comp-aeue}
An \ac{UB} on the coverage probability of a \ac{3D} mobile UAV-UE cooperatively served via CoMP transmission from \acp{BS} within a collaboration distance $R_c$ is given by:
\begin{align}
\label{cov-prob-theory-mob}		
\Pb_{{\rm c}} = \big(1-\beta + \beta\times\Pb(\bar{H}) \big)\sum_{\kappa=1}^{\infty} \Pb(n=\kappa) \int_{h_1}^{h_2}  \int_{\boldsymbol{r_\kappa}=\boldsymbol{R_c}}^{\boldsymbol{\infty}} \Pb_{{\rm c}|\boldsymbol{r}_{\kappa},z_{\infty}}^l  f_{Z_\infty}(z_\infty)
\prod_{i=0}^{\kappa} \frac{2r_i}{R_c^2} \dd{ \boldsymbol{r}_{\kappa}} \dd{z_\infty},
\end{align}
where $\Pb(\bar{H}) = 1-\Pb(H)$ from (\ref{ho-prob-comp}), $\Pb_{{\rm c}|\boldsymbol{r}_{\kappa},z_{\infty}}^l = \lVert e^{\boldsymbol{T}_{K}}\rVert_1$, and $\boldsymbol{T}_K$ is as defined in (\ref{cov-prob-theory}), with \black{$ \Omega(\varpi)_{|\boldsymbol{r}_{\kappa},z_{\infty}} =-2\pi \lambda_b\int_{\nu=R_c}^{\infty}\Big(1 -\delta_l\Pb_{l}(\nu) - \delta_n\Pb_{n}(\nu) \Big)\nu\dd{\nu}$}, $\delta_l= \Big(1 + \frac{\varpi P_t A_l G_s (\nu^2 + z_\infty^2)^{-\alpha_l/2} }{m_l} \Big)^{-m_l}$, $\delta_n= \Big(1 + \frac{\varpi P_t A_n G_s (\nu^2 + z_\infty^2)^{-\alpha_n/2} }{m_n} \Big)^{-m_n}$, $\varpi = \frac{\vartheta}{\kappa P_t\theta}$,  $\theta =\frac{\sum_{i}^{\kappa}{\zeta_{l}(r_i)^2}}{m_l \sum_{i}^{\kappa}  \zeta_{l}(r_i)}$, and $\zeta_{l}(r_i) = A_l G_s \big(r_i^2 + z_{\infty}^2\big)^{-\alpha_l/2}$.
\end{theorem}
\begin{proof}
The proof follows from (\ref{mob-cov1}) and Theorem \ref{ch5:cov-prob}, and is analogous to Theorem \ref{the-2-mob-near}. 
\end{proof}

The effect of $\beta$ on the coverage probability in (\ref{cov-prob-theory-mob}) can be interpreted in a similar way to the nearest association scheme in (\ref{cov-prob-theory01}). Moreover, conditioning on $Z_{\infty}=z_{\infty}$, and for a given $\beta$ in (\ref{cov-prob-theory-mob}), the yielded expression holds the same insights as for static UAV-UEs in Section \ref{static-uav-comp}. In particular, what the Nakagami fading parameter $m_l$, antenna down-tilting angle,  and the collaboration distance $R_c$ entail for the performance of mobile UAV-UEs is similar to the that of static UAV-UEs. Finally, a simple lower bound on the mobile UAV-UE coverage probability can be obtained similar to Corollary \ref{ch5:cov-prob-lb}, with the detailed omitted due to space limitation.

%\blue{It is clear from (\ref{cov-prob-theory01}) that if $\beta=1$, the first term vanishes and the UAV-UE will be in coverage only if there is no handover associated with its mobility. This is because the handover will always cause connection failure. Moreover, since it is hard to directly obtain insights from (\ref{cov-prob-theory01}) on the effect of the altitude $z_{\infty}$ and the altitude difference $\hbar$, several numerical results based on (\ref{cov-prob-theory01}) will be shown in Section \ref{num-result} to provide key practical insights.}
%
%This can be explained by the earlier insight of
%It is remarkable that ?
% In light of this, two major classes

\vspace{-0.4 cm}
\section{Simulation Results and Analysis}
\label{num-result}
\begin{table}[!t]		% ht
\vspace{-0.3cm}
\caption{Simulation Parameters} % title of Table
\centering % used for centering table
\scalebox{0.7}{
\begin{tabular}{|c| c| c| c| c | c|} % centered columns (3 columns)
\hline\hline %inserts double horizontal lines
Description & Parameter & Value&Description & Parameter & Value  \\ [0.5ex] % inserts table
%heading
\hline % inserts single horizontal line
LoS path-loss exponent& $\alpha_{l}$& 2.09 &$\sir$ threshold&$\vartheta$&\SI{0}{\deci\bel}\\ 
NLoS path-loss exponent& $\alpha_{n}$& 3.75 &\acp{BS}' intensity &$\lambda_{b}$&20 \acp{BS}/\SI{}{km}$^2$\\ 
LoS path-loss constant& $A_{l}$& \SI{-41.1}{dB} &Inter-cluster center distance&$2R_h$&\SI{380}{m}\\
NLoS path-loss constant& $A_{n}$ & \SI{-32.9}{dB} &Antenna main-lobe gain&$G_m$&\SI{10}{dB}\\	
Nakagami fading parameter (LoS)&$m_l$&3&Antenna side-lobe gain&$G_s$&\SI{-3.01}{dB}\\
Nakagami fading factor (NLoS)&$m_n$&1&\ac{BS} antenna height&$h_{\textrm{BS}}$ &\SI{30}{m}\\
Area fraction occupied by buildings &$a$& 0.3&UAV-UE altitude&$h_d$&\SI{120}{m}\\
Density of buildings&$\eta$ &$\SI{300}{km^{-2}}$&Simulation area&$R_{\rm sim}$&$\SI{20}{km^2}$\\
Buildings height Rayleigh parameter&$c$ &$\SI{20}{m}$&Mean altitude of mobile UAV-UEs&$L_{z_{\infty}}$&$\SI{150}{m}$\\
%Mean number of buildings &$\eta$&200 per \SI{}{km}$^2$\\
%Buildings height Rayleigh parameter &$c$&15\\
%Down-tilt angle &$\varphi_t$& $8^\circ$\\
%Vertical beamwidth &$\varphi_B$& $30^\circ$\\
%Content caching probability &$c_f$& 1\\
\hline %inserts single line
\end{tabular}}
\label{ch4:table:sim-parameter} % is used to refer this table in the text
\vspace{-0.5cm}
\end{table}

\begin{figure}[!t]
\vspace{-0.0 cm}	
    \centering
    \subfigure[UAV-UE coverage probability versus its altitude $h_d$]
    {
        \includegraphics[width=3.1in]{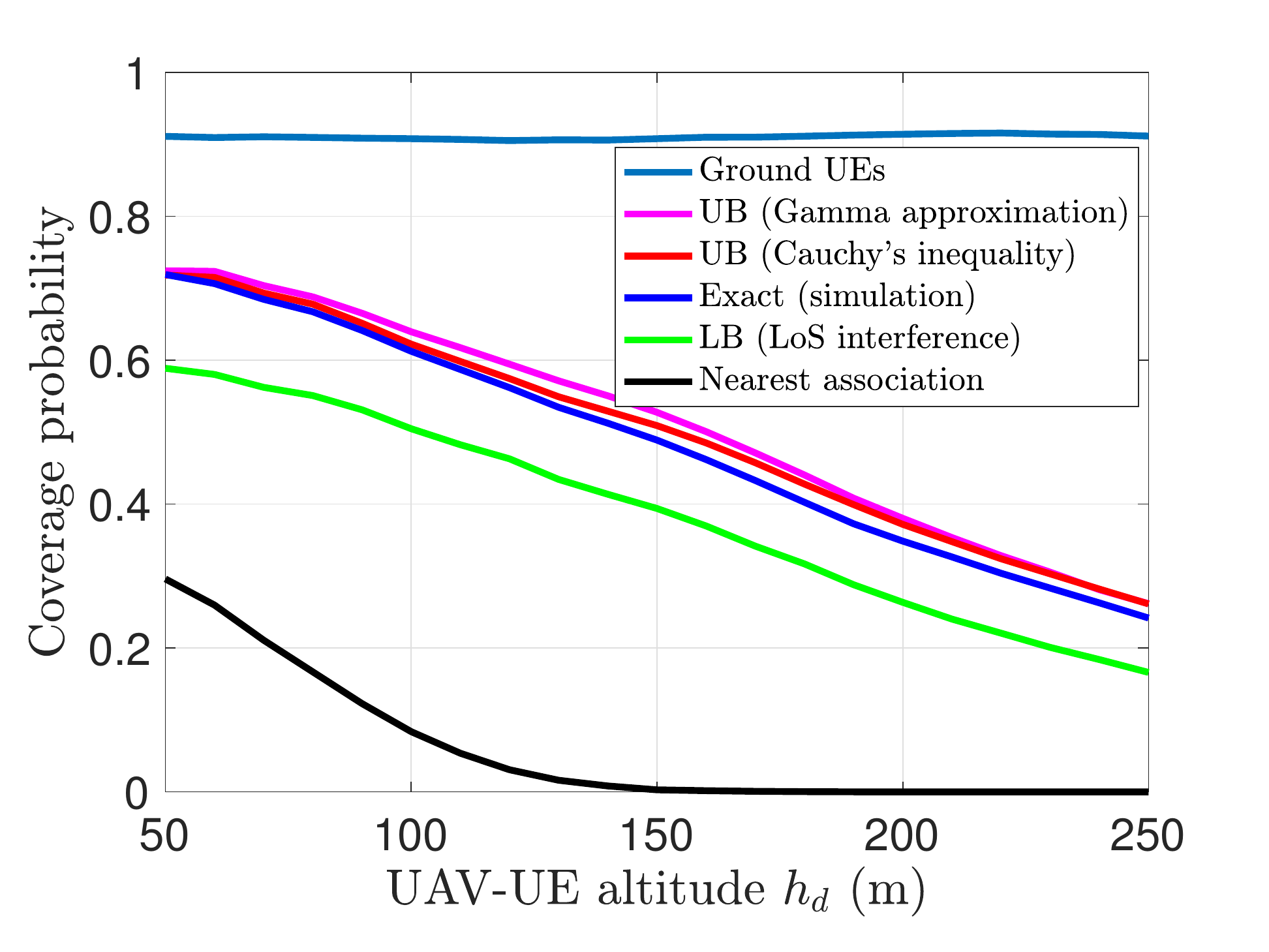}		%lower-bound.eps
        \label{cov_prob_vs_alt}
    }
    \subfigure[UAV-UE coverage probability versus \ac{BS}'s intensity  $\lambda_b$]
    {
        \includegraphics[width=3.1in]{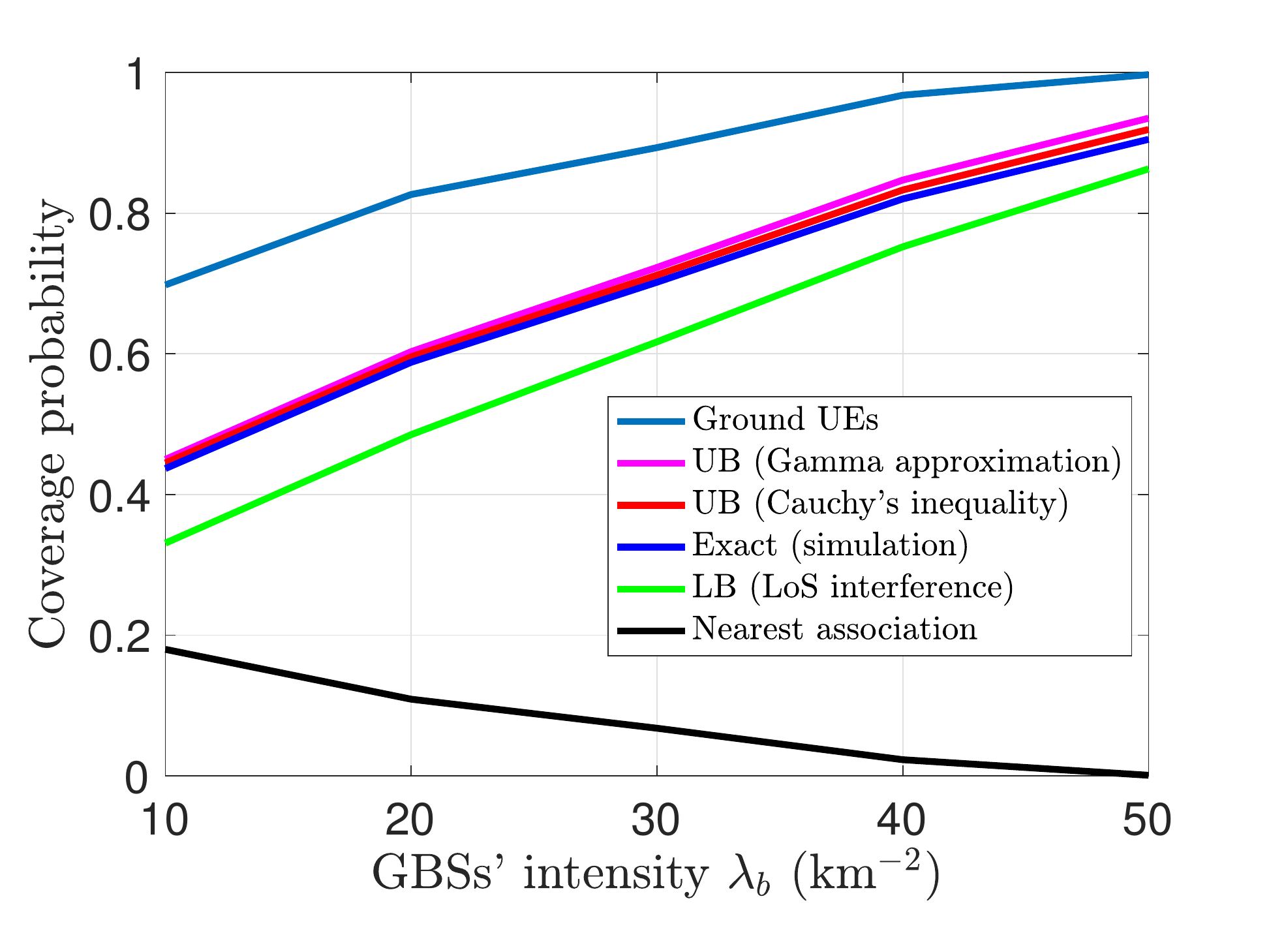}		%31
        \label{cov-prob-vs_lamda}
    }
    \caption{The derived upper and lower bounds on the static UAV-UE coverage probability are plotted versus the  UAV-UE altitude $h_d$ and \acp{BS}' intensity $\lambda_b$.} %\red{(improve the LB for very low $\vartheta$))} 
    \label{cov_prob_vs_alt_lamda}
    \vspace{-0.7 cm}
\end{figure}		% Monte Carlo simulations are used to corroborate the developed mathematical model.
For our simulations, we consider a network having the  parameter values indicated in Table \ref{ch4:table:sim-parameter}. In Fig.~\ref{cov_prob_vs_alt_lamda}, we show the effect of the UAV-UE altitude and \acp{BS}' intensity on the coverage probability of static UAV-UEs, with that of \acp{GUE} plotted for comparison. Fig.~\ref{cov_prob_vs_alt} shows that the coverage probability of UAV-UEs monotonically decreases as $h_d$ increases. This is because, as the UAV-UEs altitude increases, the signal power decreases while the \ac{LoS} interference becomes dominant. Fig.~\ref{cov_prob_vs_alt} also shows that the derived \ac{UB} on the coverage probability in (\ref{cov-prob-theory}) is considerably tight. 
% and  due to large path-loss attenuations 
% %Finally, Fig.~\ref{cov-prob-vs_lamda} shows the effect of UAV-UE altitude on its coverage probability, with that of \acp{GUE} plotted for comparison. 
Meanwhile, Fig.~\ref{cov-prob-vs_lamda} illustrates the effect of \acp{BS}' intensity $\lambda_b$ on the performance of UAV-UEs. Except for the nearest association scheme, the coverage probability improves with $\lambda_b$ since more \acp{BS} cooperate to serve the aerial (and ground) \acp{UE}. However, when the UAV-UE associates to its nearest \ac{BS}, the effect of interference increases as the network becomes denser. %(larger $\lambda_b$). 

%%%%%%%%%%%%%%%%%%%%%%%%%%%%%%%%%%%%%%%%%%%%%%%%%%%%%%%%%%%%%%%%%%%%%%%%%%%%%%%%%%%%%
%Handover rate versus in Fig.~{HO-rate-1} when the UAV-UE is associated with the nearest \ac{BS}. Monte Carlo simulations are used to corroborate the developed mathematical model. We first plot the coverage probability versus the \acp{BS}' intensity  at different collaboration distances (VIP:$\vartheta=0$). Handover rate versus in Fig.~{HO-rate-1} when the UAV-UE is cooperatively served from collaborative \acp{BS}. Monte Carlo simulations are used to corroborate the developed mathematical model. We first plot the coverage probability versus the \acp{BS}' intensity  at different collaboration distances (VIP:$\vartheta=0$). We plot the coverage probability versus the velocity of  UAV-UEs when the UAV-UE is associated with the nearest \ac{BS}. We consider a network having the  parameter values indicated in Table \ref{ch4:table:sim-parameter}. Monte Carlo simulations are used to corroborate the developed mathematical model. We first plot the coverage probability versus the \acp{BS}' intensity  at different collaboration distances (VIP:$\vartheta=0$). A\acp{UE}. 
\begin{figure}[!t]	
\vspace{-0.0 cm}
    \centering
    \subfigure[$\bar{\nu}=\SI{50}{kmh}$, $\mu=\SI{300}{km^{-2}}$]
    % Handover rate versus $\lambda_b$ at different altitude differences $\hbar$ ($\bar{\nu}=\SI{50}{kmh}$)
    { \hspace*{-0.5 in}
        \includegraphics[width=2.4in]{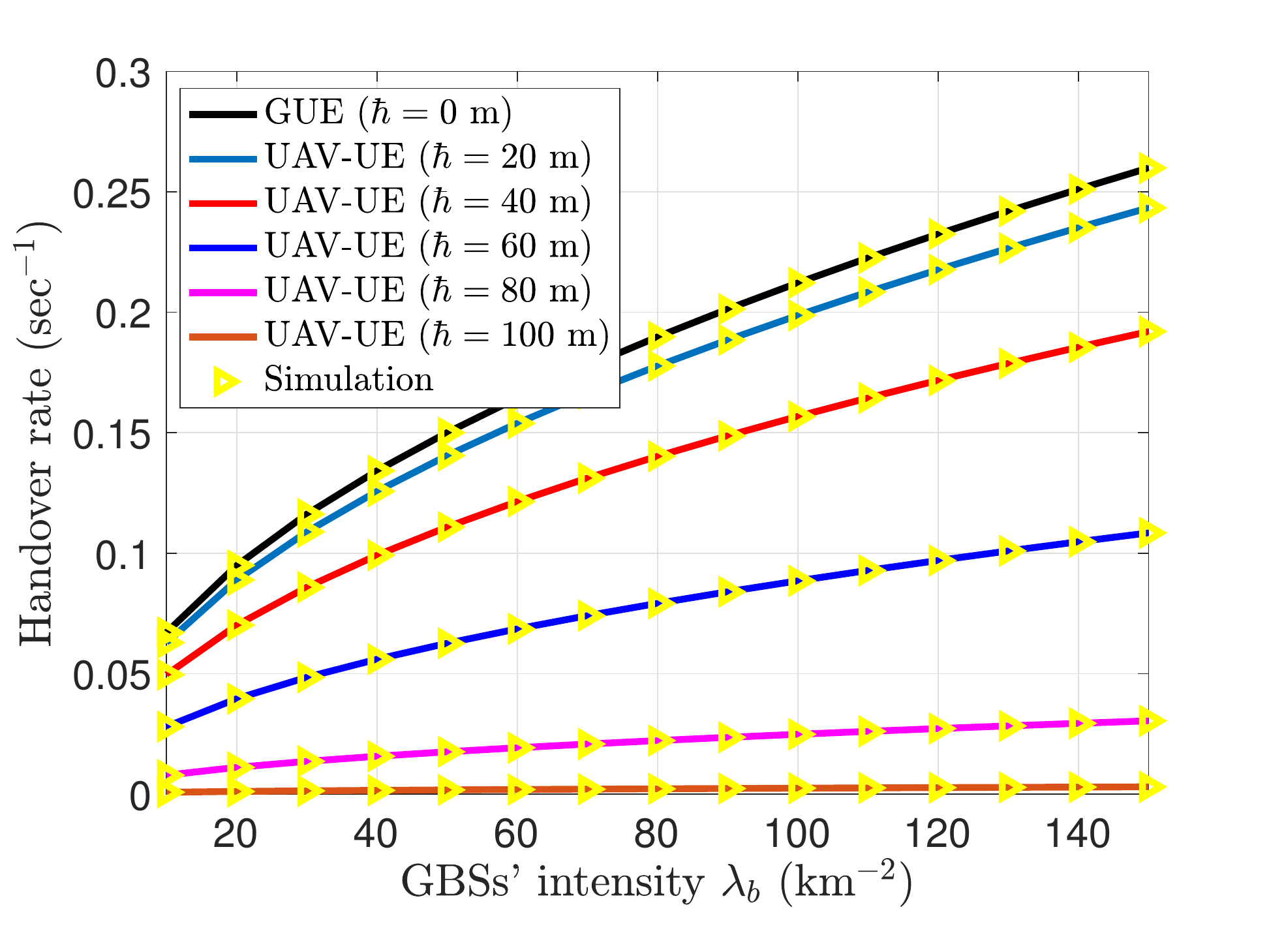}		%lower-bound.eps 
        %HO_rate1.eps
        \label{HO-rate-1}
    }%\hfill
    \subfigure[$\hbar=\SI{30}{m}$, $\mu=\SI{300}{km^{-2}}$]
    % Handover rate versus $\lambda_b$ at different velocities $\bar{\nu}$
    {
        \includegraphics[width=2.4in]{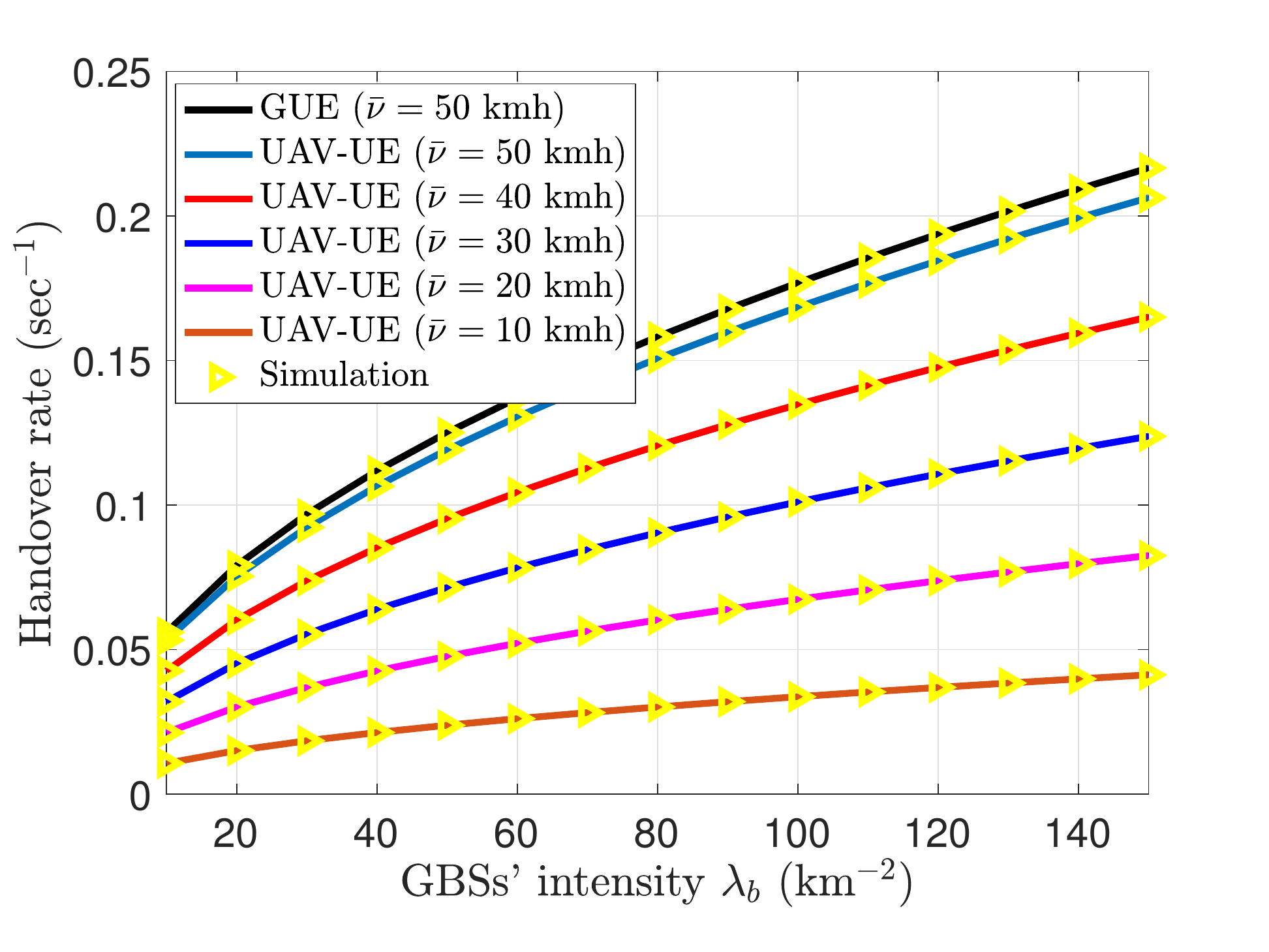}
        \label{HO-rate-2}
    }%\hfill
    \subfigure[$\beta=0.5$, $\mu=\SI{100}{km^{-2}}$, $\vartheta=\SI{-10}{dB}$]
    {
        \includegraphics[width=2.4in]{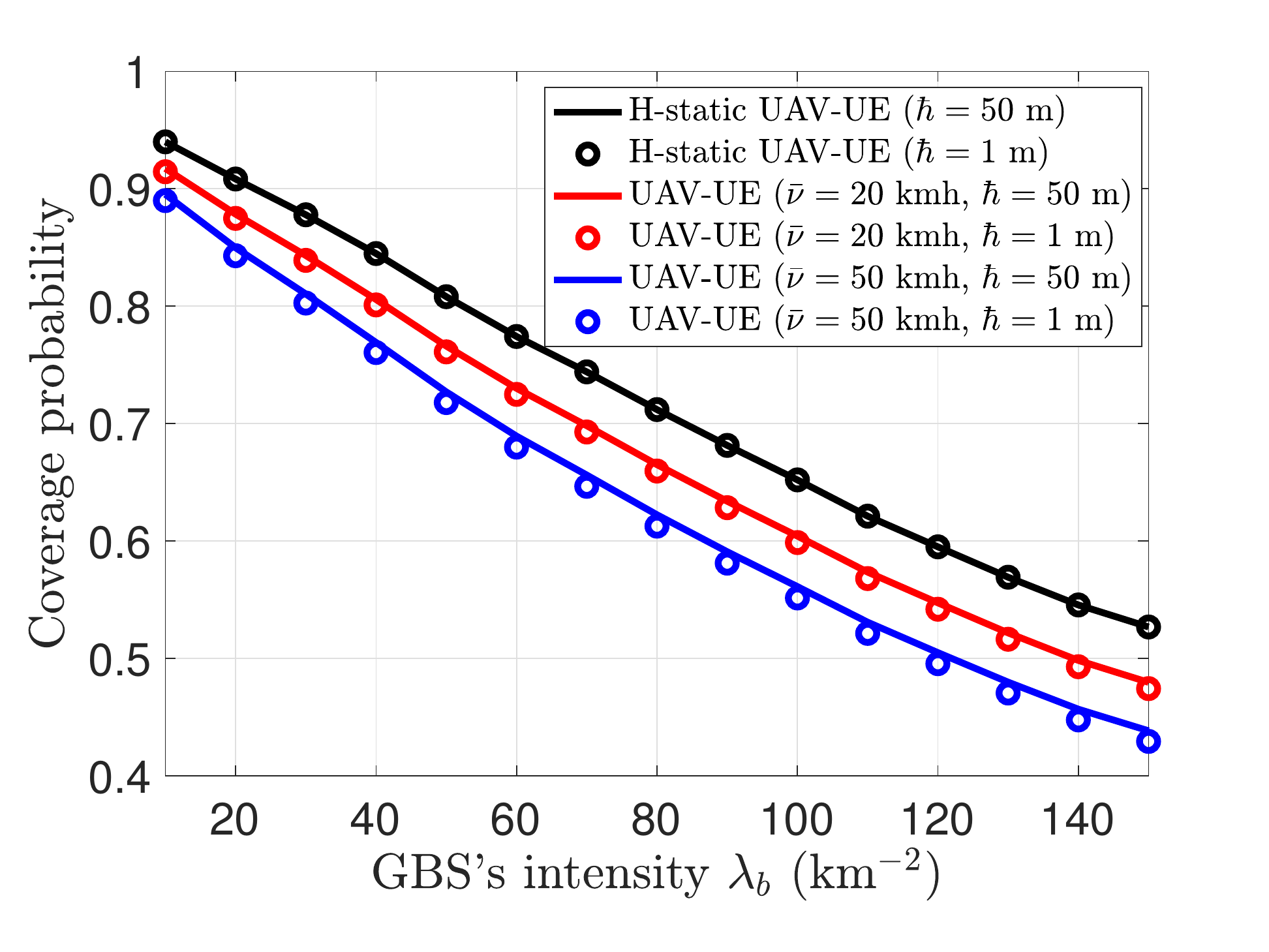}
        \label{HO-rate-3}
    }%\hfill
    \caption{Effect of the \ac{3D} mobility on the performance of aerial and \acp{UE} when they are associated with their nearest \acp{BS}. In (c), H-static refers to a UAV-UE that only moves in the vertical direction within an altitude difference $\hbar$.}		
    \label{HO-rate-123}					% $h_2=\SI{150}{m}$
 \vspace{-0.6 cm}
\end{figure}
% \SI{50}{m}
% UAV-UE with $\hbar=\SI{50}{m}$ represents a scenario when the UAV-UE moves only in the vertical direction
% Handover rate is plotted versus the \acp{BS}' intensity  $\lambda_b$ and UAV-UE velocity \blue{when the UAV-UE is associated with the nearest \ac{BS}}: mobility parameter $\lambda=\SI{100}{km^{-2}}$
Next, we study the impact of \ac{3D} mobility on the performance of UAV-UEs. We further compare the performance of UAV-UEs with their ground counterparts moving horizontally with the same velocity $\bar{\nu}$. In Fig.~\ref{HO-rate-123}, the handover rate and coverage probability of mobile aerial and ground \acp{UE} associated with their nearest \acp{BS}  are investigated. Fig.~\ref{HO-rate-1} plots the handover rate versus $\lambda_b$  at different values of the altitude difference $\hbar$. Fig.~\ref{HO-rate-1} shows that the analytical result in (\ref{ho-rate}) matches the simulation result quite well. As is the case for typical Poisson-Voronoi models, the handover rate grows linearly with the square root of the \ac{BS}'s intensity $\sqrt{\lambda_b}$. Moreover, the handover rate decreases as $\hbar$ increases, which implies that a UAV-UE having frequent up and down motions along its trajectory is susceptible to lower rates of handover. We also note that the handover rate of UAV-UEs is upper bounded by that of \acp{GUE} at $\hbar=0$. 
Fig.~\ref{HO-rate-2} shows the effect of the UAV-UE speed $\bar{\nu}$ on its handover rate.\footnote{Low values of the velocity $\bar{\nu}$ suits the motion of UAV-UEs such as surveillance cameras while higher velocities would be suitable for UAV-UEs such as flying taxis.} Intuitively, the handover rate increases as $\bar{\nu}$ increases since the UAV-UE stays shorter time in the area covered by each \ac{BS}, i.e., shorter sojourn time. 
Finally, Fig.~\ref{HO-rate-3} investigates the effect of mobility on the UAV-UE coverage probability given an arbitrary  handover penalty  $\beta$. Notice that the coverage probability decreases as $\bar{\nu}$ increases since this leads to higher handover probability (penalized by $\beta$). Moreover, the altitude difference $\hbar$ has a marginal  effect on the coverage probability of UAV-UEs. This is attributed to the fact that the increase of the altitude difference $\hbar$ for mobile UAV-UEs while keeping the same average flying altitude $L_{z_{\infty}}$ relatively yields the same average coverage probability. 

In Fig.~\ref{HO-rate-456}, we evaluate the effect of the \ac{3D} mobility on the UAV-UE performance under CoMP transmissions. Fig.~\ref{HO-rate-4} shows that the inter-CoMP handover rate monotonically decreases as $R_c$ increases since the UAV-UE would have a longer sojourn time in each cluster. Moreover, the handover rate is shown to decrease as $\hbar$ increases, i.e., when the UAV-UE has frequent up and down motions along its trajectory. We also note that the handover rate of UAV-UEs is upper bounded by that of \acp{GUE}, with $\hbar=0$. 
Fig.~\ref{HO-rate-5} shows the effect of the UAV-UE velocity $\bar{\nu}$ on the inter-CoMP handover. We note that this handover rate also increases as $\bar{\nu}$ increases since the UAV-UE will have a shorter sojourn time in each cluster. 
% given the handover penalty $\beta=0.5$,
Finally, Fig.~\ref{HO-rate-6} shows the effects of the UAV-UE velocity and altitude difference on the UAV-UE coverage probability.  Fig.~\ref{HO-rate-6} shows that the \ac{UB} on the coverage probability, characterized in Theorem \ref{ub-cov-comp-aeue}, slightly decreases as $\bar{\nu}$ increases, which corresponds to a higher handover rate (penalized by $\beta$). This slight decrease is essentially because as the inter-cluster distance becomes larger, the probability of handover decreases and  its effect gradually vanishes. Similar to the nearest association scheme, the altitude difference $\hbar$ has a minor effect on the coverage probability of UAV-UEs. In addition to its impact on the coverage probability, the mobility of UAV-UEs can decrease their throughput, particularly, when accounting for the handover execution time \cite{7827020}. 
% \red{interpretation of the effect of the altitude difference.}}  %  However, the user throughput is beyond the scope of this paper.

%%%%%%%%%%%%%%%%%%%%%%%%%%%%%%%%%%%%%%%%%%%%%%%%%%%%%%%%%%%%%%%%%%%%%%%%%%%%%%%%%%%%%
\begin{figure}[!t]	
\vspace{-0.8 cm}
    \centering
    \subfigure[$\bar{\nu}=\SI{30}{kmh}$, $\mu=\SI{300}{km^{-2}}$]			% 100
    {\hspace*{-0.5 in}
        \includegraphics[width=2.4in]{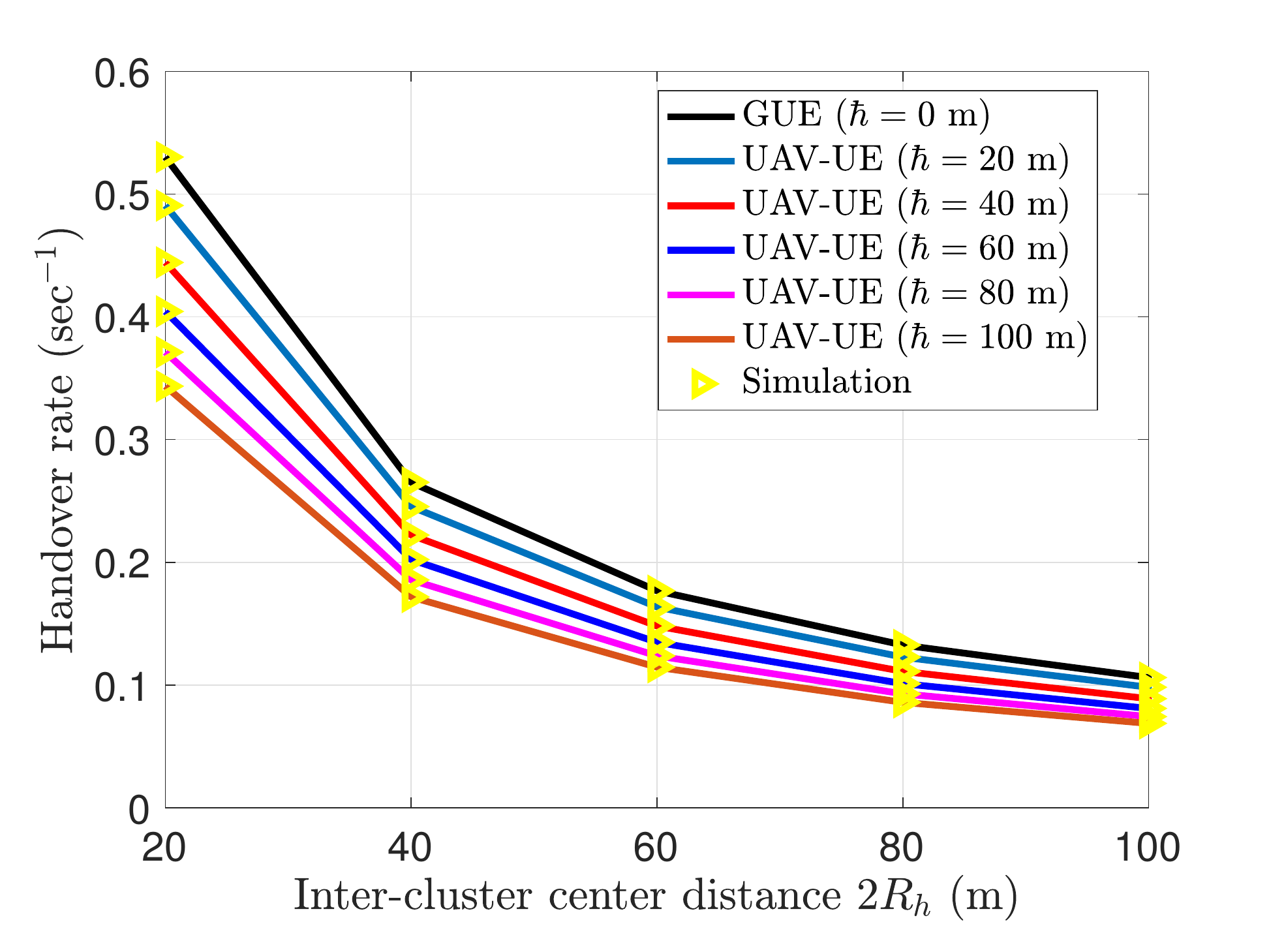}		%lower-bound.eps
        \label{HO-rate-4}
    }%\hfill
    \subfigure[$\mu=\SI{100}{km^{-2}}$, $\hbar=\SI{50}{m}$]
    {
        \includegraphics[width=2.4in]{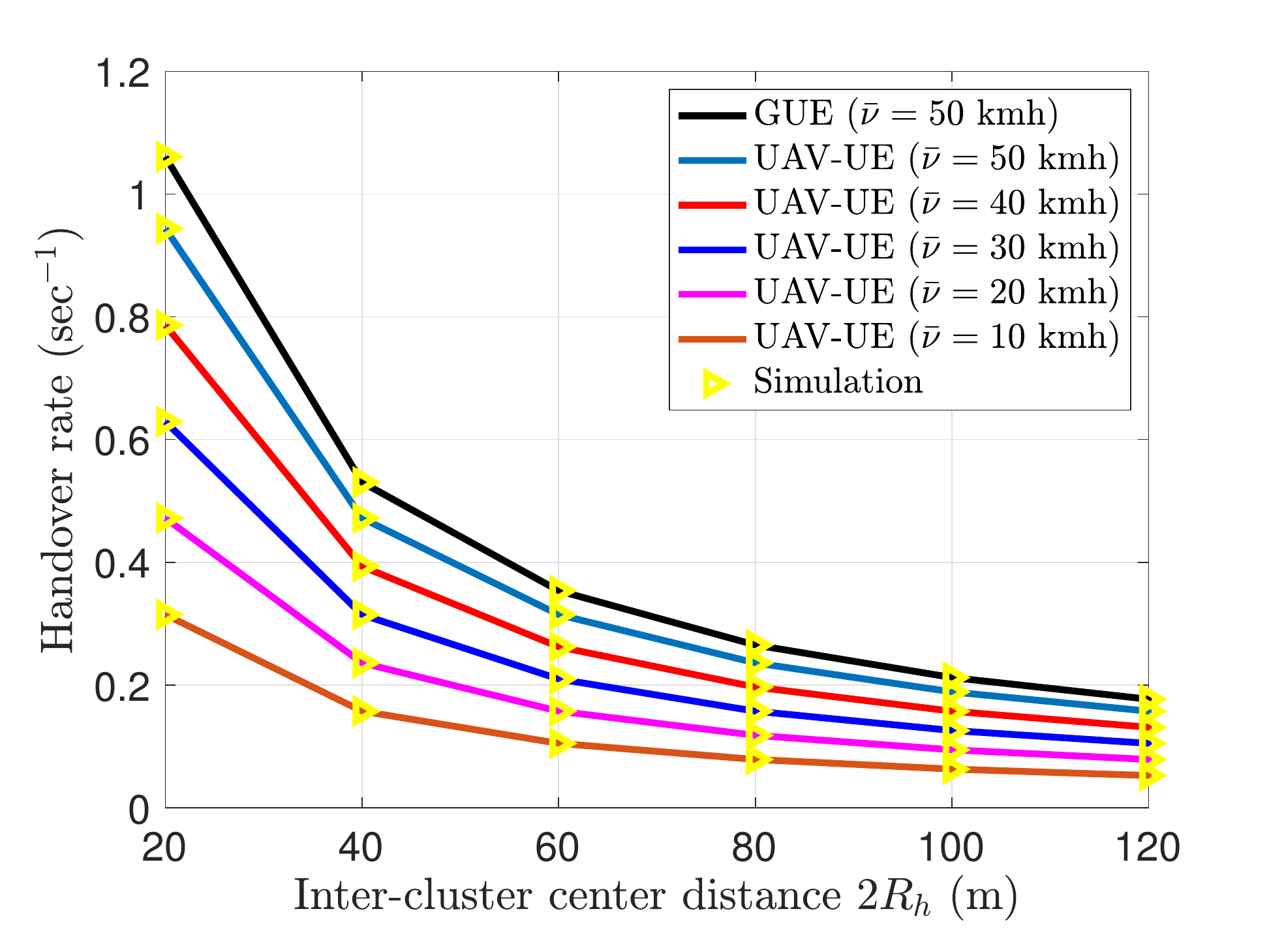}
        \label{HO-rate-5}
    }%\hfill
    \subfigure[$\mu=\SI{100}{km^{-2}}$, $\beta=1$]
    {
        \includegraphics[width=2.4in]{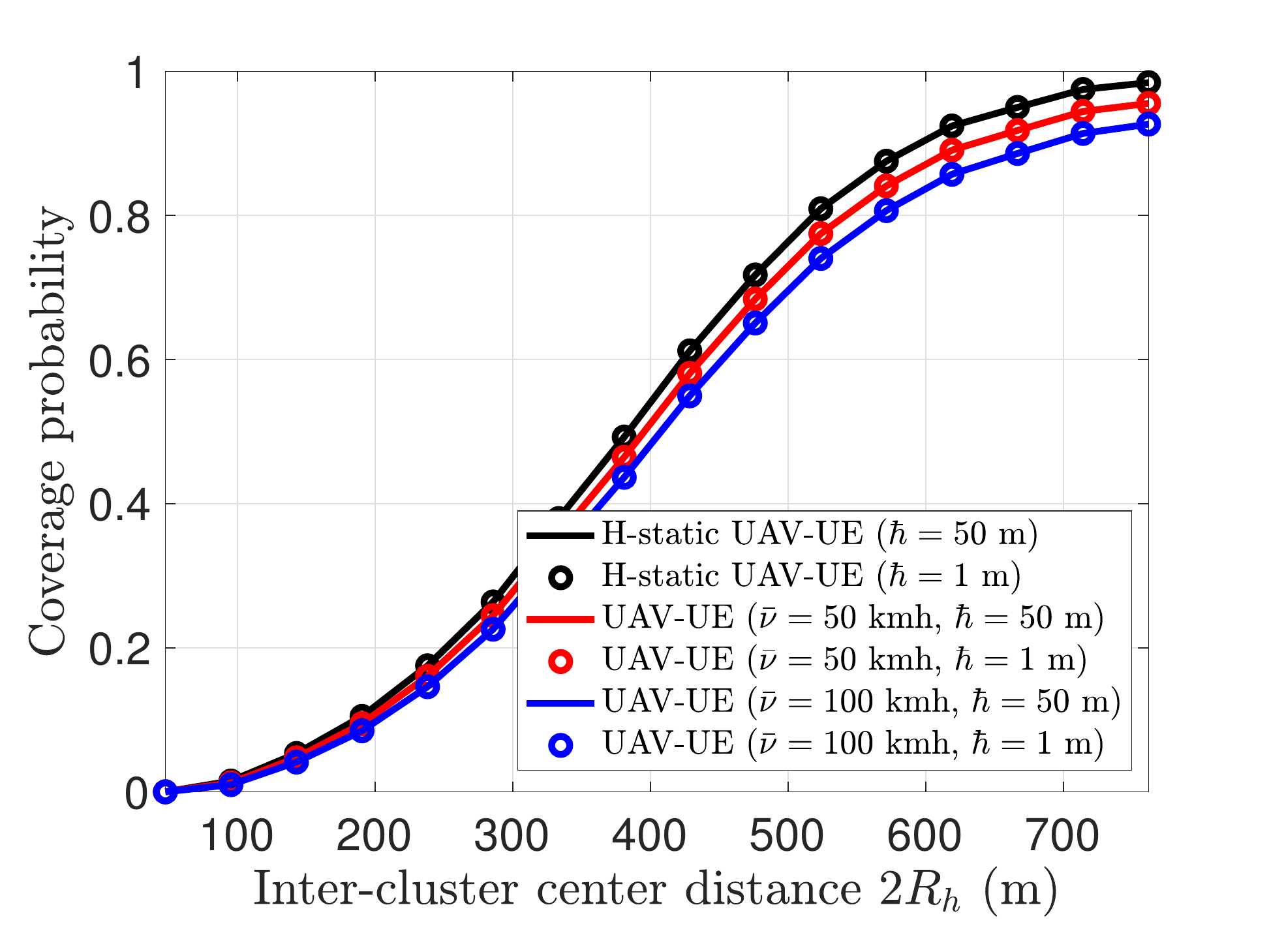}
        \label{HO-rate-6}
    }%\hfill
    \caption{Effect of the \ac{3D} mobility on the performance of aerial and ground \acp{UE} when they are served via \ac{CoMP} transmission with the inter-cluster center distance set equal to $2R_h$. In (c), H-static refers to an UAV-UE that only moves in the vertical direction within an altitude difference $\hbar$.}		
    \label{HO-rate-456}
     \vspace{-0.6 cm}
\end{figure} 		%however, this decrease is less tangible than for nearest association in Fig.~\ref{HO-rate-456}
\vspace{-0.4 cm}
\section{Conclusion} 	
\vspace{-0.2 cm}
In this paper, we have proposed a novel framework for cooperative transmission that can be leveraged to provide reliable connectivity and omnipresent mobility support for UAV-UEs. In order to analytically characterize the performance of UAV-UEs, we have employed Cauchy's inequality and moment approximation of Gamma \acp{RV} to derive upper and lower bounds on the UAV-UE coverage probability. Moreover, we have developed a novel 3D \ac{RWP} model that allowed us to explore the role of UAV-UEs' mobility in cellular networks, particularly, to quantify the handover rate   and the impact of their mobility on the achievable performance. For both static and mobile UAV-UEs, we have shown allowing \ac{CoMP} transmission significantly improves the achievable coverage probability, e.g., from $28\%$ for the baseline scenario with nearest serving \acp{BS}, to $60\%$ for static UAV-UEs. Furthermore, comparing the performance of UAV-UEs to \acp{GUE}, it is shown that the coverage probability of a UAV-UE is always upper bounded by that of a \ac{GUE} owing to the down-tilted antenna pattern and LoS-dominated interference for UAV-UEs. Our results for the case of mobile UAV-UEs have also revealed that their handover rate and handover probability decrease as the altitude difference increases, i.e., in the case of frequent up and down motions of the UAV-UEs along their trajectory. Moreover, while the altitude difference has a minor effect on the coverage probability of mobile UAV-UEs, their velocity noticeably degrades their coverage probability.
%%%%%%%%%%%%%%%%%%%%%%%%%%%%%%%%%%%%%%%%%%%%%%%%%%%%%%%%%%%%%%%%
%and we quantify what the presence of UAVs entails for the performance of conventional ground UEs
%Firstly, performance analysis of network with regular structure may lead to analytically tractable formulation, which can provides insights on the performance of networks with random structure 
%The non-caching scheme is evaluated as the performance benchmark.
%Our analysis concretely demonstrates significant improvement in the network performance when 
%We conclude that our transmission scheme can own a significant performance gain over the traditional C-RAN.

%In this paper, we have proposed a novel CB framework for spatially
%multiplexing AUs and GUs. In order to analytically express the
%SCDP, we have derived the gain of intended and interfering channels.
%We have shown that exploiting CB from massive MIMO-enabled BSs
%to spatially multiplex an AU and GUs substantially improves the
%performance of the AU, in terms of SCDP. We have then shown that
%the down-tilt of the BS antennas can lead to a tradeoff between the
%performance of AUs and GUs only if the AU?s altitude is below the
%BS height. Simulation results have shown the various properties of
%cellular communications when AUs and GUs co-exist.
%%%%%%%%%%%%%%%%%%%%%%%%%%%%%%%%%%%%%%%%%%%%%%%%%%%%%%%%%%%%%%%
\begin{appendices}
\vspace{-0.6 cm}
%%%%%%%%%%%%%%%%%%%%%%%%%%%%%%%%%%%%%%%%%%%%%%%%%%%%%%%%%%%%%%%%%%%%%%%%%%
\section{Proof of Theorem \ref{ch5:cov-prob}}
\label{ch5:theorem1}
We proceed to obtain an \ac{UB} on the coverage probability as follows:
\begin{align}
&\Pb\Big(\frac{\kappa P_tJ}{I_{\rm out}}>\vartheta\Big) =  \Pb\Big(\kappa P_tJ>\vartheta 
I_{\rm out}\Big) 
= \Eb_{I_{\rm out}} \Bigg[  \Pb\Big(\kappa P_tJ>\vartheta I_{\rm out}\Big) \Bigg]
 \nonumber \\
%  &= \Eb_{I_{\rm out}} \Bigg[  \Pb\Big(J>\frac{\vartheta}{\kappa P_t}I_{\rm out}\Big) \Bigg]
% \nonumber \\
  \label{prob-y2} %
 &\overset{(a)}{\approx} \Eb_{I_{\rm out}} \Big[ \sum_{i=0}^{K-1}  \frac{(\vartheta/\kappa P_t\theta)^i}{i!}I_{\rm out}^i
{\rm exp}\Big(-\frac{\vartheta}{\kappa P_t\theta}I_{\rm out}\Big) \Big]
  % \nonumber \\
\overset{(b)}{=} \Eb_{I_{\rm out}} \Big[ \sum_{i=0}^{K-1}  \frac{(-\varpi)^i}{i!}
\frac{d^i}{d\varpi^i} \Lc_{I_{\rm out}|\boldsymbol{r}_{\kappa}}(\varpi)  \Big],
\end{align}
where (a) follows from the \ac{PDF} of Gamma \ac{RV} $J$ with parameters $\theta$ given in (\ref{equiv-gamma}), and  $K=m_l\kappa$; (b) follows from $\varpi = \frac{\vartheta}{\kappa P_t\theta}$, along with the Laplace transform of interference, i.e., the \ac{RV} $I_{\rm out}$. Next, we derive the Laplace transform of interference:
\begin{align}
\Lc_{I_{\rm out}|\boldsymbol{r}_{\kappa}}(\varpi)  &= \Eb_{I_{\rm out}} 
\Big[ e^{-\varpi I_{\rm out}}\Big]
= \mathbb{E} \Bigg[
 e^{-\varpi\sum_{j \in \Phi_{b} \setminus \mathcal{B}(0, R_c) } \varpi \chi_j  P(u_j)^2}  
  \Bigg] 
= \mathbb{E}_{\Phi_b,\chi_j} 			% \mathbb{E}_{}
  \Bigg[  
  \prod_{j \in \Phi_{b} \setminus \mathcal{B}(0, R_c)}  e^{-\varpi \chi_j  P(u_j)^2}
  \Bigg] 
  \nonumber \\ 
  \label{before-chi}
        &\overset{(a)}{=} 
 {\rm exp}\Bigg(-2\pi \lambda_b\int_{\nu=R_c}^{\infty}\Big(1 - \mathbb{E}_{\chi} 
       e^{-\varpi \chi  P(\nu)^2}
        \Big)\nu\dd{\nu}\Bigg) 
\\ 
\label{LT_c1}      
  &\overset{(b)}{=} {\rm exp}\Bigg(-2\pi \lambda_b\int_{\nu=R_c}^{\infty}\Big(1 -
        \delta_l\Pb_{l}(\nu) - \delta_n\Pb_{n}(\nu) 
        \Big)\nu\dd{\nu}\Bigg)
% \\ 
% \label{LT_c0} 
% &
 \overset{(c)}{=}e^{\Omega(\varpi)_{|\boldsymbol{r}_{\kappa}}} ,
\end{align}
where $\delta_l=\Big(1 + \frac{\varpi P_l(\nu)^2 }{m_l} \Big)^{-m_l}$, and $\delta_n=\Big(1 + \frac{\varpi P_n(\nu)^2 }{m_n} \Big)^{-m_n}$; (a) follows from the probability generating functional (PGFL) of \ac{PPP} along with Cartesian to polar coordinates conversion \cite{haenggi2012stochastic}, (b) follows from the moments of the Gamma \ac{RV} $\chi\sim\Gamma(m_v,1/m_v)$ modeling the interfering channel gain, and (c) follows from 
$\Omega(\varpi)_{|\boldsymbol{r}_{\kappa}} = -2\pi \lambda_b\int_{\nu=R_c}^{\infty}\Big(1 -\delta_l\Pb_{l}(\nu) - \delta_n\Pb_{n}(\nu) \Big)\nu\dd{\nu}$. 
In \cite{8490204}, it is proved that  $  \sum_{i=0}^{K-1}  \frac{(-\varpi)^i}{i!} \Lc_{I|\boldsymbol{r}_{\kappa}}^{(i)}(\varpi)  = \sum_{i=0}^{K-1}  p_i$, where $p_i =\frac{(-\varpi)^i}{i!} \Lc_{I|\boldsymbol{r}_{\kappa}}^{(i)}(\varpi)$ can be computed from the recursive relation:
$p_i = \sum_{l=0}^{i-1}\frac{i-l}{i} p_l t_{i-l}$, with $t_{i-1} =\frac{(-\varpi)^{i-1}}{(i-1)!} \Omega^{(i-1)}(\varpi)$, and $\Omega^{(i-1)}(\varpi)= \frac{d^{i-1}}{d \varpi^{i-1}} \Omega(\varpi)_{|\boldsymbol{r}_{\kappa}}$. After some algebraic manipulation as in \cite{8490204}, $\Pb_{{\rm c}|\boldsymbol{r}}^l$ can be expressed in a compact form $\Pb_{{\rm c}|\boldsymbol{r}}^l = \lVert e^{\boldsymbol{T}_{k}}\rVert_1$, where  $\lVert.\rVert_1$ represents the induced $\ell_1$ norm, and $\boldsymbol{T}_{K}$ is the lower triangular Toeplitz matrix whose entries are $t_i$, $i=|\{1,\dots,K\}$. This completes the proof.

\vspace{-0.3 cm}
\section{Proof of Corollary \ref{ch5:cov-prob-lb}}
\vspace{-0.3 cm}
\label{ch5:theorem2}
We first write the exponent power of (\ref{before-chi}) as
\begin{align}
%\label{log-laplac}
\Omega(\varpi)_{|\boldsymbol{r}_{\kappa}} &= 
        -2\pi  \lambda_b \mathbb{E}_{\chi}  \int_{\nu=R_c}^{\infty}
    \big(1 -       e^{-\varpi \chi P(\nu)^2} \big)\nu\dd{\nu}
 \overset{(a)}{=}  -2\pi  \lambda_b \mathbb{E}_{\chi}  \int_{\nu=R_c}^{\infty}
    \big(1 -       e^{-\varpi L \chi (\nu^2 + h^2)^{-\alpha_l/2}} \big)\nu\dd{\nu},
    \nonumber 
\end{align}   
where (a) follows from $P(\nu)^2= P_t A_l G_s \big(\nu^2 + h^2\big)^{-\alpha_l/2}$ and substituting $L=P_tA_l G_s$. Let $z = \nu^2+h^2$, and $\dd{z} = 2\nu\dd{\nu}$, we hence get 
\begin{align}
\label{T_1}
\Omega(\varpi)_{|\boldsymbol{r}_{\kappa}} &= -\pi  \lambda_b \mathbb{E}_{\chi}  \int_{z=R_c^2+h^2}^{\infty}
    \big(1 -       e^{-\varpi L \chi z^{-\alpha_l/2}} \big)\dd{z}.
%    \nonumber \\
%    &=  \pi  \lambda_b R_{ch}^2  +    \pi  \lambda_b R_{ch}^2  \Big(sR_{ch}^2\Big)
%       \mathbb{E}_{\chi}  \int_{z=R_c^2+h^2}^{\infty}
%    \big(1 -       e^{-\varpi L \chi z^{-\alpha_l/2}} \big)z\dd{z},
\end{align}
By changing the variables $y=z^{-\alpha_l/2}$, $z = y^{-2/\alpha_l}$, and $\dd{z} = \frac{-2}{\alpha_l}y^{\frac{-2}{\alpha_l} - 1} \dd{y}$, and solving the reproduced integrals as in \cite{6042301}, we get
 \begin{align}
\Omega(\varpi)_{|\boldsymbol{r}_{\kappa}} & \overset{}{=} 
\pi  \lambda_b R_{ch}^2 - 
\delta_l \pi  \lambda_b  (\varpi L)^{\delta_l} \mathbb{E}_{\chi} \Big[ \chi^{\delta_l} \gamma(-\delta_l,\varpi L \chi R_{ch}^{-\alpha_l/2}) \Big]
\nonumber \\
\label{gamma-incomp}
& \overset{(a)}{=}
\pi  \lambda_b R_{ch}^2 - 
\delta_l \pi  \lambda_b  (\varpi L)^{\delta_l} \mathbb{E}_{\chi} \Big[ \chi^{\delta_l} 
\epsilon  {}_1 F_1(-\delta_l;1-\delta_l;-\varpi L \chi R_{ch}^{-\alpha_l/2}) 
\Big],
\end{align}
 where $R_{ch}^2 = R_c^2+h^2$, $\delta_l = \frac{2}{\alpha_l}$, $\epsilon= \frac{(\varpi L \chi)^{-\delta_l} R_{ch}^2}{\delta_l}$, and $\gamma(s,x) = \int_{0}^{x} t^{s-1} e^{-t}$ is the lower incomplete Gamma function; (a) follows from ${}_1 F_1(s;s+1;-x) = \frac{s}{x^s} \gamma(s,x)$, where ${}_1 F_1(\cdot;\cdot;\cdot)$ is  the confluent hypergeometric function of the first kind. By rearranging (\ref{gamma-incomp}), we can obtain     
 \begin{align}
\Omega(\varpi)_{|\boldsymbol{r}_{\kappa}} & = \pi  \lambda_b R_{ch}^2 \Bigg(1 - \mathbb{E}_{\chi} \Big[{}_1 F_1(-\delta_l;1-\delta_l;-\varpi L \chi R_{ch}^{-\alpha_l/2}) \Big] \Bigg). 
\end{align} 				% \varpi^{(k)}(\varpi) =
 
 The non-zero terms in $\boldsymbol{T}_{k}$ can be then determined from: 
\begin{align}
 t_{k} &=\frac{(-\varpi)^{k}}{k!} \Omega(\varpi)_{|\boldsymbol{r}_{\kappa}}^{(k)}=
\pi  \lambda_b R_{ch}^2
\frac{(-\varpi)^{k}}{k!}  \frac{d^k}{d \varpi^k}
 \Bigg[ 1 -
  \mathbb{E}_{\chi} \Big[ {}_1 F_1(-\delta_l;1-\delta_l;-\varpi L \chi R_{ch}^{-\alpha_l/2}) \Big] 
 \Bigg]
 \nonumber \\
 &= %\frac{(-\varpi)^{k}}{k!} \Omega(\varpi)_{|\boldsymbol{r}_{\kappa}}^{(k)}=
\pi  \lambda_b R_{ch}^2
 \mathbb{E}_{\chi} \Bigg[ \frac{(-\varpi)^{k}}{k!} (-L \chi R_{ch}^{-\alpha_l/2})^{k}  \frac{d^k}{d (-\varpi L \chi R_{ch}^{-\alpha_l/2})^k} 
 \Big[ 1 -  {}_1 F_1(-\delta_l;1-\delta_l;-\varpi L \chi R_{ch}^{-\alpha_l/2}) \Big] 
 \Bigg]
  \nonumber \\
  &= %\frac{(-\varpi)^{k}}{k!} \Omega(\varpi)_{|\boldsymbol{r}_{\kappa}}^{(k)}=
\pi  \lambda_b R_{ch}^2
 \mathbb{E}_{\chi} \Bigg[
\frac{(\varpi L \chi R_{ch}^{-\alpha_l/2} )^{k}}{k!} \frac{d^k}{d (-\varpi L \chi R_{ch}^{-\alpha_l/2})^k} 
 \Big[ 1 -
  {}_1 F_1(-\delta_l;1-\delta_l;-\varpi L \chi R_{ch}^{-\alpha_l/2}) \Big] 
 \Bigg]
  \nonumber \\
  &\overset{(b)}{=} 
\pi  \lambda_b R_{ch}^2
\Bigg(\textbf{1}\{k=0\} -  \frac{(\varpi L R_{ch}^{-\alpha_l/2})^{k}}{\Gamma(k+1)} 
\frac{\delta_l}{(\delta_l -k)} 
\mathbb{E}_{\chi} \Big[ \chi^k
 {}_1 F_1(k-\delta_l;k+1-\delta_l;-\varpi L \chi R_{ch}^{-\alpha_l/2}) \Big] \Bigg),
\nonumber
 \end{align}
where (b) follows from the derivatives for hypergeometric functions: $\frac{d^k}{d z^k}  {}_1 F_1(a;b;z) = \frac{\prod_{p=0}^{k-1}(a+p)}{\prod_{p=0}^{k-1}(b+p)} \times  {}_1 F_1(a+k;b+k;z)$. By letting $a_k = (\varpi L R_{ch}^{-\alpha_l/2})^{k}$, we get 
 \begin{align}
 t_{k} &= \pi  \lambda_b R_{ch}^2
\Bigg(\textbf{1}\{k=0\} -  \frac{\delta_l a_k}{(\delta_l -k)\Gamma(k+1)} 
\mathbb{E}_{\chi} \Big[ \chi^k
 {}_1 F_1(k-\delta_l;k+1-\delta_l;-\varpi L \chi R_{ch}^{-\alpha_l/2}) \Big] \Bigg). 
 \nonumber
\end{align}
Lastly, to get a closed-form expression for $t_{k}$, we average over $\chi\sim \Gamma(m_l,1/m_l)$ as follows: 
\begin{align}
\label{closed-form0}
   t_{k} &= \pi  \lambda_b R_{ch}^2  \Bigg( \textbf{1}\{k=0\} - 
b_k 		% \frac{\delta_l a_k}{(\delta_l -k)\Gamma(k+1)}  \frac{m_l^{m_l}}{\Gamma(m_l)}  
\int_{\chi=0}^{\infty}   \chi^{k+m_l-1} e^{-m_l\chi}  {}_1 F_1(k-\delta_l;k+1-\delta_l;-\varpi L R_{ch}^{-\alpha_l/2} \chi) \dd{\chi} \Bigg)
   \nonumber \\
 &= \pi  \lambda_b R_{ch}^2  \Bigg( \textbf{1}\{k=0\} - 
b_k 	\Gamma(k+m_l) m_l^{-(k+m_l)}  {}_2 F_1(k+m_l,k-\delta_l;k+1-\delta_l;-\varpi L R_{ch}^{-\alpha_l/2} m_l)	 \Bigg)
   \nonumber \\
 & \overset{(c)}{=} \pi  \lambda_b R_{ch}^2 \Big(  \textbf{1}\{k=0\} - c_k 
  {}_2 F_1(k+m_l,k-\delta_l;k+1-\delta_l;-\varpi L R_{ch}^{-\alpha_l/2} m_l)  \Big),
  % \nonumber 
 \end{align}
where $b_k = \frac{\delta_l a_k m_l^{m_l} }{(\delta_l -k)\Gamma(k+1) \Gamma(m_l)}$, $c_k =  \frac{m_l^{m_l}}{\Gamma(m_l)} b_k=\frac{\delta_l a_k \Gamma(k+m_l) m_l^{-k} }{(\delta_l -k)\Gamma(k+1) \Gamma(m_l)}$, and (c) follows from solving the integral in (\ref{closed-form0}) \cite[Eq. 7.525]{gradshteyn2014table} and rearranging the right hand side. This completes the proof.
\vspace{-0.3 cm}
\section{Proof of Lemma \ref{near-HO}}
\label{proof:near-HO}
When $\varphi_n=0$, the conditional probability of handover can be expressed as
\begin{align}
\Pb(H|r_0) &=  1 -  \Eb_{\rho_n,Z_n,Z_{n-1}} 
 \Big[  e^{-\pi\lambda_b\big((\bar{\nu}{\rm cos}(\varphi_n))^2 + 2r_0\bar{\nu}{\rm cos}(\varphi_n)\big)}  
\Big]
  \\
 \label{Jensen-LB}
 &\overset{(a)}{\leq}  1 - 
\underbrace{
 e^{-\pi\lambda_b \Eb_{\rho_n,Z_n,Z_{n-1}} \big(\frac{2r_0\bar{\nu} \varrho_n}{\sqrt{\varrho_n^2 + (z_n-z_{n-1})^2}} 
 + 
 (\frac{\bar{\nu}\varrho_n}{\sqrt{\varrho_n^2 + (z_n-z_{n-1})^2}} )^2\big) }}_{\Pb(\bar{H}|r_0)} ,
 \end{align}
where (a) follows from Jensen's inequality, with $\Pb(\bar{H}|r_0)$ being an \ac{LB} on the probability of no handover conditioned on $r_0$. We obtain $\Pb(\bar{H}|r_0)$ in (\ref{Jensen-LB}) as follows:
\begin{align}
 &\Pb(\bar{H}|r_0)= e^{-\pi\lambda_b 
\Eb_{\rho_n,Z_n,Z_{n-1}} \big(
 \frac{2r_0\bar{\nu} \varrho_n}{\sqrt{\varrho_n^2 + (z_n-z_{n-1})^2}} +
 (\frac{\bar{\nu}\varrho_n}{\sqrt{\varrho_n^2 + (z_n-z_{n-1})^2}} )^2  \big) }     
   \nonumber \\
% \label{ub-prob-no-ho}	
   &\overset{(b)}{=} e^{-\pi\lambda_b \Eb_{Z_n,Z_{n-1}}
 \big[ r_0\sqrt{\pi } \bar{\nu} {}_1 F_1 \left(\frac{1}{2};0;\pi  (z_n-z_{n-1})^2 \mu \right)  			
+ \pi \mu \bar{\nu}^2 \big(   
 \frac{1}{\pi  \mu }  -(z_{n}-z_{n-1})^2 e^{\pi  \mu  (z_{n}-z_{n-1})^2} 
 \Gamma \left(0,\pi  (z_{n}-z_{n-1})^2 \mu \right)  
   \big) \big]} ,
   \nonumber 
\end{align}
where (b) follows from averaging over $\rho_n$ whose \ac{PDF} is $f_{\rho_n}(\varrho_n)$. By changing the variables: $p=z_n-z_{n-1}$, with $f_P(p) = \frac{\hbar-|p|}{\hbar^2}, \forall -\hbar\leq p\leq\hbar$, we get
\begin{align}
\label{first-int}
 \Pb(\bar{H}|r_0)  &\overset{}{=} e^{-\frac{\pi\lambda_b r_0\sqrt{\pi } \bar{\nu}}{\hbar^2}
 \int_{-\hbar}^{\hbar}  (\hbar-|p|){}_1 F_1 \left(\frac{1}{2};0;\pi  p^2 \mu \right) \dd{p} }
 e^{-\frac{\pi\lambda_b \pi \mu \bar{\nu}^2}{\hbar^2} \int_{-\hbar}^{\hbar} (\hbar-|p|) \big(   
 \frac{1}{\pi  \mu }  -p^2 e^{\pi  \mu  p^2} 
 \Gamma \left(0,\pi  p^2 \mu \right)  
   \big)\dd{p}} 
 \end{align}
 \begin{align}
 %\\
 \label{mejer-exp}
&\overset{(c)}{=} e^{-\frac{\pi\lambda_b r_0 \sqrt{\pi }  \bar{\nu}}{\hbar^2}
 \Bigg(
    2
\frac{\pi  \hbar^2 \mu  G_{2,3}^{2,2}\Big(\hbar^2 \pi  \mu \big|
\begin{array}{c}
 \frac{1}{2},\frac{1}{2} \\
 0,1,-\frac{1}{2} \\
\end{array}
\Big)-G_{2,3}^{2,2}\Big(\hbar^2 \pi  \mu \big|
\begin{array}{c}
 1,\frac{3}{2} \\
 1,2,0 \\ 
\end{array}
\Big)}{\pi^2 \mu } 
     \Bigg) }									
     e^{-\pi\lambda_b \zeta(\mu,\hbar)  },			% \pi \mu \bar{\nu}^2 	[Eq. 3.194.1]
     \end{align}
where $(c)$ follows from solving the left integral of (\ref{first-int}) \cite[Section 7.8]{gradshteyn2014table}, and the substitution  
\begin{align}
\zeta(\mu,\hbar) &=  \frac{\pi \mu \bar{\nu}^2}{\hbar^2} \int_{-\hbar}^{\hbar} (\hbar-|p|) \big(   
 \frac{1}{\pi  \mu }  -p^2 e^{\pi  \mu  p^2} 
 \Gamma \left(0,\pi  p^2 \mu \right)  
   \big)\dd{p}
\nonumber \\
&=  \bar{\nu}^2   - \frac{\pi \mu  \bar{\nu}^2}{\hbar^2}  \int_{-\hbar}^{\hbar}(\hbar-|p|)  p^2 e^{\pi  \mu  p^2}  \Gamma \left(0,\pi  p^2 \mu \right)  \dd{p} 
\nonumber \\
\label{zeta-int}
&\overset{(d)}{=}  \bar{\nu}^2   - \frac{2\pi \mu  \bar{\nu}^2}{\hbar^2}  \int_{0}^{\hbar}(\hbar-p)  p^2 e^{\pi  \mu  p^2}  \Gamma \left(0,\pi  p^2 \mu \right)    \dd{p}. 
\end{align} 
where (d) follows from the symmetry of the integrand. From (\ref{mejer-exp}) and (\ref{zeta-int}), with the fact that $\Pb(H|r_0) = 1 -\Pb(\bar{H}|r_0)$, the proof is completed. 
%%%%%%%%%%%%%%%%%%%%%%%%%%%%%%%%%%%%%%%%%%%%%%%%%%%%%%%%%%%%%%%%%%%%%%%%%%%%%%%%%%%%%
\vspace{-0.5 cm}
\section{Proof of Proposition \ref{inter-CoMP-HO}}
\label{proof:inter-CoMP-HO}
Following the Buffon's needle approach for hexagonal cells \cite{tabassum2018mobility}, we have
\begin{align}
\label{E-of_N}
\Eb[N]=\frac{4\sqrt{3}}{3\pi l} \Eb[V_{\rho}] \Eb[T]  
=\frac{2}{\pi R_h} \Eb[V_{\rho}] \Eb[T], 
%=\frac{4\sqrt{\sqrt{3}}}{\pi R_c \sqrt{2\pi}} \Eb[V_{\rho}] \Eb[T],
\end{align}			% ={\rm arccos}(\frac{\varrho_n}{u_n}) 
where $\Eb[V_{\rho}]$ represents the average the horizontal velocity of the UAV-UE. Given the constant velocity assumption, $\Eb[V_{\rho}] = \bar{\nu} \Eb[{\rm cos}(\varphi_n)]$, where $\varphi_n= {\rm arccos}\big(\frac{\varrho_n}{\sqrt{\varrho_n^2 + (z_n-z_{n-1})^2}}\big)$. We hence have 
 \begin{align}							%\bar{V}_{\rho} &=
 \Eb[V_{\rho}] &=  \Eb_{\rho_n,Z_n,Z_{n-1}}\Big[\frac{\bar{\nu} \varrho_n}{\sqrt{\varrho_n^2 + (z_n-z_{n-1})^2}}\Big] 
%\nonumber \\
%& 
\overset{(a)}{=} \frac{\sqrt{\pi } \bar{\nu} }{2} 
\Eb_{Z_n,Z_{n-1}}
\Big[{}_1 F_1\big(\frac{1}{2};0;\pi  (z_n-z_{n-1})^2 \mu \big)\Big] 
\end{align}
where (a) follows from averaging over the RV $\rho_n$. By proceeding similar to Appendix \ref{proof:near-HO} to obtain $\Eb[V_{\rho}]$, the handover rate can be obtained from $H= \frac{\Eb[N]}{\Eb[T]}$. This  completes the proof.
\end{appendices}

\vspace{-0.6 cm}
\bibliographystyle{IEEEtran}
\bibliography{bibliography}

\begin{thebibliography}{10}
\providecommand{\url}[1]{#1}
\csname url@samestyle\endcsname
\providecommand{\newblock}{\relax}
\providecommand{\bibinfo}[2]{#2}
\providecommand{\BIBentrySTDinterwordspacing}{\spaceskip=0pt\relax}
\providecommand{\BIBentryALTinterwordstretchfactor}{4}
\providecommand{\BIBentryALTinterwordspacing}{\spaceskip=\fontdimen2\font plus
\BIBentryALTinterwordstretchfactor\fontdimen3\font minus
  \fontdimen4\font\relax}
\providecommand{\BIBforeignlanguage}[2]{{%
\expandafter\ifx\csname l@#1\endcsname\relax
\typeout{** WARNING: IEEEtran.bst: No hyphenation pattern has been}%
\typeout{** loaded for the language `#1'. Using the pattern for}%
\typeout{** the default language instead.}%
\else
\language=\csname l@#1\endcsname
\fi
#2}}
\providecommand{\BIBdecl}{\relax}
\BIBdecl

\bibitem{amer2020caching}
R.~Amer, W.~Saad, H.~ElSawy, M.~Butt, and N.~Marchetti, ``Caching to the sky:
  Performance analysis of cache-assisted {CoMP} for cellular-connected
  {UAVs},'' in \emph{Proc. of the IEEE Wireless Communications and Networking
  Conference ({WCNC})}, Marrakech, Morocco, April. 2019.

\bibitem{8473483}
M.~Vondra, M.~Ozger, D.~Schupke, and C.~Cavdar, ``Integration of satellite and
  aerial communications for heterogeneous flying vehicles,'' \emph{IEEE
  Network}, vol.~32, no.~5, pp. 62--69, September 2018.

\bibitem{saad2019vision}
W.~Saad, M.~Bennis, and M.~Chen, ``A vision of {6G} wireless systems:
  Applications, trends, technologies, and open research problems,'' \emph{IEEE
  Network, to appear}, 2019.

\bibitem{kishk2019capacity}
M.~A. Kishk, A.~Bader, and M.-S. Alouini, ``Capacity and coverage enhancement
  using long-endurance tethered airborne base stations,'' \emph{arXiv preprint
  arXiv:1906.11559}, 2019.

\bibitem{kishk20193}
------, ``On the {3-D} placement of airborne base stations using tethered uavs,''
  \emph{arXiv preprint arXiv:1907.04299}, 2019.

\bibitem{eldosouky2019drones}
A.~Eldosouky, A.~Ferdowsi, and W.~Saad, ``Drones in distress: A game-theoretic
  countermeasure for protecting uavs against gps spoofing,'' \emph{arXiv
  preprint arXiv:1904.11568}, 2019.

\bibitem{8660516}
M.~{Mozaffari}, W.~{Saad}, M.~{Bennis}, Y.~{Nam}, and M.~{Debbah}, ``A tutorial
  on {UAVs} for wireless networks: Applications, challenges, and open
  problems,'' \emph{IEEE Communications Surveys Tutorials}, pp. 1--1, 2019.

\bibitem{8533634}
M.~{Mozaffari}, A.~{Taleb Zadeh Kasgari}, W.~{Saad}, M.~{Bennis}, and
  M.~{Debbah}, ``Beyond {5G} with {UAVs}: Foundations of a {3D} wireless
  cellular network,'' \emph{IEEE Transactions on Wireless Communications},
  vol.~18, no.~1, pp. 357--372, Jan 2019.

\bibitem{8470897}
Y.~{Zeng}, J.~{Lyu}, and R.~{Zhang}, ``Cellular-connected {UAV}: Potential,
  challenges, and promising technologies,'' \emph{IEEE Wireless
  Communications}, vol.~26, no.~1, pp. 120--127, February 2019.

\bibitem{qualcomm2017unmanned}
L.~Qualcomm, ``Unmanned aircraft systems' trial report,'' 2017.

\bibitem{lin2018sky}
X.~Lin, V.~Yajnanarayana, S.~D. Muruganathan, S.~Gao, H.~Asplund, H.-L.
  Maattanen, M.~Bergstrom, S.~Euler, and Y.-P.~E. Wang, ``The sky is not the
  limit: {LTE} for unmanned aerial vehicles,'' \emph{IEEE Communications
  Magazine}, vol.~56, no.~4, pp. 204--210, April 2018.

\bibitem{van2016lte}
B.~Van~der Bergh, A.~Chiumento, and S.~Pollin, ``{LTE} in the sky: trading off
  propagation benefits with interference costs for aerial nodes,'' \emph{IEEE
  Communications Magazine}, vol.~54, no.~5, pp. 44--50, May 2016.

\bibitem{azari2017coexistence}
M.~M. Azari, F.~Rosas, A.~Chiumento, and S.~Pollin, ``Coexistence of
  terrestrial and aerial users in cellular networks,'' in \emph{Proc. of IEEE
  Globecom Workshops (GC Wkshps)}, Singapore, Dec 2017, pp. 1--6.

\bibitem{8756296}
R.~{Amer}, W.~{Saad}, and N.~{Marchetti}, ``Towards a connected sky:
  Performance of beamforming with down-tilted antennas for ground and {UAV}
  user co-existence,'' \emph{IEEE Communications Letters}, pp. 1--1, 2019.

\bibitem{d2019cell}
C.~D'Andrea, A.~Garcia-Rodriguez, G.~Geraci, L.~G. Giordano, and S.~Buzzi,
  ``Cell-free massive {MIMO} for {UAV} communications,'' \emph{arXiv preprint
  arXiv:1902.03578}, 2019.

\bibitem{rahmati2019energy}
A.~Rahmati, Y.~Yap{\i}c{\i}, N.~Rupasinghe, I.~Guvenc, H.~Dai, and A.~Bhuyany,
  ``Energy efficiency of {RSMA} and {NOMA} in cellular-connected mmwave {UAV}
  networks,'' \emph{arXiv preprint arXiv:1902.04721}, 2019.

\bibitem{cherif2019downlink}
N.~Cherif, M.~Alzenad, H.~Yanikomeroglu, and A.~Yongacoglu, ``Downlink coverage
  and rate analysis of an aerial user in integrated aerial and terrestrial
  networks,'' \emph{arXiv preprint arXiv:1905.11934}, 2019.

\bibitem{liu2018multi}
L.~Liu, S.~Zhang, and R.~Zhang, ``Multi-beam {UAV} communication in cellular
  uplink: Cooperative interference cancellation and sum-rate maximization,''
  \emph{arXiv preprint arXiv:1808.00189}, 2018.

\bibitem{85317111}
S.~{Zhang}, Y.~{Zeng}, and R.~{Zhang}, ``Cellular-enabled {UAV} communication:
  A connectivity-constrained trajectory optimization perspective,'' \emph{IEEE
  Transactions on Communications}, vol.~67, no.~3, pp. 2580--2604, March 2019.

\bibitem{zhang2019trajectory}
S.~Zhang and R.~Zhang, ``Trajectory optimization for cellular-connected {UAV}
  under outage duration constraint,'' \emph{arXiv preprint arXiv:1901.04286},
  2019.

\bibitem{8654727}
U.~{Challita}, W.~{Saad}, and C.~{Bettstetter}, ``Interference management for
  cellular-connected {UAVs}: A deep reinforcement learning approach,''
  \emph{IEEE Transactions on Wireless Communications}, vol.~18, no.~4, pp.
  2125--2140, April 2019.

\bibitem{8421028}
P.~K. {Sharma} and D.~I. {Kim}, ``Coverage probability of {3-D} mobile {UAV}
  networks,'' \emph{IEEE Wireless Communications Letters}, vol.~8, no.~1, pp.
  97--100, Feb 2019.

\bibitem{8671460}
------, ``Random {3D} mobile {UAV} networks: Mobility modeling and coverage
  probability,'' \emph{IEEE Transactions on Wireless Communications}, vol.~18,
  no.~5, pp. 2527--2538, May 2019.

\bibitem{8681266}
S.~{Enayati}, H.~{Saeedi}, H.~{Pishro-Nik}, and H.~{Yanikomeroglu}, ``Moving
  aerial base station networks: A stochastic geometry analysis and design
  perspective,'' \emph{IEEE Transactions on Wireless Communications}, vol.~18,
  no.~6, pp. 2977--2988, June 2019.

\bibitem{8692749}
M.~M. {Azari}, F.~{Rosas}, and S.~{Pollin}, ``Cellular connectivity for {UAVs}:
  Network modeling, performance analysis and design guidelines,'' \emph{IEEE
  Transactions on Wireless Communications}, pp. 1--1, 2019.

\bibitem{abs-1804-04523}
\BIBentryALTinterwordspacing
S.~Euler, H.~Maattanen, X.~Lin, Z.~Zou, M.~Bergstr{\"{o}}m, and J.~Sedin,
  ``Mobility support for cellular connected unmanned aerial vehicles:
  Performance and analysis,'' 2018. [Online]. Available:
  \url{http://arxiv.org/abs/1804.04523}
\BIBentrySTDinterwordspacing

\bibitem{HCC:3325421.3329770}
A.~Fakhreddine, C.~Bettstetter, S.~Hayat, R.~Muzaffar, and D.~Emini, ``Handover
  challenges for cellular-connected drones,'' in \emph{Proc. of ACM Workshop on
  Micro Aerial Vehicle Networks, Systems, and Applications}, NY, USA, 2019, pp.
  9--14.

\bibitem{ding2016performance}
M.~Ding, P.~Wang, D.~L{\'o}pez-P{\'e}rez, G.~Mao, and Z.~Lin, ``Performance
  impact of {LoS} and {NLoS} transmissions in dense cellular networks,''
  \emph{IEEE Transactions on Wireless Communications}, vol.~15, no.~3, pp.
  2365--2380, March 2016.

\bibitem{8713514}
B.~{Galkin}, J.~{Kibilda}, and L.~{Da Silva}, ``A stochastic model for {UAV}
  networks positioned above demand hotspots in urban environments,'' \emph{IEEE
  Transactions on Vehicular Technology}, pp. 1--1, 2019.

\bibitem{austin2011unmanned}
R.~Austin, \emph{Unmanned aircraft systems: {UAVS} design, development and
  deployment}.\hskip 1em plus 0.5em minus 0.4em\relax John Wiley \& Sons, 2011,
  vol.~54.

\bibitem{7880694}
Z.~{Chen}, J.~{Lee}, T.~Q.~S. {Quek}, and M.~{Kountouris}, ``Cooperative
  caching and transmission design in cluster-centric small cell networks,''
  \emph{IEEE Transactions on Wireless Communications}, vol.~16, no.~5, pp.
  3401--3415, May 2017.

\bibitem{haenggi2012stochastic}
M.~Haenggi, \emph{\text{Stochastic geometry for wireless networks}}.\hskip 1em
  plus 0.5em minus 0.4em\relax \textit{Cambridge University Press}, 2012.

\bibitem{heath2011multiuser}
R.~W. Heath~Jr, T.~Wu, Y.~H. Kwon, and A.~C. Soong, ``Multiuser {MIMO} in
  distributed antenna systems with out-of-cell interference,'' \emph{IEEE
  Transactions on Signal Processing}, vol.~59, no.~10, pp. 4885--4899, Oct
  2011.

\bibitem{8490204}
X.~Yu, C.~Li, J.~Zhang, M.~Haenggi, and K.~B. Letaief, ``A unified framework
  for the tractable analysis of multi-antenna wireless networks,'' \emph{IEEE
  Transactions on Wireless Communications}, vol.~17, no.~12, pp. 7965--7980,
  Dec 2018.

\bibitem{1233531}
C.~{Bettstetter}, G.~{Resta}, and P.~{Santi}, ``The node distribution of the
  random waypoint mobility model for wireless ad hoc networks,'' \emph{IEEE
  Transactions on Mobile Computing}, vol.~2, no.~3, pp. 257--269, July 2003.

\bibitem{bettstetter2004stochastic}
C.~Bettstetter, H.~Hartenstein, and X.~P{\'e}rez-Costa, ``Stochastic properties
  of the random waypoint mobility model,'' \emph{Wireless Networks}, vol.~10,
  no.~5, pp. 555--567, 2004.

\bibitem{1624340}
E.~{Hyytia}, P.~{Lassila}, and J.~{Virtamo}, ``Spatial node distribution of the
  random waypoint mobility model with applications,'' \emph{IEEE Transactions
  on Mobile Computing}, vol.~5, no.~6, pp. 680--694, June 2006.

\bibitem{6477064}
X.~{Lin}, R.~K. {Ganti}, P.~J. {Fleming}, and J.~G. {Andrews}, ``Towards
  understanding the fundamentals of mobility in cellular networks,'' \emph{IEEE
  Transactions on Wireless Communications}, vol.~12, no.~4, pp. 1686--1698,
  April 2013.

\bibitem{8048668}
X.~{Xu}, Z.~{Sun}, X.~{Dai}, T.~{Svensson}, and X.~{Tao}, ``Modeling and
  analyzing the cross-tier handover in heterogeneous networks,'' \emph{IEEE
  Transactions on Wireless Communications}, vol.~16, no.~12, pp. 7859--7869,
  Dec 2017.

\bibitem{tabassum2018mobility}
H.~Tabassum, M.~Salehi, and E.~Hossain, ``Mobility-aware analysis of {5G} and
  {B5G} cellular networks: A tutorial,'' \emph{arXiv preprint
  arXiv:1805.02719}, 2018.

\bibitem{7006787}
S.~{Sadr} and R.~S. {Adve}, ``Handoff rate and coverage analysis in multi-tier
  heterogeneous networks,'' \emph{IEEE Transactions on Wireless
  Communications}, vol.~14, no.~5, pp. 2626--2638, May 2015.

\bibitem{7827020}
R.~{Arshad}, H.~{ElSawy}, S.~{Sorour}, T.~Y. {Al-Naffouri}, and M.~{Alouini},
  ``Velocity-aware handover management in two-tier cellular networks,''
  \emph{IEEE Transactions on Wireless Communications}, vol.~16, no.~3, pp.
  1851--1867, March 2017.

\bibitem{6042301}
J.~G. {Andrews}, F.~{Baccelli}, and R.~K. {Ganti}, ``A tractable approach to
  coverage and rate in cellular networks,'' \emph{IEEE Transactions on
  Communications}, vol.~59, no.~11, pp. 3122--3134, November 2011.

\bibitem{gradshteyn2014table}
I.~S. Gradshteyn and I.~M. Ryzhik, \emph{Table of integrals, series, and
  products}.\hskip 1em plus 0.5em minus 0.4em\relax Academic press, 2014.

\end{thebibliography}
\end{document}